\documentclass[11pt,draftclsnofoot,journal,letterpaper, onecolumn]{IEEEtran}

\ifCLASSOPTIONcompsoc
  \usepackage[nocompress]{cite}
\else
  \usepackage{cite}
\fi

\ifCLASSINFOpdf
\else
\fi

\usepackage{bm}
\usepackage{url}
\usepackage{cite}
\usepackage{xurl}
\usepackage{array}
\usepackage{xfrac}
\usepackage{amsthm}
\usepackage{amsmath,amssymb,amsfonts}

\usepackage{cases}
\usepackage{xcolor}
\usepackage{comment}
\usepackage{nicefrac}
\usepackage{graphicx}
\usepackage{textcomp}
\usepackage{booktabs}
\usepackage{multirow}
\usepackage{makecell}
\usepackage{subfigure}
\usepackage{autobreak}
\usepackage{algorithmic}
\usepackage{threeparttable}
\usepackage{nicefrac, xfrac}
\usepackage{color,xcolor,colortbl}
\usepackage{tcolorbox}

\usepackage{tikz}
\usetikzlibrary{decorations.pathreplacing}
\usetikzlibrary{tikzmark,calc,matrix,positioning,arrows.meta}

\newtheorem{lemma}{Lemma}
\newtheorem{remark}{Remark}

\newtheorem{corollary}{Corollary}
\newtheorem{definition}{Definition}
\newtheorem{proposition}{Proposition}
\newtheorem{keyquestion}{Question}

\theoremstyle{definition}

\usepackage{cite}
\usepackage{hyperref}
\hypersetup{
    colorlinks,
    linkcolor={red!70!black},
    citecolor={blue!70!black},
    urlcolor={blue!70!black},
    hypertexnames=false
}
\usepackage{bookmark}

\begin{document}
\title{Learning From Social Interactions: Personalized Pricing and Buyer Manipulation}

\author{Qinqi~Lin,
        Lingjie~Duan,~\IEEEmembership{Senior Member,~IEEE,}
        and~Jianwei~Huang,~\IEEEmembership{Fellow,~IEEE}
\IEEEcompsocitemizethanks{\IEEEcompsocthanksitem This paper was published in IEEE Transactions on Mobile Computing, December 2024 [\href{https://ieeexplore.ieee.org/document/10554984}{DOI: 10.1109/TMC.2024.3411111}]. Part of the results appeared in the WiOpt 2022 Conference~\cite{lin2022personalized}.

Qinqi Lin is with the School of Science and Engineering, Shenzhen Institute of Artificial Intelligence and Robotics for Society, The Chinese University of Hong Kong, Shenzhen, Shenzhen 518172, China (e-mail: qinqilin@link.cuhk.edu.cn).

Lingjie Duan was with the Engineering Systems and Design Pillar, Singapore University of Technology and Design, Singapore 487372 (e-mail: lingjie\_duan@sutd.edu.sg).  

Jianwei Huang is with the School of Science and Engineering, Shenzhen Institute of Artificial Intelligence and Robotics~for~Society, The Chinese University of Hong Kong, Shenzhen, Shenzhen 518172, China (e-mail: jianweihuang@cuhk.edu.cn). 
}
}


\IEEEtitleabstractindextext{%
\vspace{-45pt}
\begin{abstract}
As the sociological theory of homophily suggests, people tend to interact with those of~similar~preferences. Motivated by this well-established phenomenon, today's online sellers, such as Amazon,~seek~to learn a new buyer's private preference from his friends' purchase records. Although such learning~allows the seller to enable personalized pricing and boost revenue, buyers are also increasingly aware~of~these practices and may alter their social behaviors accordingly. This paper presents the first study regarding how buyers strategically manipulate their social interaction signals considering their preference correlations, and how a seller can take buyers' strategic social behaviors into consideration~when~designing~the pricing scheme. Starting with the fundamental two-buyer network, we propose and analyze a parsimonious model that uniquely captures the double-layered information asymmetry between the~seller~and~buyers, integrating both individual buyer information and inter-buyer correlation information. Our analysis reveals that only high-preference buyers tend to manipulate their social interactions to evade the seller's personalized pricing, but surprisingly, their payoffs may actually worsen as a result. Moreover, we demonstrate that the seller can considerably benefit from the learning practice, regardless of whether the buyers are aware of this fact or not. Indeed, our analysis reveals that buyers' learning-aware strategic manipulation has only a slight impact on the seller's revenue. In light of the tightening regulatory policies concerning data access, it is advisable for sellers to maintain transparency with buyers regarding their access~to~buyers' social interaction data for learning purposes. This finding aligns well with current~informed-consent~industry practices for data sharing. Finally, we explore the seller's dynamic learning process across~multiple~interconnected buyers, and show that learning previous buyers' preferences may not necessarily help infer other buyers' preferences in the seller's subsequent learning phase.
\end{abstract}

}

\maketitle

\IEEEdisplaynontitleabstractindextext
\IEEEpeerreviewmaketitle

\section{Introduction}\label{sec:introduction}

With the ever-increasing penetration of online social media (e.g., WeChat, Facebook, and Twitter),~people today can freely interact with one another online, exchanging product views or sharing user experiences \cite{jain2018effect}. For example, there were surprisingly 1.2 billion engagements on Facebook during~the~\mbox{four-day}~Electronic Entertainment Experience 2021, where people shared likes, comments, or posts towards the newly released gaming products \cite{batchelor_2021}. These online social interactions generate a wealth of data that~unveils~valuable information about people's preferences and even the underlying correlations behind these preferences.

The old saying, ``Birds of a feather flock together,'' encapsulates the tendency of individuals to socially interact with others who share similar preferences. This established phenomenon of homophily and its underlying sociological theory \cite{mcpherson2001birds} empower today's online sellers to effectively link buyers' mutual interactions in the online social network to their preference correlations. Indeed, over the recent years, researchers have been exploring various methodologies to infer customers' similarities or differences from their social interaction data to inform the sellers' decision-making \cite{braun2011scalable}. For example, Amazon is known to acquire and utilize buyers' Facebook social network data to train its personalization features \cite{bustos_2011_facebook,schulze_2018_facebook}. In this context, Amazon leverages the preferences of a buyer's friends, with whom the buyer interacts on Facebook, to deliver targeted purchase recommendations \cite{parfeni_2010_amazon}. Then, Amazon can learn a new buyer's private preferences from correlated purchase records of his friends in the past. Such learning paves the way for sellers' personalized pricing that tailors the price to the target buyer's product preference in future sales.\footnote{It is worth noting that price discrimination is commonly practiced in various consumer markets through personalized coupons, discounts, and fees \cite{dube2017scalable}. In Appendix \ref{Appendix:Evidence}, we provide in-depth discussions on personalized pricing and its integration with social network data in various consumer markets (such as product steering) and credit markets.} Therefore, the correlation information distilled from buyers' social interaction data can be of great help in boosting the revenue of today's online sellers.

However, as sellers increasingly utilize social interaction data for learning, there is a rising concern among buyers regarding the public exposure or usage of their social network data without their explicit consent or awareness. For example, Amazon and Netflix were reported to have undisclosed data-sharing deals with Facebook, which grants the former parties privileged access to users' social interaction information~\cite{x_2018}. Given such concerns, regulators have recently required online platforms to inform people of their data sharing with third parties and the corresponding purpose of data usage (e.g., in the informed consent policies in \cite{GDPR,a2023_california}). Following this, Amazon is now expected to warn buyers of its access to their social interaction data when buyers log in and connect to their Facebook accounts~\cite{-_by_-_boulton_2021}. 

Once aware of the seller's data access and the possibility of being charged personalized prices, buyers also have the incentive to strategically manipulate their product-related online behaviors or data to mislead the seller's learning and obtain better offerings \cite{li2023beating}. For example, people are found to provide untruthful data or misrepresent themselves on Facebook and Twitter to hide their actual information \cite{sannon2018privacy}. Buyers sharing similar high preferences for a new product may intentionally avoid discussing it or the related products on social media before the selling season, aiming to confuse the seller on their preference correlation. By doing~so, they mimic buyer pairs with different preferences and may enjoy a lower price than their friends. We are thus well motivated to understand the following question:

\begin{keyquestion}
Anticipating the seller's learning practice, how should buyers optimally manipulate their social interaction data, and can they benefit from such social manipulations?
\end{keyquestion}

Buyers use their social interactions as a means of \textit{jointly signaling} their preference correlation with each other to the seller. In signaling parlance, buyers are signal senders, and the seller is the receiver. Our key difference from the traditional signaling paradigm \cite{spence1978job} is that what senders signal is the correlation between their private information instead of an individual's private information alone (e.g., \cite{drakopoulos2021persuading,conitzer2012hide}). Another difference lies in the coupling between the senders in the interdependent social network, where one buyer's attempt to manipulate his social interaction data to confuse the seller also affects his interacting friends' manipulation incentives and social decision-making.

On the other hand, buyers' manipulations could interfere with the seller's learning of the true homophily information of their private preferences. This raises the question of how the seller should respond to such buyers' social manipulations that may impede her preference learning and personalized pricing. Actually, if the seller disregards current privacy regulations (at the risk of potential penalties) and does not inform buyers of her access to social interaction data, buyers may not know the seller's learning practice and thus never manipulate. This motivates us to ask the following question:

\begin{keyquestion}
How should the seller strategically learn from manipulated buyer data and redesign the optimal pricing scheme, and does concealing this practice from buyers result in sufficient revenue gain?
\end{keyquestion}

Previous research (e.g., \cite{bimpikis2021data,ali2020voluntary,montes2019value,belleflamme2016monopoly,conitzer2012hide,acquisti2005conditioning}) has focused on how a seller can infer a single buyer's private preference from his own behaviors or data, such as previous transactions \cite{bimpikis2021data}, purchase records \cite{conitzer2012hide,acquisti2005conditioning}, or direct disclosure \cite{ali2020voluntary}. These studies overlook the potential of leveraging today's fast-growing online social networks to understand the correlation between buyers' private preferences. Unlike these works, we investigate the seller's opportunity to learn this correlation from buyers' social interaction data. However, a challenge arises as the seller cannot directly infer a buyer's private preference until she knows both the buyer's preference correlation with his friends and his friends' purchase records. Furthermore, buyers may manipulate their social interactions, posing another challenge for the seller to discern the true correlation from manipulated data.

We model the above interactions as a dynamic Bayesian game between the seller and buyers. Starting with the basic two-buyer network, we propose and analyze a parsimonious model to capture the unique \textit{double-layered information asymmetry} existing between the seller and buyers, integrating both individual buyer and inter-buyer correlation information. In addition, we introduce two benchmark cases: the \textit{no-learning benchmark}, where the seller cannot access buyers' social interaction data, and the \textit{undisclosed-learning benchmark}, where the seller accesses such data without buyers' awareness and manipulation. Through comparisons, we analyze the effects of buyer manipulation and the seller's learning practice to provide valuable insights. Finally, we investigate the seller's dynamic learning from multiple interconnected buyers, which can be particularly challenging due to buyer manipulation. To address this challenge, we incorporate the seller's knowledge of prior buyers' preferences into her subsequent learning process, and examine the effectiveness of the dynamic learning in this context.

We summarize our main contributions below.

\begin{itemize}
    \item \emph{New personalized pricing via learning from buyers' social interactions:} To the best of our knowledge, this is the first analytical study on how buyers proactively manipulate their social interaction signals, and how a seller optimally takes buyers' strategic social behaviors into consideration when designing the pricing scheme. Our study helps understand the growing use of social network data for user profiling and the implications of informed-consent industry practice.
    
    \item \emph{Closed-form perfect Bayesian equilibrium (PBE) analysis under double-layered information asymmetry:} Our model uniquely incorporates a double-layered information asymmetry between the seller and buyers, integrating both individual buyer and inter-buyer correlation information. This poses great challenges to our game analysis, on top of which we also need to ensure information consistency from correlated buyers to the seller. To resolve these challenges, we start with an efficient reduction of the large equilibrium space to facilitate a forward game analysis. Then, we alternate it with backward induction to achieve a complete closed-form characterization of a unique PBE.
    
    \item \emph{Impact of buyer manipulation:} Our PBE analysis shows that the seller does not always trust the buyers' manipulated social interaction data to exercise pure personalized pricing. Instead, the seller may randomize personalized pricing with uniform pricing, which in turn mitigates buyers' incentives to manipulate and thus facilitates the seller's learning. Surprisingly, our analysis reveals that buyers could be worse off after strategically manipulating their social interactions.
    
    \item \emph{Effect of the seller's transparency about learning process:} Our analysis shows that the seller substantially benefits from the learning process, irrespective of buyer awareness. Considering the tightened privacy regulations and potential penalties, it is advisable for the seller to adopt an informed-consent practice, disclosing her use of buyers' social interaction data for learning.
    
    \item \emph{Effectiveness of the seller's dynamic learning:} When the seller learns across multiple interconnected buyers over time, we show that learning previous buyers' preferences may not necessarily help infer the other buyers' preferences in the seller's subsequent learning. Furthermore, we identify the sufficient conditions under which this conclusion holds in general multi-buyer networks.
\end{itemize}

The rest of the paper is organized as follows. Section~\ref{relatedwork} reviews the related work, and we present the system model in Section \ref{systemmodel}. We introduce the no-learning and the undisclosed-learning benchmarks in Section \ref{benchmark}. To derive the PBE, we begin by analyzing the seller's strategic learning to reduce the equilibrium space in Section \ref{PBE1}. We then conclude the PBE analysis in Section \ref{PBE2}. Section \ref{analysis} further investigates the welfare impact of buyer manipulation. Section \ref{sec:multi} extends our analysis to explore the seller's dynamic learning process over multiple interconnected buyers. In Section \ref{Numerical}, we apply our proposed pricing mechanism using real-world data in a numerical study. The robustness of our major insights is demonstrated in Section \ref{extensions} by relaxing the assumptions of a uniform prior preference distribution and binary interaction frequencies. Finally, Section \ref{conclusion} concludes this paper.

\section{Related Work}\label{relatedwork}

Our work is related to two main streams of literature:~(i)~information sharing and privacy in social networks, and (ii) price discrimination with buyer recognition.

\subsection{Information Sharing and Privacy in Social Networks}
Recent theoretical studies have shed light on individuals' information-sharing behaviors in social networks, highlighting the impact of potential privacy leakage on their decision-making (e.g., \cite{ding2020multi,olteanu2019co,rajtmajer2017ultimatum,gradwohl2017information}). For instance, Ding et al. in \cite{ding2020multi} analyzed the multi-party privacy conflict (MPC) in online social networks, where the private information of one user is disclosed by others who \mbox{co-own} the data. Olteanu et al. in \cite{olteanu2019co} also investigated such interdependent privacy and focused on the co-location information shared through social interactions. In addition to informational interdependencies, Gradwohl in \cite{gradwohl2017information} further considered network effects in social interactions and studied the impact of privacy enhancements on information sharing. These prior works focused on the social interactions among users without considering any seller's engagement or countermeasure. By contrast, our work explicitly considers the seller's learning from users' social interactions to infer private information and enable personalized pricing, which in turn affects the users' incentives and social interactions in the first place.

\subsection{Price Discrimination with Buyer Recognition}

On the seller's side, there is growing literature on price~discrimination with buyer recognition, where the seller learns buyers' preferences from their online behaviors or data (e.g., \cite{chen2020competitive,montes2019value,belleflamme2016monopoly,conitzer2012hide,acquisti2005conditioning}). For example, Conitzer et al. in \cite{conitzer2012hide} and Acquisti et al. in \cite{acquisti2005conditioning} investigated the scenario where the seller conditions pricing on a single buyer's purchase records in his past purchases. Particularly, Conitzer et al. in \cite{conitzer2012hide} allowed a buyer to hide his own past purchase records to thwart the seller's personalized pricing. Bellflamme et al. in \cite{belleflamme2016monopoly} further allowed a buyer to exploit hiding technologies to conceal his private valuation with a cost when the seller attempts to track his private information through profiling technologies. Montes et al. in \cite{montes2019value} and Chen et al. in~\cite{chen2020competitive} explored the duopoly setting, where sellers compete to target strategic buyers and personalize price offerings. However, the aforementioned studies did not incorporate the correlation among buyers' preferences into the learning process, which could significantly enhance the seller's ability to discriminate prices.

\section{System Model}\label{systemmodel}

\begin{figure*}[h]
\centering
\includegraphics[width=\linewidth]{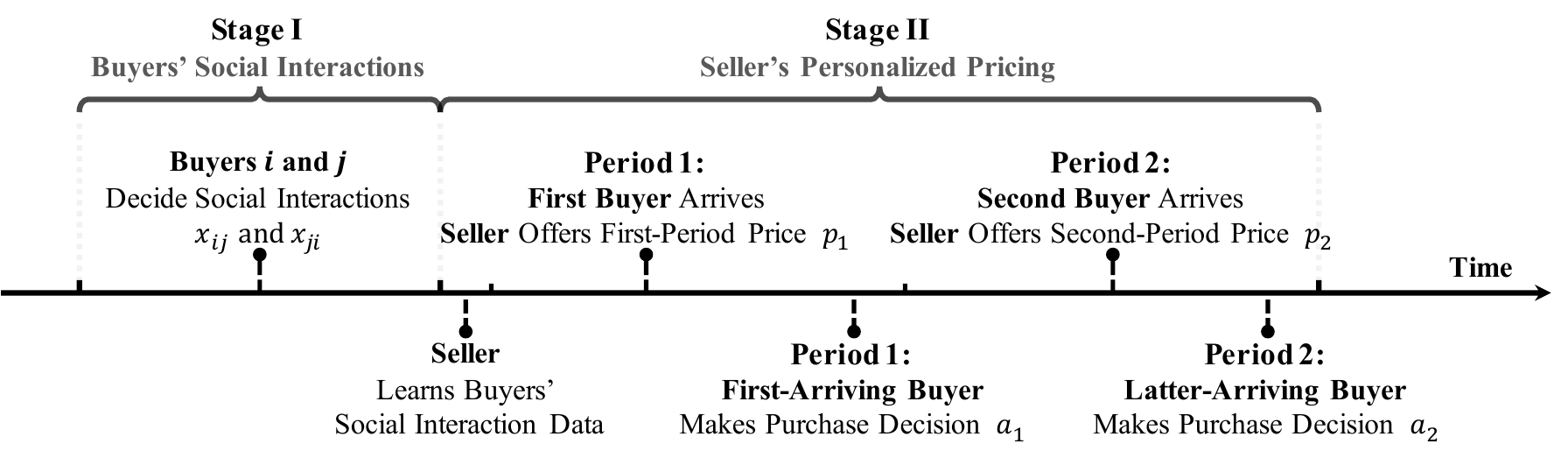}
\caption{Timeline of the dynamic Bayesian game: In Stage I, two buyers $i$ and $j$ socially interact to discuss a new product in the online social network. In Stage II, these buyers sequentially arrive and decide whether to purchase the product or not. The product seller learns the social interaction data from Stage I and then sequentially prices to each arriving buyer in the two consecutive selling periods of Stage II.}
\label{Timeline}
\end{figure*}

We begin with a stylized yet fundamental model between~a new product seller and any two socially connected buyers.\footnote{For convenience, we hereafter use female pronouns to refer to the seller and male (or plural) pronouns to refer to each buyer (or buyers).} This model captures buyers' pairwise relationship for learning their preference correlation. In Section \ref{sec:multi}, we will further explore more than two buyers' mutual relationships for the seller's dynamic learning.

To illustrate our system model, let us consider a scenario in the online marketplace where Amazon serves as the seller and offers the newly released gaming product, Battlefield 2042 \cite{a2021_battlefield}. In this context, the seller (Amazon) has access to the product-related discussions between two Facebook friends, denoted as buyers $i$ and $j$, regarding the upcoming release of Battlefield 2042 on the social media platform. The seller's objective is to learn the preference correlation between these two buyers and personalize price offerings for the upcoming selling season.

As illustrated in Fig. \ref{Timeline}, we consider two stages: \textit{Stage~I (Buyers' Social Interactions)}, where the two buyers engage with each other in the online social network, and \textit{Stage~II (Seller's Personalized Pricing)}, comprised of two purchase~periods during which the buyers and the seller interact in the product market. Next, we model these two stages in Section \ref{model_stageI} and Section \ref{model_stageII}, respectively. Finally, in Section \ref{game}, we formally formulate our dynamic Bayesian game.

\subsection{Model of Buyers' Social Interactions in Stage I}\label{model_stageI}

In Stage I, the two buyers socially interact to discuss a new product or related products on social media. As illustrated in Fig. \ref{Timeline}, each buyer decides his \textit{social interaction frequency} to the other buyer, jointly considering his \textit{social interaction utility} in Stage I and his \textit{purchase surplus} in Stage II later. Buyers' social decision-making closely relates to their purchase preferences, as elaborated in the following.

\subsubsection{Buyer Preference} A buyer $i$ ($j$, respectively) has a \emph{purchase preference} $v_i$ ($v_j$, respectively) for the new product, either high $v_{\text {\tiny H}}$ or low~$v_{\text {\tiny L}}$. This indicates his maximum willingness to pay for the~product, with $v_{\text {\tiny L}}<v_{\text {\tiny H}}$. Without accessing the two buyers' social data, the seller believes that each buyer's preference is independently and identically distributed, being high $(v_{\text {\tiny H}})$ or low $(v_{\text {\tiny L}})$ with an equal prior probability $1/2$. Here, we focus on this binary preference model for ease of exposition, and we extend it similarly to a more involved case of \textit{continuous preference distributions} in Appendix \ref{Appendix:continuous} to demonstrate~the~robustness of our key insights.

\subsubsection{Social Interaction Frequency} The two buyers \mbox{$i$ and $j$} simultaneously decide their \textit{social~interaction frequency} $x_{ij}\in\{0,1\}$ and $x_{ji}\in\{0,1\}$ with each other. Here, $x_{ij}$ indicates how often buyer $i$ interacts with buyer $j$ (e.g., through messages, comments, and sharing). Specifically, $x_{ij}=1$ implies buyer $i$ interacts with buyer~$j$ frequently to discuss the product on social media, whereas $x_{ij}=0$ means buyer $i$ seldom reaches out to buyer $j$. 

\subsubsection{Social Interaction Utility} As the sociological theory of homophily suggests, people tend to interact with those of similar preferences \cite{mcpherson2001birds}. Thus, we model the \emph{social interaction utility} $u_{i}(x_{ij},x_{ji})$ in Table~\ref{table} to quantify how much each buyer $i$ gains (or suffers) when interacting with the other. Table \ref{table} includes two sub-tables, Tables \ref{table_same} and \ref{table_differ}, depending on whether their preferences for the new product are the same or different. Within each of the pairs in the same round bracket, the first number represents the buyer $i$'s social interaction utility, and the second represents the buyer $j$'s, depending on their decisions ($x_{ij}$ and $x_{ji}$) in social interaction frequency.

\begin{table}[h]
\centering
\caption{Social Interaction Utility of Buyers $i$ and $j$}
\subtable[Same Preference $v_i=v_j$]{
\begin{tabular}{cp{2cm}p{2cm}}
\specialrule{\heavyrulewidth}{1.2pt}{0pt}
\multicolumn{1}{c}{\makecell[c]{}} &\makecell[c]{$x_{ji}=1$} &\makecell[c]{$x_{ji}=0$} \\
\specialrule{\lightrulewidth}{0pt}{1.2pt}
$x_{ij}=1$ &\cellcolor[gray]{0.8}{\makecell[c]{$(1,1)$}} &\cellcolor[gray]{0.8}{\makecell[c]{$(1-l,l)$}}\\
$x_{ij}=0$ &\cellcolor[gray]{0.8}{\makecell[c]{$(l,1-l)$}} &\cellcolor[gray]{0.8}{\makecell[c]{$(0,0)$}}\\
\specialrule{\heavyrulewidth}{1.2pt}{0pt}
\end{tabular}
\label{table_same}
}
\subtable[Different Preferences $v_i\neq v_j$]{
\begin{tabular}{cp{2cm}p{2cm}p{2cm}p{2cm}}
\specialrule{\heavyrulewidth}{1.2pt}{0pt}
\multicolumn{1}{c}{\makecell[c]{}} &\makecell[c]{$x_{ji}=1$} &\makecell[c]{$x_{ji}=0$}\\
\specialrule{\lightrulewidth}{0pt}{1.2pt}
$x_{ij}=1$ &\cellcolor[gray]{0.92}{\makecell[c]{$(-c,-c)$}} &\cellcolor[gray]{0.92}{\makecell[c]{$(-c+r,-r)$}}\\
$x_{ij}=0$ &\cellcolor[gray]{0.92}{\makecell[c]{$(-r,-c+r)$}} &\cellcolor[gray]{0.92}{\makecell[c]{$(0,0)$}}\\
\specialrule{\heavyrulewidth}{1.2pt}{0pt}
\end{tabular}
\label{table_differ}
}
\label{table}
\end{table}

Table \ref{table_same} tells that two buyers have the same preference \mbox{$v_i=v_j$}. First, we consider the combination case of \mbox{$x_{ij}=1$} and $x_{ji}=1$, leading to the social interaction~utilities $(1,1)$ for both buyers. Notice that we normalize such social interaction happiness as a unit $1$. In this case, when frequently interacting with the other buyer $j$ of the same preference, a buyer $i$ experiences social happiness through~gaining~empathy and reinforcing connections \cite{seiter2015secret}. Then we move on to the case of $x_{ij}=1$ and $x_{ji}=0$, where the social interaction utilities are $(1-l,l)$. Here, we let parameter \mbox{$l\in(0,1)$} denote the loss one buyer incurs due to the other buyer's low social response when both buyers have the same preference. In this case, if buyer $j$ seldom talks back to buyer $i$'s frequent social interaction $x_{ij}=1$, buyer $i$ would feel a lack of responses and gain less happiness $1-l$ \cite{krasnova2010online}. Alternatively, the other buyer $j$, when facing buyer $i$'s frequent social interactions of $x_{ij}=1$, also gains less happiness $l$ as he remains silent with $x_{ji}=0$. Finally, for the case of $x_{ij}=0$ and $x_{ji}=0$, the social interaction utilities are $(0,0)$. Here, we normalize each buyer's social interaction utility from rare interactions as zero, no matter whether they share the same preference in Table \ref{table_same} or not in Table \ref{table_differ}.

Table \ref{table_differ} tells that two buyers have different preferences $v_i\neq v_j$. First we consider the case of $x_{ij}=1$ and $x_{ji}=1$, the social interaction utilities are $(-c,-c)$. When attempting to interact with the other buyer $j$ who differs in preference, a buyer $i$ experiences embarrassment with even a disutility $c$ to maintain such social interaction \cite{mcpherson2001birds}. Then we move on to the case of \mbox{$x_{ij}=1$ and $x_{ji}=0$}, where the social interaction utilities are \mbox{$(-c+r,-r)$}. In this case, the other buyer $j$'s low social responses instead relax buyer $i$'s awkwardness $c$ given different preferences to a certain extent \mbox{$r\in(0,c)$}. Yet, in turn, when receiving frequent messages with different opinions from buyer $i$, buyer $j$ gets upset with a disutility~$r$.

\subsection{Model of Personalized Pricing in Stage II} \label{model_stageII}

At the beginning of Stage II, the seller can observe the \emph{common interaction frequency} between buyer $i$ and $j$, denoted~as,
\begin{equation}\label{hat}
    \hat{x}\triangleq\min\{x_{ij},x_{ji}\},
\end{equation}
which is also binary in set $\{0,1\}$. The \emph{minimum} operation captures the mutual essence of social interactions. This can be a reasonable approximation of reality, as when processing interaction data in a large social network, the seller needs to locate the two buyers to identify the related data and would lose certain data if any buyer chooses not to interact with the other buyer. Consider an example involving Alice $i$ and Bob $j$. Alice posts a co-owned photo with Bob ($x_{ij}=1$), and the seller can identify their correlation after Bob comments on the post ($x_{ji}=1$). Some privacy protection mechanism even asks Bob to grant or deny publicizing such a post \cite{such2016resolving}. That is, Bob's choice of denial $x_{ji}=0$ disables the availability of the co-owned photo from Alice, i.e., $\min\{x_{ij}, x_{ji}\}=0$. Overall, the seller needs to learn from such observable mutual common data $\hat{x}$ in (\ref{hat}) to infer the preference correlation between the two buyers.

Next, the seller announces prices $p_1$ and $p_2$~\mbox{sequentially}, for the two sequentially arriving buyers for the new product, respectively (with one buyer in each period of Stage~II, see Fig. \ref{Timeline}). Notice that the arrival sequence is \textit{random} (i.e., buyer $i$ may arrive either after or before buyer $j$)~due~to information diffusion in marketing or randomness in individual behaviors. This sequential purchase pattern is commonly observed in practice and widely adopted in the literature on dynamic pricing (e.g., \cite{acquisti2005conditioning,conitzer2012hide}). We denote the first-arriving buyer's binary purchase decision as $a_1$: $a_1=1$ if he decides to purchase (i.e., buyer's preference is no less than the offered price $v_1\ge p_1$), and $a_1=0$ otherwise. Similarly, we define $a_2$ for the latter-arriving buyer's purchase decision.  

Given incomplete information of buyers' preferences, the seller prices in the consecutive selling periods in Stage II to maximize the \textit{expected sale revenue} $\tilde{\Pi}$, by taking the expectation over all possible arriving buyers in the market,~i.e.,
\begin{equation}\label{expect_revenue}
    \tilde{\Pi}(p_1,p_2)\triangleq\mathbb{E}_{v_1,v_2}\{p_1a_1+p_2a_2\}.
\end{equation}
On the other hand, the \textit{final payoff} of each buyer consists of the purchase surplus in Stage II and the social interaction utility in Stage I (see Table \ref{table}). To illustrate, if buyer $i$ arrives in the selling period $t\in\{1,2\}$ of Stage II (hence with $v_t=v_i$), then the final payoff $\pi_i$ of buyer $i$ is: 
\begin{equation}\label{payoff1}
    \pi_i(x_{ij},x_{ji};v_t=v_i)\triangleq \max\{v_t-p_t,0\}+u_i(x_{ij},x_{ji}).
\end{equation}
Yet when deciding the social interaction frequency in Stage I, buyer $i$ is not sure whether he arrives earlier than buyer $j$ or not in the following Stage II. Thus, anticipating the seller's pricing in Stage II, buyer $i$ makes the social decision $x_{ij}$ to maximize his \textit{expected total payoff} $\tilde{\pi}_i$ over all possible arrival sequences, i.e.,
\begin{equation}\label{expect_payoff}
\begin{aligned}
\tilde{\pi}_i(x_{ij},x_{ji})&\triangleq\mathbb{E}_{t\in\{1,2\}}\left[\pi_i(x_{ij},x_{ji};v_t=v_i)\right].
\end{aligned}
\end{equation}
Here, buyer $i$ may deviate from the interaction frequency determined purely according to Table~\ref{table}. We broadly define such deviation as buyer $i$'s \textit{manipulation} of his social interaction with buyer $j$, which is expected to disguise their (true) homophily information. We further use the term ``unilateral manipulation'' to refer to the situation in which one buyer independently alters the interaction frequency from the one set by Table \ref{table}, with no corresponding change from the other buyer.

\subsection{Dynamic Bayesian Game Formulation}\label{game}

We formally model the interactions among the seller and~the buyers as a two-stage dynamic Bayesian game as follows, with the decision-making timing illustrated in Fig. \ref{Timeline}. The solution concept we adopt is \textit{perfect Bayesian equilibrium~(PBE)}. To differentiate from the two benchmark models to be introduced later, we terminate the current one \textit{strategic-learning model}. 

\begin{itemize}
    \item \textbf{Stage I:} Two buyers $i$ and $j$ simultaneously decide their social interaction frequencies $x_{ij}$ and $x_{ji}$, with the goals of maximizing their individual expected total payoffs $\tilde{\pi}_i$ or $\tilde{\pi}_j$ in (\ref{expect_payoff}), respectively.\footnote{We model buyers as strategizers in response to the growing awareness of individuals' privacy concerns regarding their online social activities \cite{consumers}. In particular, we highlight that privacy concerns here stem from the potential economic consequences of information disclosure, specifically due to price discrimination.}
    
    \item \textbf{Stage II:} After accessing buyers' social interaction data, the seller decides and announces prices $p_1$ and $p_2$ sequentially for the arriving buyers to maximize her expected sale revenue in (\ref{expect_revenue}) with strategic learning. Notice that here, buyers in any selling period $t$ decide to purchase or not:
    \begin{equation}\label{purchasedecision}
        a_t^*(v_t,p_t)=\boldsymbol{1}(v_t\ge p_t),\quad\forall t\in\{1,2\},
    \end{equation}
    where $\boldsymbol{1}(\cdot)$ is an indicator function. 
\end{itemize}

Next, we explain the \emph{information structure} of this game. At the beginning of the game, both buyers' product preferences are private and only known to themselves;\footnote{Before deciding the social interaction frequencies on social media, the two buyers are aware of each other's preferences, and this knowledge arises from the buyers' prior acquaintance through historical or physical interactions.} the~seller only knows the prior distribution (e.g., through market surveys or previous transaction data on related products). After learning from buyers' mutual interactions in the social network in Stage I, the seller infers the correlation between buyers' preferences at the beginning of Stage II. Yet, she still does not know each individual's private preference. Only after offering the first-period price $p_1$ in Stage II, the seller may use the price to sample the first-arriving buyer's preference from his purchase decision. Together with the correlation learned from Stage~I, the seller can infer the latter-arriving buyer's preference. Overall, the seller's incomplete information about each buyer's preference gradually decreases along the timeline in Fig. \ref{Timeline}. Notably, our game incorporates a unique double-layered information asymmetry existing between the seller and buyers, integrating both the individual buyer information and inter-buyer correlation information.

Notice that the seller cannot charge a personalized price to the first-arriving buyer. This is because the seller only knows the \textit{correlation} of both buyers' preferences instead of this buyer's \textit{individual preference}. But for the latter-arriving buyer, the seller can combine the purchase record in the prior selling period to tell this buyer's preference. In this sense, only the second-period price $p_2$ can be personalized.

\section{Two Benchmarks}\label{benchmark}

Before we analyze the PBE of the strategic-learning model defined in Section \ref{game}, we first introduce two benchmarks~for later comparisons.

\subsection{No-learning Benchmark}\label{bench_nolearning}

In the first benchmark, we consider the case where the~seller cannot access buyers' social interaction data. In this case, Stages I and II of the dynamic game outlined in Section~\ref{game} are separated. Therefore, no buyer has the incentive to~manipulate, and the seller has to set a uniform price in both selling periods. Here, the solution concept of PBE degenerates to the subgame perfect equilibrium (SPE), as detailed in the following lemma.

\begin{lemma}\label{lemma:nolearning}
In the no-learning benchmark, the SPE of the~dynamic game is given as follows.
\begin{itemize}
    \item \mbox{\emph{In Stage I,}} if buyers $i$ and $j$ have the same preference \mbox{($v_i=v_j$)}, both choose the maximum social interaction frequency, i.e., \mbox{$x^*_{ij}=x^*_{ji}=1$}. If they have different preferences \mbox{($v_i\neq v_j$)}, they choose the minimum interaction frequency, i.e., \mbox{$x^*_{ij}=x^*_{ji}=0$}.
    
    \item \mbox{\emph{In Stage II,}} the seller charges the same uniform price in both selling periods:
    \begin{equation}p_1^*=p_2^*=\left\{
    \begin{aligned}
    v_{\text{\tiny H}}, &\quad\textrm{if $2v_{\text{\tiny L}} \le v_{\text{\tiny H}}$,}\\
    v_{\text{\tiny L}}, &\quad\textrm{otherwise.}
    \end{aligned}
    \right.
    \end{equation}
\end{itemize}
\end{lemma}

\subsection{Undisclosed-learning Benchmark}\label{bench_unaware}

Next, we introduce the second benchmark case, where the seller can ideally access buyers' social interaction data in Stage I in an undisclosed way. In other words, buyers are not aware of such access or learning by the seller. Therefore, buyers' social behaviors $(x_{ij}^*,x_{ji}^*)$ in Stage I remain the same as those in the no-learning benchmark. However, in Stage II, the seller can personalize the price offering for the arriving buyer in the second selling period.

\begin{lemma}\label{unaware}
In the undisclosed-learning benchmark, the SPE of~the dynamic game is given as follows.
\begin{itemize}
    \item \mbox{\emph{In Stage I,}} buyers' equilibrium social interaction frequencies are the same as those in the no-learning benchmark.
    
    \item \mbox{\emph{In Stage II,}} the seller's equilibrium pricing depends on the relation between $v_{\text{\tiny L}}$ and $v_{\text{\tiny H}}$ as follows.
    \begin{itemize}  
        \item [(i)] \mbox{\emph{\textbf{Large $\bm{v_{\text{\tiny L}}}$ regime} with $v_{\text{\tiny L}}\in[2v_{\text{\tiny H}}/3,v_{\text{\tiny H}})$:}}~The~seller charges low prices in both selling periods, i.e., $p_1^*=p_2^*=v_{\text{\tiny L}}$.
        \item [(ii)] \emph{\textbf{Small $\bm{v_{\text{\tiny L}}}$ regime} with $v_{\text{\tiny L}}\in(0,2v_{\text{\tiny H}}/3)$:} The seller charges a high price $p_1^*=v_{\text{\tiny H}}$ in the first selling period. In the second selling period, the seller offers a personalized price that depends on the buyers' common interaction frequency $\hat{x}^*=\min\{x_{ij}^*,x_{ji}^*\}$ in Stage I and the purchase record $a_1^*\in\{0,1\}$ of the first selling period in Stage II:
    \begin{numcases}{p_2^*=}
    v_{\text{\tiny H}}\boldsymbol{1}(a_1^*=1)+v_{\text{\tiny L}}\boldsymbol{1}(a_1^*=0),&\textrm{if $\ \hat{x}^\ast=1$,}\label{Unaware_p2_1}\\
    v_{\text{\tiny L}}\boldsymbol{1}(a_1^*=1)+v_{\text{\tiny H}}\boldsymbol{1}(a_1^*=0),&\textrm{if $\ \hat{x}^\ast=0$.}\label{Unaware_p2_2}
    \end{numcases}
    \end{itemize}
\end{itemize}
\end{lemma}

Since buyers do not manipulate their social interaction data in Stage I, the seller learns that the product preferences of the two buyers are positively correlated when observing \mbox{$\hat{x}^*=\min\{x_{ij}^*,x_{ji}^*\}=1$}. Further, if the first-arriving buyer purchases the product at $p_1^*=v_{\text{\tiny H}}$, the seller knows that his preference is $v_1=v_{\text{\tiny H}}$. Hence, the seller would charge \mbox{$p_2^*=v_{\text{\tiny H}}$} in \eqref{Unaware_p2_1} to the latter-arriving buyer, who has the same high preference $v_2=v_{\text{\tiny H}}$. However, if the first-arriving buyer does not purchase, the seller learns that his preference is $v_1=v_{\text{\tiny L}}$ and charges $p_2^*=v_{\text{\tiny L}}$ in \eqref{Unaware_p2_1} to the latter-arriving buyer who has the same low preference. Similar arguments apply when the seller observes $\hat{x}^*=0$.

By comparing the above two benchmarks, we have the following conclusion about the effect of the seller's undisclosed learning practice.

\begin{corollary}\label{salerevenue_un}
In the undisclosed-learning benchmark, the seller's expected revenue equals or surpasses the no-learning benchmark. This revenue increase occurs as long as $v_{\text{\tiny L}}<2v_{\text{\tiny H}}/3$, with the potential to attain a maximum gain of $25\%$.
\end{corollary}

\section{Efficient Reduction of Equilibrium Space}\label{PBE1}

We now proceed to our strategic-learning model, presented in Section \ref{game}, where buyers are informed about the seller's learning practice. We start with a streamlined equilibrium analysis of the seller's strategic learning from Stage I (Section \ref{learning1}) and pricing in Stage II (Section \ref{learning2}), significantly trimming the equilibrium space's feasible set. This process enables us to formulate a belief about buyers' manipulation structure in Stage I through forward analysis. This, in turn, facilitates a tractable analysis of PBE using backward induction in the following section.

Given the \textit{double-layered information asymmetry} between the seller and the two buyers, which incorporates both~individual buyer information and inter-buyer correlation information, there are a large number of possibilities of buyers'~private information behind the seller's observable social interaction data $\hat{x}$ in (\ref{hat}). This motivates us to first reduce the searching space of all possible strategy combinations to facilitate our forward analysis in this section for the seller's learning.

\subsection{Learning from Buyers' Social Interaction Data}\label{learning1}

We first guide the seller's strategic learning from buyers'~social interaction data in Stage I with the following lemma.

\begin{lemma}\label{lemma1_buyer}
In the PBE of the strategic-learning model, buyers $i$ and $j$ with the same low preference choose the maximum social interaction frequency in Stage I:
\begin{equation}
x_{ij}^*(v_i=v_j=v_{\text{\tiny L}})=x_{ji}^*(v_j=v_i=v_{\text{\tiny L}})=1.    
\end{equation}
Given different preferences, both buyers choose the minimum~social interaction frequency in Stage I: 
\begin{equation}
x_{ij}^*(v_j\neq v_i)=x_{ji}^*(v_i\neq v_j)=0.
\end{equation}
\end{lemma}

Intuitively, a low-preference buyer just affords a personalized price $p_2=v_{\text{\tiny L}}$ in Stage II. Likewise, when facing higher prices, he also receives a zero purchase surplus by deciding not to purchase the product. When interacting with another low-preference buyer, manipulating social interaction data from high to low interaction frequency does not improve this low-preference buyer's purchase surplus, but reduces his social interaction utility as in Table \ref{table_same}. However, when meeting a high-preference buyer $j$, a low-preference buyer $i$ loses from frequent interactions given their different preferences, as stated in Table~\ref{table_differ}. In this case, this low-preference buyer $i$ honestly chooses \mbox{$x_{ij}^*=0$}, as any manipulation to \mbox{$x_{ij}^*=1$} does not help improve his purchase surplus in Stage II. As a result, the common interaction frequency signal in \eqref{hat} to the seller is always \mbox{$\hat{x}=0$}. The high-preference buyer $j$ cannot manipulate their preference correlation unilaterally, and thus also chooses \mbox{$x_{ji}^*=0$} honestly given their different preferences.

After applying Lemma \ref{lemma1_buyer}, the remaining case to consider is \mbox{$v_j=v_i=v_{\text{\tiny H}}$}. In this case, we generally allow these high-preference buyers to employ a mixed strategy, i.e., choosing the low social interaction frequency $0$ with a manipulation probability while choosing the high interaction frequency $1$ with complementary probability. Formally, we define such \textit{manipulation strategies} for high-preference buyers below. Notice that as both high-preference buyers are symmetric, they would choose the same manipulation probability, and this is the unique equilibrium structure for them.

\begin{definition}[Manipulation strategy for a high-preference buyer]
When interacting with another high-preference buyer, a high-preference buyer chooses the low social interaction frequency $0$ with a manipulation probability $\rho$, i.e.,
\begin{equation}\label{rho}
\begin{aligned}
     \rho&\triangleq {\rm{Pr}}\left(x_{ij}(v_i=v_j=v_{\text {\tiny H}})=0 \right)\\
     &={\rm{Pr}}\left(x_{ji}(v_j=v_i=v_{\text {\tiny H}})=0 \right),
\end{aligned}
\end{equation}
whereas he chooses high interaction frequency $1$ with complementary probability $1-\rho$.
\end{definition}

To enable the seller's personalized pricing in Stage II,~we now analyze the seller's posterior belief of buyers' product preferences by learning from their common interaction frequency $\hat{x}$. According to the Bayes' theorem, we have for the seller:
\begin{equation}\label{belief1}
{\rm{Pr}}(v_i=v_j=v_{\text {\tiny H}}|\hat{x}=1)=\frac{{\rm{Pr}}(\hat{x}=1|v_i=v_j=v_{\text {\tiny H}})}{{\rm{Pr}}(\hat{x}=1|v_i=v_j=v_{\text {\tiny H}})+1},
\end{equation}
and
\begin{equation}\label{belief0}
{\rm{Pr}}(v_i=v_j=v_{\text {\tiny H}}|\hat{x}=0)=\frac{1-{\rm{Pr}}(\hat{x}=1|v_i=v_j=v_{\text {\tiny H}})}{3-{\rm{Pr}}(\hat{x}=1|v_i=v_j=v_{\text {\tiny H}})},
\end{equation}
which cannot be directly determined using backward induction from Stage II. Indeed, to ensure belief consistency over time, we need to examine the consistency between buyers' social decisions and the beliefs above, based on which the seller optimizes pricing schemes in Stage II (see Section \ref{PBE2}).

\subsection{Learning from First-arriving Buyer's Purchase}\label{learning2}

Next, we present another lemma that helps us narrow down the seller's pricing-decision space in Stage II. 

\begin{lemma}\label{pricebinary}
In Stage II of the strategic-learning model, the seller's equilibrium price $p_t^*$ in any selling period $t \in\{1,2\}$ is either $v_{\text{\tiny H}}$ or $v_{\text{\tiny L}}$.
\end{lemma}

Lemma \ref{pricebinary} helps the seller restrict attention to binary~pricing choices in each selling period of Stage II. Strategically, the seller may consider learning the first-arriving buyer's preference from his purchase decision by offering a high price \mbox{$p_1=v_{\text{\tiny H}}$}. In this way, the seller can facilitate personalized pricing for the latter-arriving buyer with the preference correlation learned from Stage I. On the other hand, when offering a low price $p_1=v_{\text{\tiny L}}$, the seller cannot infer the first-arriving buyer's preference. Despite this, the seller can still guarantee a revenue of $v_{\text{\tiny L}}$ from each arriving buyer. In this sense, the seller faces the tradeoff of whether to~enable personalized pricing in Stage II. In addition, the seller's pricing also needs to account for buyers' potential manipulations in Stage I.

\section{PBE Analysis with Buyer Awareness}\label{PBE2}

After screening the buyers' social behaviors in Stage I using Lemma \ref{lemma1_buyer} and the seller's strategic pricing in Stage II using Lemma \ref{pricebinary}, we now analyze the PBE of the strategic-learning model in this section. The PBE analysis is still involved as we need to ensure information consistency from correlated buyers to the seller over time. To achieve this, we alternate \textit{forward analysis} (see Section \ref{PBE1}) with \textit{backward induction} while ensuring consistency as follows (see more technical details in Appendix \ref{Appendix:PBE}).

\begin{itemize}
    \item \mbox{\emph{Step 1:}} We first apply backward induction to analyze the seller's optimal pricing schemes in Stage II, given the posterior belief on buyers' social decisions in \eqref{belief1} and \eqref{belief0}. The seller infers the latter-arriving buyer's preference according to the purchase behavior of the first-arriving buyer, combined with their preference correlation learned from Stage I.
    
    \item \mbox{\emph{Step 2:}} Based on the seller's optimal pricing derived in Step 1, we next analyze the buyers' related social behaviors in Stage I. The key idea is to ensure \textit{belief consistency} across stages: the buyers determine their social interaction frequencies while anticipating the seller's learning and pricing in Stage II; these equilibrium strategies should be consistent with the seller's belief about their preference correlation.
\end{itemize}

We fully characterize all PBE in closed form and classify the structural results with four preference regions in Fig. \ref{regions}.

\begin{figure}[h]
\begin{minipage}[t]{0.45\linewidth}
\vspace{0pt}
\includegraphics[width=\linewidth]{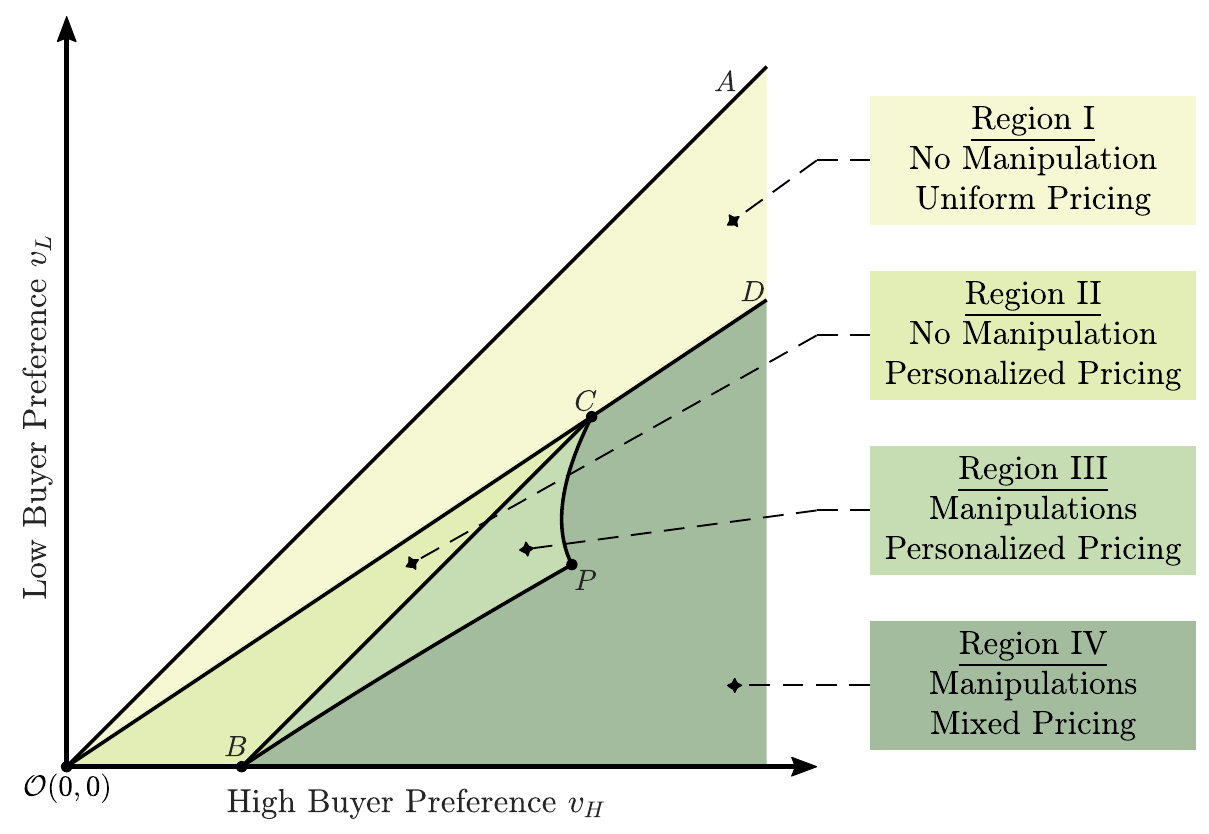}
\end{minipage}
\hspace{15pt}
\begin{minipage}[t]{0.45\linewidth}
\vspace{28pt}
\scriptsize
\begin{tabular}{cp{6cm}}
\specialrule{\heavyrulewidth}{1.2pt}{0pt}
\cellcolor[gray]{0.92}{\makecell[c]{Curve}} &{\makecell[c]{Corresponding Equation in $(v_\text{\tiny H},v_\text{\tiny L})$-Plane}}\\
\specialrule{\lightrulewidth}{0pt}{1.2pt}
OA &\cellcolor[gray]{0.92}{\makecell[c]{$v_\text{\tiny L}=v_\text{\tiny H}$}}\\
\cellcolor[gray]{0.92}{OCD} & {\makecell[c]{$v_\text{\tiny L}=2v_\text{\tiny H}/3$}}\\
BC &\cellcolor[gray]{0.92}{\makecell[c]{$v_\text{\tiny H}-v_\text{\tiny L}=2(1-l)$}}\\
\cellcolor[gray]{0.92}{CP} & {\makecell[c]{$(v_\text{\tiny H}-v_\text{\tiny L})v_\text{\tiny L}=8(1-l)^2$}}\\
BP & \cellcolor[gray]{0.92}{\makecell[c]{$(v_\text{\tiny H}-v_\text{\tiny L})(v_\text{\tiny H}-2v_\text{\tiny L})=4(1-l)^2$}}\\
\specialrule{\heavyrulewidth}{1.2pt}{0pt}
\end{tabular}
\end{minipage}
\caption{PBE in four regions to tell whether high-preference buyers manipulate their social interaction data with the other high-preference buyer, and how the seller prices in the PBE.}
\label{regions}
\end{figure}

\begin{proposition}\label{region1}
In the strategic-learning model, if low and high preferences are similar (see Region I of Fig. \ref{regions}), there exists a unique PBE as follows. 
\begin{itemize}
    \item \emph{In Stage I}, buyers $i$ and $j$ with the same high preference choose the maximum social interaction frequency:
\begin{equation}\label{region1_buyer}
x_{ij}^*(v_i=v_j=v_\text{\tiny H})=x_{ji}^*(v_j=v_i=v_\text{\tiny H})=1.
\end{equation}
    \item \emph{In Stage II}, the seller charges low prices $p_1^*=p_2^*=v_\text{\tiny L}$ in both selling periods.
\end{itemize}
\end{proposition}

Note that (\ref{region1_buyer}) for \mbox{$v_i=v_j=v_\text{\tiny H}$} together with Lemma~\ref{lemma1_buyer} (for the other preference distributions) completes all the~PBE results in Stage I. Proposition \ref{region1} shows that when low~product preference $v_\text{\tiny L}$ is close to high preference $v_\text{\tiny H}$, it is optimal for the seller not to practice the learning-enabled personalized pricing. Instead, the seller charges \mbox{a low uniform price} $v_\text{\tiny L}$ in both selling periods without losing any buyers' demands. Given $v_\text{\tiny L}$ close to $v_\text{\tiny H}$, charging a high price $v_\text{\tiny H}$ only benefits the seller slightly when facing a high-preference buyer. However, this could result in a significant demand loss when facing a low-preference buyer. As the seller will not personalize the price offerings, buyers in Stage~I~behave~honestly without any manipulation of their social data.

\begin{proposition}\label{Region2}
In the strategic-learning model, if the difference between low and high preferences is small (see Region II of Fig. \ref{regions}), there exists a unique PBE as follows. 
\begin{itemize}
    \item \emph{In Stage I}, buyers $i$ and $j$ with the same high preference choose the maximum social interaction frequency in (\ref{region1_buyer}).
    
    \item \emph{In Stage II}, the seller charges a high price \mbox{$p_1^*=v_\text{\tiny H}$}~in the first selling period. In the second selling period, the seller offers a personalized price that depends on the buyers' common interaction frequency \mbox{$\hat{x}^*=\min\{x_{ij}^*,x_{ji}^*\}$}~in Stage I and the purchase record $a_1^*\in\{0,1\}$ of the first selling period in Stage II:
    \begin{numcases}{p_2^*=}
    v_{\text{\tiny H}}\boldsymbol{1}(a_1^*=1)+v_{\text{\tiny L}}\boldsymbol{1}(a_1^*=0),&\textrm{if $\ \hat{x}^\ast=1$,}\label{condition_1}\\
    v_{\text{\tiny L}}\boldsymbol{1}(a_1^*=1)+v_{\text{\tiny H}}\boldsymbol{1}(a_1^*=0),&\textrm{if $\ \hat{x}^\ast=0$.}\label{condition_2}
    \end{numcases}
\end{itemize}
\end{proposition}

As the low preference $v_{\text{\tiny L}}$ becomes no longer close to high preference $v_{\text{\tiny H}}$, the seller has the motivation to enable personalized pricing in Stage II. In this case, she charges a high price $p_1^*=v_{\text{\tiny H}}$ in the first selling period to learn the buyer's preference, and further tailors the second-period price $p_2$ accordingly. Although facing a possible personalized price in Stage II, Proposition \ref{Region2} shows that the two high-preference buyers still do not manipulate their social interactions in Stage I. This is because buyers' social interaction loss from manipulation outweighs the potential purchase gain from pretending as a low-preference type (i.e., $v_{\text{\tiny H}}-v_{\text{\tiny L}}<2(1-l)$).

\begin{proposition}\label{Region3}
In the strategic-learning model, if the difference between low and high preferences is medium (see Region III of Fig. \ref{regions}), there exists a unique PBE as follows.
\begin{itemize}
    \item \emph{In Stage I,} buyers $i$ and $j$ with the same high preference do not always choose the high social interaction frequency. Instead, they randomize low/high frequencies with a manipulation probability of $\rho^*=1-2(1-l)/(v_{\text{\tiny H}}-v_{\text{\tiny L}})$.
    \item \emph{In Stage II,} the seller charges a high price $p_1^*=v_{\text{\tiny H}}$ in the first selling period. The second-period price is the same as (\ref{condition_1}) and (\ref{condition_2}).
\end{itemize}
\end{proposition}

With a larger potential purchase gain when viewed as a low-preference type (i.e., $v_{\text{\tiny H}}-v_{\text{\tiny L}}>2(1-l)$), high-preference buyers are motivated to manipulate their social interactions with a positive probability of $\rho^*>0$. By doing so, they aim to avoid being charged high personalized prices but receive lower prices $v_{\text{\tiny L}}$ instead. As $v_{\text{\tiny H}}$ increases or $v_{\text{\tiny L}}$ decreases, the purchase gain $v_{\text{\tiny H}}-v_{\text{\tiny L}}$ from manipulation becomes more significant. Hence, the high-preference buyers would raise the manipulation probability $\rho^*$ to achieve a greater purchase surplus. 

Although the seller may fail to identify high-preference buyers given such social interaction manipulations, Proposition \ref{Region3} suggests that the seller still takes the same personalized pricing strategy in (\ref{condition_1}) and (\ref{condition_2}) as in Region II of Fig. \ref{regions}. This is because the probability of these buyers' manipulation is still minor given the moderate potential purchase gain, and the seller can still capture buyers' surplus in most cases by enabling personalized pricing.

\begin{proposition}\label{Region4}
In the strategic-learning model, if the difference between low and high preferences is large (see Region IV of Fig.~\ref{regions}), there exists a unique PBE as follows.
\begin{itemize}
    \item \emph{In Stage I,} buyers $i$ and $j$ with the same high preference do not always choose the high social interaction frequency. Instead, they randomize low/high frequencies with the~following manipulation probability:
    \begin{equation}\label{region4_rho}
    \rho^*=
    \left\{\begin{array}{ll}
    1-\sqrt{v_{\text{\tiny L}}/2(v_{\text{\tiny H}}-v_{\text{\tiny L}})},&\textrm{if $\ v_{\text{\tiny L}}\ge 2v_{\text{\tiny H}}/5$,}\\
    1-\sqrt{(v_{\text{\tiny H}}-2v_{\text{\tiny L}})/(v_{\text{\tiny H}}-v_{\text{\tiny L}})},&\textrm{otherwise.}
    \end{array}\right.
    \end{equation}

    \item \emph{In Stage II,} the seller does not always charge a high price $p_1^*=v_{\text{\tiny H}}$ in the first selling period and enable personalized pricing in (\ref{condition_1}) and (\ref{condition_2}) in the second selling period. More specifically, we have the following two cases.
    \begin{itemize}    
    \item [(i)]\mbox{\emph{\textbf{Case 1} with $v_{\text{\tiny L}} \ge 2v_{\text{\tiny H}}/5$:}} When observing buyers' high common interaction frequency \mbox{$\hat{x}^*=1$}, the seller randomizes with offering a low uniform price $p_1^*=p_2^*=v_{\text{\tiny L}}$ in both selling periods. i.e.,
    \begin{equation}\label{region4_mix1}
    \left\{ \begin{array}{ll}
    p_1^*=p_2^*=v_{\text{\tiny L}},&\textrm{w.p. $\frac{1}{2}-\frac{\sqrt{2}(1-l)}{\sqrt{v_{\text{\tiny L}}(v_{\text{\tiny H}}-v_{\text{\tiny L}})}}$,}\\
    p_1^*=v_{\text{\tiny H}}, \eqref{condition_1},&\textrm{otherwise.}
    \end{array}\right.
    \end{equation}
    Whereas observing \mbox{$\hat{x}^*=0$}, the seller purely charges a high price \mbox{$p_1^*=v_{\text{\tiny H}}$} in the first selling period and enables the personalized pricing in (\ref{condition_2}) in the second selling period.
        
    \item [(ii)]\mbox{\emph{\textbf{Case 2} with $v_{\text{\tiny L}}<2v_{\text{\tiny H}}/5$:}} When observing buyers' low common interaction frequency \mbox{$\hat{x}^*=0$}, the seller randomizes with offering a high uniform price $p_1^*=p_2^*=v_{\text{\tiny H}}$ in both selling periods. i.e.,
    \begin{equation}\label{region4_mix2}
    \left\{ \begin{array}{ll}
    p_1^*=p_2^*=v_{\text{\tiny H}},&\textrm{w.p. $1-\frac{2(1-l)}{\sqrt{(v_{\text{\tiny H}}-v_{\text{\tiny L}})(v_{\text{\tiny H}}-2v_{\text{\tiny L}})}}$,}\\
        p_1^*=v_{\text{\tiny H}}, (\ref{condition_2}),&\textrm{otherwise.}\\
    \end{array}\right.
    \end{equation}
    Whereas observing \mbox{$\hat{x}^*=1$}, the seller purely charges a high price $p_1^*=v_{\text{\tiny H}}$ in the first selling period and enables the personalized pricing in (\ref{condition_1}) in the second selling period.
    \end{itemize}
\end{itemize}
\end{proposition}

With a sufficiently large potential purchase gain $v_{\text{\tiny H}}-v_{\text{\tiny L}}$ in Region IV of Fig. \ref{regions}, high-preference buyers are now more willing to manipulate their social interactions. These high-preference buyers would raise the manipulation probability $\rho^*$ as $v_{\text{\tiny H}}$ increases or $v_{\text{\tiny L}}$ decreases in Case 1, because the purchase gain $v_{\text{\tiny H}}-v_{\text{\tiny L}}$ becomes more significant under these changes. On the other hand, the seller expects such social manipulations, and introduces a mixture (mixed strategy) of uniform pricing $v_{\text{\tiny L}}$ with personalized pricing in (\ref{region4_mix1}) when observing the high common interaction frequency $\hat{x}^*=1$. Given a higher $v_{\text{\tiny L}}$ in Case 1, such a mixture helps the seller gain more revenue from those low-preference buyers who never manipulate and frequently interact with each other. As a result, the seller mitigates the potential revenue loss from those manipulating high-preference buyers.

In contrast, when $v_{\text{\tiny L}}$ becomes lower, as in Case 2, the seller tends to mix personalized pricing with another pricing scheme in \eqref{region4_mix2}, offering a high uniform price $v_{\text{\tiny H}}$ when observing the low common interaction frequency $\hat{x}^*=0$. Specifically, the seller strategically infers those manipulating high-preference buyers upon $\hat{x}^*=0$. With small $v_{\text{\tiny L}}$, the seller worries little about the demand loss from low-preference buyers, who never manipulate and seldom interact with high-preference buyers. As a result, the seller mitigates the loss by extracting revenue from those manipulating high-preference buyers. Meanwhile, such a mixture of uniform pricing $p_1^*=p_2^*=v_{\text{\tiny H}}$ upon $\hat{x}^*=0$ also reduces high-preference buyers' incentives to manipulate. That is, when choosing the low interaction frequency of $\hat{x}^*=0$, the high-preference buyers are less likely to avoid personalized pricing but still face social loss. Therefore, even though the potential purchase gain $v_{\text{\tiny H}}-v_{\text{\tiny L}}$ enlarges as $v_{\text{\tiny H}}$ increases or $v_{\text{\tiny L}}$ decreases, high-preference buyers' manipulation level $\rho^*$ decreases in Case 2.

\section{Impact of Buyer Manipulation}\label{analysis}

After analyzing the PBE in Propositions \ref{region1}-\ref{Region4}, we are ready~to understand how buyers' social data manipulations affect the payoffs of both the buyers and the seller. We investigate this through welfare comparisons with the two benchmark cases (\textit{undisclosed-learning} and \textit{no-learning}) in Section \ref{benchmark}.

\subsection{Manipulation Effect on Buyer Payoffs}\label{impact_buyerpayoff}

First, we are interested in whether the buyers really achieve higher payoffs by manipulating their social interactions~to hurdle the seller's personalized pricing. Proposition \ref{dilemma_buyer} and Fig. \ref{userpayoff} below provide the answer by comparing the strategic-learning model to the undisclosed-learning benchmark, in which buyers are unaware of the seller's learning practice and thus never manipulate (see Section \ref{bench_unaware}).

Formally, we examine the \textit{average buyer payoff} $\tilde{\pi}$ given by
\begin{equation}\label{average_buyer}
\tilde{\pi}\triangleq\mathbb{E}_{v_i,v_j}\{\tilde{\pi}_i(x_{ij},x_{ji})\},
\end{equation}
which takes expectation over various buyer preference distributions while accounting for the preference correlation between buyers. Indeed, the awareness of the seller's learning would not affect low-preference buyers' payoffs, as they never manipulate and always receive zero purchase surplus (see Lemma \ref{lemma1_buyer} and Lemma \ref{pricebinary}). Hence, the investigation on (\ref{average_buyer}) to compare with the undisclosed-learning benchmark actually sheds light on the high-preference buyers, as~discussed in the following.

\begin{figure}[h]
\centering
\subfigure[Large {$l\in(1/2,1)$}]{
\begin{minipage}[t]{0.45\linewidth}
\centering
\includegraphics[width=0.8\linewidth]{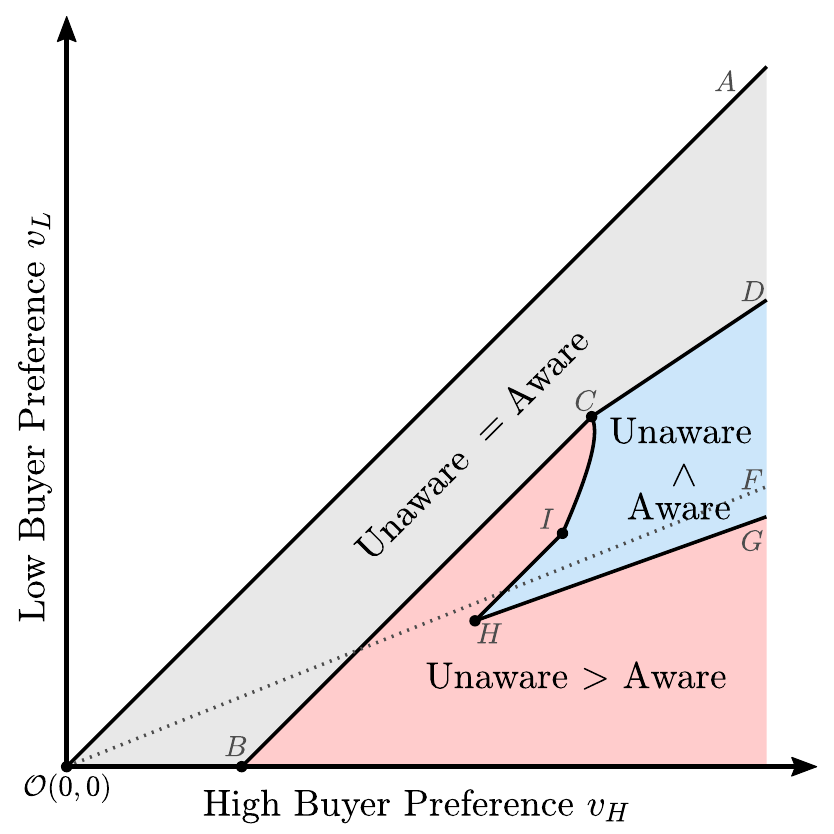}
\label{useparadox1_2}
\end{minipage}
}
\subfigure[Small {$l\in\left(0,1/2\right]$}]{
\begin{minipage}[t]{0.45\linewidth}
\centering
\includegraphics[width=0.8\linewidth]{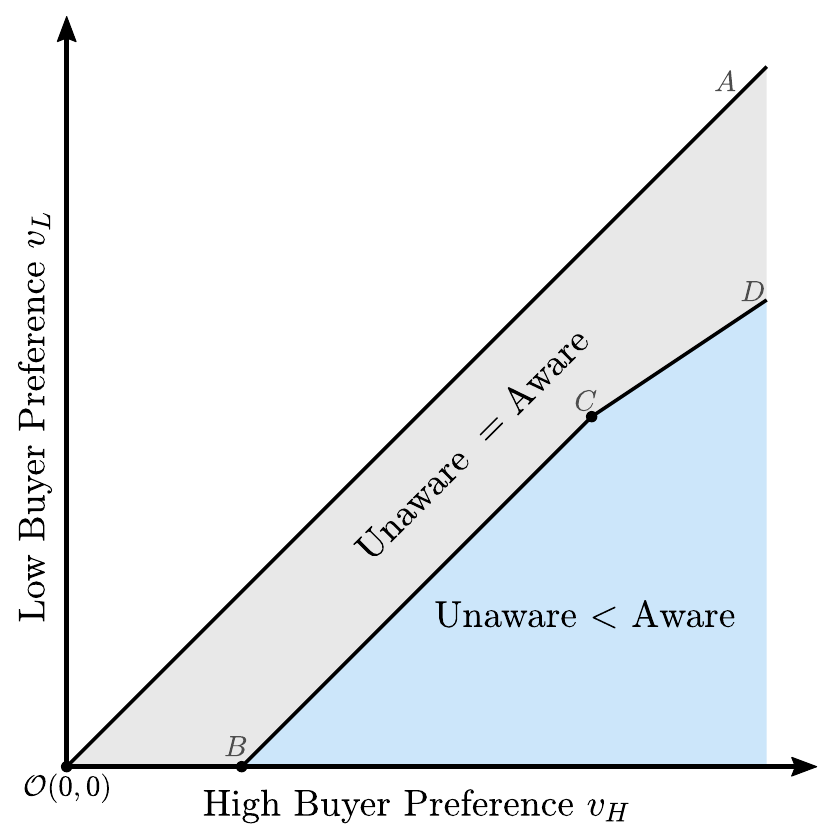}
\label{useparadox3}
\end{minipage}
}
\begin{minipage}[t]{\linewidth}
\vspace{5pt}
\centering
\scriptsize
\begin{tabular}{cp{6cm}}
\specialrule{\heavyrulewidth}{1.2pt}{0pt}
\cellcolor[gray]{0.92}{\makecell[c]{Curve}} &{\makecell[c]{Corresponding Equation in $(v_\text{\tiny H},v_\text{\tiny L})$-Plane}}\\
\specialrule{\lightrulewidth}{0pt}{1.2pt}
OA &\cellcolor[gray]{0.92}{\makecell[c]{$v_\text{\tiny L}=v_\text{\tiny H}$}}\\
\cellcolor[gray]{0.92}{CD} & {\makecell[c]{$v_\text{\tiny L}=2v_\text{\tiny H}/3$}}\\
OF &\cellcolor[gray]{0.92}{\makecell[c]{$v_\text{\tiny L}=2v_\text{\tiny H}/5$}}\\
\cellcolor[gray]{0.92}{HI} & {\makecell[c]{$v_\text{\tiny H}-v_\text{\tiny L}=2l$}}\\
BC &\cellcolor[gray]{0.92}{\makecell[c]{$v_\text{\tiny H}-v_\text{\tiny L}=2(1-l)$}}\\
\cellcolor[gray]{0.92}{CI} & {\makecell[c]{$v_\text{\tiny H}-v_\text{\tiny L}+2l\sqrt{v_\text{\tiny L}/2(v_\text{\tiny H}-v_\text{\tiny L})}=2$}}\\
HG & \cellcolor[gray]{0.92}{\makecell[c]{$v_\text{\tiny L}=\min\left\{(2l-1)v_\text{\tiny H}/(2l-1+l^2),2v_\text{\tiny H}/5\right\}$}}\\
\specialrule{\heavyrulewidth}{1.2pt}{0pt}
\end{tabular}
\vspace{5pt}
\end{minipage}
\caption{Comparison of an average buyer's average payoff, when aware versus unaware of the seller's learning (i.e., in the strategic-learning model versus the undisclosed-learning benchmark). As noticed in Table \ref{table_same}, the parameter $l$ represents the loss one buyer incurs due to the other buyer's low social response when both buyers have the same high preference.}
\label{userpayoff}
\end{figure}

\begin{proposition}\label{dilemma_buyer}
Compared with the undisclosed-learning benchmark, an average buyer's payoff can decrease after being aware of the seller's learning in the strategic-learning model, as illustrated in the red region of Fig. \ref{useparadox1_2}.
\end{proposition}

Proposition \ref{dilemma_buyer} reveals an interesting dilemma that \emph{high-preference buyers, after strategic manipulation to thwart 
personalized pricing, may end up with lower payoffs than no manipulation} (see the Red region of Fig. \ref{useparadox1_2}). In particular, this happens when one high-preference buyer's social loss $l$, incurred~by the other high-preference buyer's low interaction response in Table \ref{table_same}, is significant \mbox{(i.e., $l>1/2$)}. One may imagine that each buyer is able to confuse the seller through \textit{unilateral manipulation}, as the seller only observes the common mutual interactions \mbox{$\hat{x}=\min\{x_{ij},x_{ji}\}$}. However, both high-preference buyers \mbox{$i$ and $j$} tend to manipulate unilaterally by themselves in this case. Consequently, both buyers suffer the social interaction loss associated with the other's manipulation, imposing a \textit{negative externality} on each other.

Such a payoff loss will not happen when $v_{\text{\tiny H}}-v_{\text{\tiny L}}$ is sufficiently large (see the Blue region of Fig. \ref{useparadox1_2}). In this case, the benefit of a buyer's manipulation to avoid personalized pricing in Stage II is now significant enough to outweigh the social loss incurred from the other buyer's manipulation in Stage I. Furthermore, the seller's randomization of personalized pricing with uniform pricing in Proposition \ref{Region4} reduces buyers' incentives to manipulate their social interactions, which in turn alleviates the negative externality they impose on each other.

When $l$ becomes small in Fig. \ref{useparadox3}, each high-preference buyer prefers the other high-preference buyer to manipulate his social interaction data to avoid the seller's personalized pricing, rather than doing it himself. As a result, both buyers can fully account for the external social loss associated with the other's manipulation when making social decisions and thus successfully circumvent the worse-off dilemma.

\subsection{Guidance for the Seller's Learning Practice}\label{guidance_seller}

Next, we investigate whether it is beneficial for the seller~to learn from buyers' social interaction data, and whether~she should actively inform buyers of such data access. For example, Amazon today routinely seeks explicit consent from buyers regarding Amazon's access to its buyers' Facebook data~\cite{-_by_-_boulton_2021}. However, some other companies are found not to inform buyers regarding their data access \cite{hamilton_2020}. To investigate these issues, we compare the seller's expected~revenue obtained in the strategic-learning model with the no-learning and the undisclosed-learning benchmarks.

\begin{proposition}\label{salerevenue1}
In the strategic-learning model, the seller's~expected revenue is no less than the no-learning benchmark. Moreover, the strategic learning's revenue gain is positive as long as $v_{\text{\tiny L}}<2v_{\text{\tiny H}}/3$, and it can reach a maximum of $25\%$.
\end{proposition}

With a significant revenue gain up to $25\%$ beyond the no-learning benchmark, the seller should always learn from buyers' social interaction data and enable personalized pricing accordingly, especially when buyers' preferences for~the product are diverse \mbox{(i.e., $v_{\text{\tiny L}}<2v_{\text{\tiny H}}/3$)}. Importantly, this applies no matter whether buyers are aware (or not) of such learning and follow-up personalized pricing (see Corollary \ref{salerevenue_un}). Despite buyers' social manipulation to hurdle the seller's personalized pricing, the seller can strategically learn and introduce a mixture with uniform pricing (see Proposition~\ref{Region4}) to mitigate the manipulation effect and avoid revenue loss.

\begin{proposition}\label{salerevenue2}
In the strategic-learning model, the seller's expected revenue is no more than the undisclosed-learning benchmark, with a revenue loss of $8.3\%$ at most.
\end{proposition}

While buyers' awareness would reduce the seller's revenue gain from learning, the corresponding loss is relatively small, with only 8.3\% in the worst case. Given the increasing trend of better personal data protection from regulators, it is advisable for the seller to inform buyers of her access to their social data. Indeed, this matches well with Amazon's current practice regarding \textit{informed consent} for data sharing \cite{-_by_-_boulton_2021}. From a regulatory standpoint, imposing a fine on sellers that exceeds the potential gain from not informing buyers (an 8.3\% limit in our model) should be adequate to encourage the seller's adoption of informed-consent practices.

\section{Multi-buyer Dynamic Learning}\label{sec:multi}

Previously, we focused on the basic instance involving two buyers with pairwise correlation, which serves as guidance for the seller's learning to exercise personalized pricing. In this section, we extend to investigate more general dynamic learning from multiple interconnected buyers in the social network across different time stages. In particular, we are interested in understanding whether the seller can benefit from learning previous buyers' preferences to infer the other buyers in the subsequent time stages.

Formally, we introduce the network model of the multi-buyer dynamic learning process in Section \ref{subsection:network_model}. To examine the impact of involving more interconnected buyers, we first analyze multi-buyer variants of benchmarks in Section \ref{subsec:multi_buyer_benchmark}. However, due to buyer manipulation, the dynamic Bayesian game for general multi-buyer networks is highly challenging to analyze. Later in Section \ref{sec:multiple_awareness}, we will first draw~an~important conclusion through a three-buyer network instance: \textit{learning prior buyers' preferences may not contribute to inferring the other buyers in the seller's subsequent learning}. Remarkably, we further identify the conditions under which this conclusion persists in more general multi-buyer networks. These findings offer additional guidance for the learning practice of today's online sellers beyond those obtained in Section \ref{guidance_seller}.

\subsection{Multi-buyer Network Model}\label{subsection:network_model}

In this subsection, we formally introduce a general network model, where the seller implements dynamic learning across multiple interconnected buyers.

Consider a social network of $N$ interconnected buyers, where each buyer is connected to at least one other buyer within the network. Additionally, we assume that when a buyer $i$ interacts with two or more of his neighbor friends in the social network, the social interaction utility $u_i$ from these interactions will be additive for this buyer. We adopt a multi-buyer variation of the two-stage formulation in Fig.~\ref{Timeline}: In Stage I, the $N$ interconnected buyers interact with each other in the online social network; Stage II consists of $N$ purchase periods, with one buyer per period, during which interaction between the buyer and the seller in the product market occur. As modeled in Section \ref{model_stageII}, buyers sequentially arrive at the market with a random sequence.

\subsection{Multi-buyer Variants of Benchmarks}\label{subsec:multi_buyer_benchmark}

We start with multi-buyer variants of our no-learning and undisclosed-learning benchmarks in Section \ref{benchmark}. Through the analysis in this subsection, we aim to understand how more interconnected buyers or additional correlation information can affect the seller's learning and pricing strategies.

First, we revisit the \textit{no-learning benchmark} established in Section \ref{bench_nolearning} and observe that the conclusion drawn there~(see Lemma \ref{lemma:nolearning}) remains applicable to the multi-buyer learning process here. Then we turn to the \textit{undisclosed-learning benchmark}, and the following proposition will outline the unique equilibrium under its multi-buyer setting as in Section \ref{bench_unaware}.

\begin{proposition}\label{unaware_general}
In the undisclosed-learning benchmark, given $N$ interconnected buyers in the social network with their preferences unknown to the seller, the unique SPE of the dynamic game is given as follows.

\begin{itemize}
    \item \mbox{\emph{In Stage I,}} if two connected buyers $i$ and $j$ have the same preference \mbox{($v_i=v_j$)}, both choose the maximum social interaction frequency, i.e., \mbox{$x^*_{ij}=x^*_{ji}=1$}. If they have different preferences \mbox{($v_i\neq v_j$)}, they choose the minimum interaction frequency, i.e., \mbox{$x^*_{ij}=x^*_{ji}=0$}. 
    \item \mbox{\emph{In Stage II,}} the seller's equilibrium pricing depends on the relation between $v_{\text{\tiny L}}$ and $v_{\text{\tiny H}}$ as follows.
    \begin{itemize}
    \item [(i)] \emph{\textbf{Large $\bm{v_{\text{\tiny L}}}$ regime} with $v_{\text{\tiny L}}\in\left[ \nicefrac{Nv_{\text{\tiny H}}}{N+1},v_{\text{\tiny H}}\right)$}: The seller charges low prices in each selling period, i.e., $p_t^*=v_{\text{\tiny L}}$, for each $t\in\left\{1,2,\dots,N\right\}$.
        
    \item [(ii)]\emph{\textbf{Small $\bm{v_{\text{\tiny L}}}$ regime} with $v_{\text{\tiny L}}\in\left(0,\nicefrac{Nv_{\text{\tiny H}}}{N+1}\right)$}: The seller charges a high price $p_1^*=v_{\text{\tiny H}}$ in the first selling period. Then, the seller perfectly infers all buyers' preferences. In each subsequent selling period, the seller offers a personalized price equal to the arriving buyer's preference, i.e., $p_t^*=v_t$, for each $t\in\left\{2,3,\dots,N\right\}$.
    \end{itemize} 
\end{itemize}
In the undisclosed-learning benchmark, the seller's expected revenue is no less than the no-learning benchmark (with a revenue gain up to ${\textstyle\tfrac{N-1}{2N}}\times 100\%$).
\end{proposition}

Proposition \ref{unaware_general} generalizes Lemma \ref{unaware} and Corollary \ref{salerevenue_un}~beyond the two-buyer network. Moreover, as more interconnected buyers participate, setting a high price of~\mbox{$p_1=v_{\text{\tiny H}}$} in the first selling period to ascertain the first-arriving buyer's preference can expose more buyers' preference information. This, in turn, further empowers the seller's ability to implement personalized pricing. As a consequence, the seller benefits more from the undisclosed learning practices, and the revenue gain increases with the number of interconnected buyers $N$.

We next proceed by introducing an extra buyer into the social network, \textit{whose preference is already known to the seller due to the past learning process}. In what follows, we will~describe the equilibrium outcome of this new setting under the~\textit{undisclosed-learning benchmark}, and then compare it with Proposition \ref{unaware_general} to establish the effect of such prior knowledge on the seller's follow-up learning process. Henceforth, we use the term ``\textit{unknown buyers}'' to refer to buyers whose preferences are unknown to the seller until a new-round learning takes place, and ``\textit{known buyers}'' refers to users whose preferences are already known to the seller.

\begin{proposition}\label{unaware_general_known}
In the undisclosed-learning benchmark, as long as one known buyer is connected to the $N$ interconnected unknown buyers in the social network, the seller can perfectly infer all buyers' preferences. This prior knowledge allows the seller to charge a personalized price to each arriving buyer, i.e., $p_t^*=v_t$, for each $t\in\{1,2,\dots, N\}$. After introducing this known buyer to the unknown buyer network, the seller's expected revenue always increases (by at most $\frac{1}{2N}\times 100\%$).
\end{proposition}

Proposition \ref{unaware_general_known} remains valid even when more than one known buyer is introduced. Since neither known nor unknown buyers manipulate in Stage I, the seller can perfectly observe the pairwise preference correlation among buyers based on their social interaction data. Given the known buyer's preference information, the seller can perfectly infer all other buyers' product preferences. This enables the seller to provide personalized prices to each arriving buyer, obtaining the maximum revenue.

\subsection{Strategic Learning with Buyer Awareness}\label{sec:multiple_awareness}

Now, we turn to the seller's strategic learning in the presence of buyers' awareness. In this subsection, we first reveal the robustness of some major insights beyond the two-buyer network in Section \ref{robustness:general}. However, dealing with a general multi-buyer network is challenging, and we thus explore the seller's strategic learning with additional knowledge about prior buyers' preferences. Specifically, we introduce a three-buyer network instance with concrete discussions in Section \ref{subsec:three-buyer} and conclude the effectiveness of such prior knowledge for more general networks in Section \ref{subsec:general_buyer}.

\subsubsection{Robustness against General Multi-buyer Network}\label{robustness:general}

In this subsubsection, we consider a general social network (without any additional prior knowledge) and demonstrate the robustness of some major insights beyond the two-buyer network. Specifically, the methodology of our two-buyer equilibrium analysis in Sections~\ref{PBE1} and~\ref{PBE2} remains applicable. In particular, Proposition \ref{multiple:buyer_manipulation} below generalizes Lemma \ref{lemma1_buyer} for the two-buyer case into a more general social network.

\begin{proposition}\label{multiple:buyer_manipulation}
Consider a general social network of $N$ interconnected unknown buyers. In the PBE, social manipulations can only happen between two high-preference buyers.
\end{proposition}

Though Proposition \ref{multiple:buyer_manipulation} shows key properties for buyers' equilibrium manipulation strategies, the dynamic learning over multiple interconnected buyers remains challenging~to analyze. The challenges arise from the immensely vast~equilibrium space: (i) many combinations of social interaction possibilities for the seller to strategically learn from Stage I, and (ii) the presence of more than two selling periods (one per buyer) in Stage II.

To circumvent these complications, we will henceforth fix the prior buyers' preference information, obtained from past learning, and proceed to investigate the seller's subsequent learning process \textbf{\textit{with such prior knowledge}}.

\begin{remark}
Our multi-buyer network model, which accounts for known buyers, readily includes influencers or celebrities. In this model, influencers or celebrities can substitute for known buyers, with their preferences being known to the seller. Consequently, the analysis presented in Section 8.3 for multi-buyer learning applies to social networks featuring these influential individuals.
\end{remark}

\subsubsection{Three-buyer Network Instance with Prior Knowledge}\label{subsec:three-buyer}
Before jumping into a general network of multiple interconnected buyers, we first start with a simple yet fundamental extension from the two-buyer social network in Section \ref{systemmodel}. This extended instance allows for a comprehensive analysis and provides clean insights with more concrete illustrations, which guides us to conclude the effectiveness of the seller's dynamic learning in more general networks in the following subsubsection.

As illustrated in Fig. \ref{3buyer_ring}, we introduce a third (incumbent) buyer $k$ to the two connected buyers $i$ and~$j$. We consider that the preference of buyer $k$ is already known to the seller due to her past learning practice. Now, our key question is whether the seller benefits from such prior knowledge when pricing to (new) buyers $i$ and~$j$, whose private preferences are unknown to the seller yet. Without loss of generality, we assume that $v_i\le v_j$ henceforth.

\begin{figure}[h]
\begin{center}
\begin{tikzpicture}[scale=0.8]
\draw[thick,-latex,black!50] (-0.74,0.06) -- (0.74,0.06);
\draw[thick,latex-,black!50] (-0.74,-0.06) -- (0.74,-0.06);
\draw[thick,latex-,black!50] (-0.185,1.515) -- (-0.935,0.216);
\draw[thick,-latex,black!50] (-0.065,1.515) -- (-0.815,0.216);
\draw[thick,latex-,black!50] (0.185,1.515) -- (0.935,0.216);
\draw[thick,-latex,black!50] (0.065,1.515) -- (0.815,0.216);
\filldraw[color=red,fill=red!10,thick] (0,1.732) circle (0.22);
\filldraw[color=blue,fill=blue!15,thick] (-1,0) circle (0.22);
\filldraw[color=blue,fill=blue!15,thick] (1,0) circle (0.22);
\node [black] at (0,1.732) {\color{red}\small $k$};
\node [black] at (-1,0) {\color{blue}\small $i$};
\node [black] at (1,0) {\color{blue}\small $j$};
\node [black,align=center] at (0,-0.6) {\color{blue}\scriptsize New Buyers $i$ and $j$ whose preferences are unknown to the seller.};
\node [black,align=left] at (0,2.2) {\color{red}\scriptsize Incumbent Buyer $k$ whose preference is known to the seller.};
\draw [->,thin,red] (0.23,1.732) to [in = 270, out = 0] (1,2.0);
\draw[blue,decorate,decoration=brace,thin] (1,-0.27) -- (-1,-0.27);
\end{tikzpicture}
\end{center}
\vspace{-15pt}
\caption{Illustration of a network instance with three connected buyers. Here, the upper node represents the known buyer $k$, and the two~lower nodes correspond to unknown buyers $i$ and $j$. The two opposing directional arrows connecting any buyer pair indicate their bilateral interactions. Notice that the arrows in this graph do not depict the specific frequency choices of buyers' social interactions.}
\label{3buyer_ring}
\end{figure}
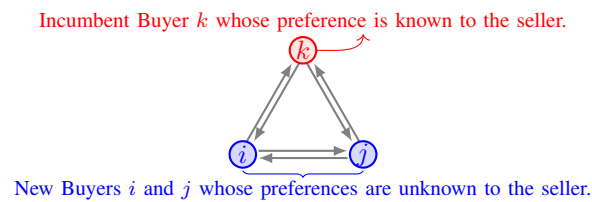

With preference information known to the seller, buyer $k$ has no incentive to manipulate his social interactions and preference correlation with any other buyer (e.g., buyers $i$ and $j$ in Fig. \ref{3buyer_ring}) in Stage I. In what follows, we analyze the effect of such prior knowledge on the seller as well as buyers $i$ and~$j$, based on the high or low preference type of buyer~$k$.

\noindent
\mbox{\textbf{Pricing from Prior Knowledge of $v_k=v_{\text{\tiny L}}$.}} Let us first consider the case in which the known buyer $k$ has a low~preference, i.e., \mbox{$v_k=v_{\text{\tiny L}}$}. 

\begin{figure}[h]
\begin{center}
\begin{tikzpicture}
\filldraw[color=black,fill=black!5,ultra thin] (-4.25,0.1) rectangle (4.25,1.3);
\filldraw[color=black,fill=black,thick] (-3.95,1.025) circle (0.12);
\filldraw[color=black,fill=white,thick] (0.5,1.025) circle (0.12);
\node at (-2.35,1.02) {\scriptsize{High-preference buyer.}};
\node at (2.1,1.02) {\scriptsize{Low-preference buyer.}};
\draw[black,-latex,thin] (-4.07,0.675) -- (-3.27,0.675);
\draw[black,-latex,thin,densely dashed] (-3.1,0.675) -- (-2.3,0.675);
\draw[red,-latex,thin] (-4.07,0.325) -- (-3.27,0.325);
\draw[red,-latex,thin,densely dashed] (-3.1,0.325) -- (-2.3,0.325);
\node at (0.7,0.67) {\scriptsize{Truthful high/low social interaction frequency.}};
\node at (0.9,0.32) {\scriptsize{Manipulated high/low social interaction frequency.}};
\node at (-3.185,0.675){\scriptsize{/}};
\node at (-3.185,0.325){\scriptsize{/}};
\end{tikzpicture}
\end{center}
\vspace{-12pt}
\begin{minipage}[t]{0.22\linewidth}
\vspace{-1pt}
\begin{center}
\begin{tikzpicture}[scale=0.8]
\draw[thick,-latex,black!50] (-0.74,0.06) -- (0.74,0.06);
\draw[thick,latex-,black!50] (-0.74,-0.06) -- (0.74,-0.06);
\draw[thick,latex-,black!50] (-0.185,1.515) -- (-0.935,0.216);
\draw[thick,-latex,black!50] (-0.065,1.515) -- (-0.815,0.216);
\draw[thick,latex-,black!50] (0.185,1.515) -- (0.935,0.216);
\draw[thick,-latex,black!50] (0.065,1.515) -- (0.815,0.216);
\filldraw[color=red,fill=red!10,thick] (0,1.732) circle (0.22);
\filldraw[color=blue,fill=blue!15,thick] (-1,0) circle (0.22);
\filldraw[color=blue,fill=blue!15,thick] (1,0) circle (0.22);
\node [black] at (0,1.732) {\color{red}\small $k$};
\node [black] at (-1,0) {\color{blue}\small $i$};
\node [black] at (1,0) {\color{blue}\small $j$};
\node [black,align=center] at (0,-0.695) {\color{blue}\scriptsize Unknown Buyers};
\node [black,align=right] at (0.95,1.45) {\color{red}\scriptsize Known};
\node [black,align=right] at (0.95,1.15) {\color{red}\scriptsize Buyer};
\draw [<-,thin,red] (0.23,1.732) to [in = 60, out = 0] (0.95,1.475);
\draw[blue,decorate,decoration=brace,thin] (1,-0.28) -- (-1,-0.28);
\end{tikzpicture}
\end{center}
\end{minipage}
\subfigure[$v_i=v_j=v_{\text{\tiny L}}$]{
\begin{minipage}[t]{0.22\linewidth}
\vspace{0pt}
\begin{center}
\begin{tikzpicture}[scale=0.8]
\draw[thick,-latex,black] (-0.74,0.06) -- (0.74,0.06);
\draw[thick,latex-,black] (-0.74,-0.06) -- (0.74,-0.06);
\draw[thick,latex-,black,black] (-0.185,1.515) -- (-0.935,0.216);
\draw[thick,-latex,black,black] (-0.065,1.515) -- (-0.815,0.216);
\draw[thick,latex-,black,black] (0.185,1.515) -- (0.935,0.216);
\draw[thick,-latex,black,black] (0.065,1.515) -- (0.815,0.216);
\filldraw[color=black,fill=white,thick] (0,1.732) circle (0.22);
\filldraw[color=black,fill=white,thick] (-1,0) circle (0.22);
\filldraw[color=black,fill=white,thick] (1,0) circle (0.22);
\end{tikzpicture}
\end{center}
\label{Ring_L_a}
\end{minipage}
}
\subfigure[$v_i\neq v_j$]{
\begin{minipage}[t]{0.22\linewidth}
\vspace{0pt}
\begin{center}
\begin{tikzpicture}[scale=0.8]
\draw[thick,-latex,black,densely dashed] (-0.74,0.06) -- (0.74,0.06);
\draw[thick,latex-,black,densely dashed] (-0.74,-0.06) -- (0.74,-0.06);
\draw[thick,latex-,black] (-0.185,1.515) -- (-0.935,0.216);
\draw[thick,-latex,black] (-0.065,1.515) -- (-0.815,0.216);
\draw[thick,latex-,black,densely dashed] (0.185,1.515) -- (0.935,0.216);
\draw[thick,-latex,black,densely dashed] (0.065,1.515) -- (0.815,0.216);
\filldraw[color=black,fill=white,thick] (0,1.732) circle (0.22);
\filldraw[color=black,fill=white,thick] (-1,0) circle (0.22);
\filldraw[color=black,fill=black,thick] (1,0) circle (0.22);
\end{tikzpicture}
\end{center}
\label{Ring_L_b}
\end{minipage}
}
\subfigure[$v_i=v_j=v_{\text{\tiny H}}$]{
\begin{minipage}[t]{0.22\linewidth}
\vspace{0pt}
\begin{center}
\begin{tikzpicture}[scale=0.8]
\draw[thick,-latex,black] (-0.74,0.06) -- (0.74,0.06);
\draw[thick,latex-,black] (-0.74,-0.06) -- (0.74,-0.06);
\draw[thick,latex-,black,densely dashed] (-0.185,1.515) -- (-0.935,0.216);
\draw[thick,-latex,black,densely dashed] (-0.065,1.515) -- (-0.815,0.216);
\draw[thick,latex-,black,densely dashed] (0.185,1.515) -- (0.935,0.216);
\draw[thick,-latex,black,densely dashed] (0.065,1.515) -- (0.815,0.216);
\filldraw[color=black,fill=white,thick] (0,1.732) circle (0.22);
\filldraw[color=black,fill=black,thick] (-1,0) circle (0.22);
\filldraw[color=black,fill=black,thick] (1,0) circle (0.22);
\end{tikzpicture}
\end{center}
\label{Ring_L_c}
\end{minipage}
}
\vspace{-10pt}
\caption{Illustration of the three buyers' social interactions in the unique PBE given the upper buyer $k$'s preference is low.}
\label{Ring_L}
\vspace{-10pt}
\end{figure}
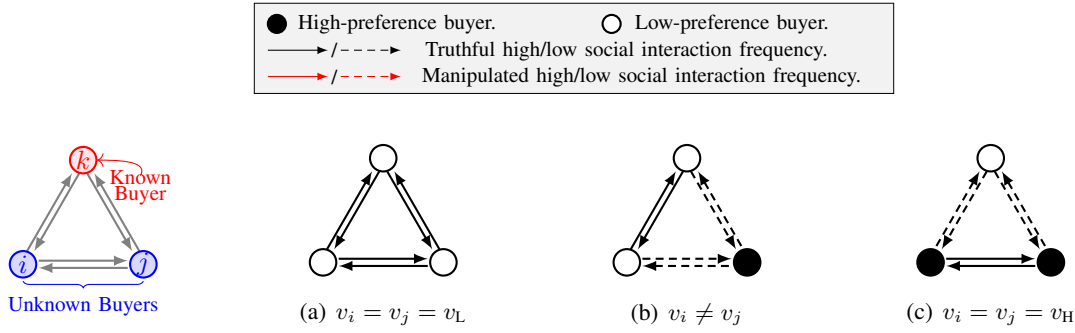

In this case, a unique PBE exists in which buyers never manipulate their social interactions with any other buyer in Stage I, as illustrated in Fig. \ref{Ring_L}. This finding is consistent with Proposition \ref{multiple:buyer_manipulation}. The seller thus perfectly infers the remaining buyers' preferences, $v_i$ and $v_j$, through their low-preference friend $k$ in Fig. \ref{Ring_L_a} and Fig. \ref{Ring_L_b}. Even when both buyers $i$ and $j$ have high preferences, as illustrated in Fig. \ref{Ring_L_c}, they still have no incentive to manipulate their social interaction data. 

\begin{lemma}\label{Lemma:price_prior_knowledge_VL_three}
Given the prior knowledge of \mbox{$v_k=v_{\text{\tiny L}}$}, in each selling~period of Stage II, the seller offers a personalized price equal to the arriving buyer's preference, i.e., \mbox{$p_t^*=v_t$, for each $t\in\{1,2\}$}. Therefore, the seller always extracts maximum surplus from buyers $i$ and $j$, leading to a significant revenue improvement than the two-buyer case in the absence of buyer $k$ (by at most $27.3\%$).
\end{lemma}

Not surprisingly, the prior knowledge of \mbox{$v_k=v_{\text{\tiny L}}$} always benefits the seller's subsequent learning practice, resulting in a substantial revenue gain. However, this positive effect might not persist when the seller's knowledge is limited to prior buyers with high preferences, as discussed in what~follows.

\noindent
\textbf{Pricing from Prior Knowledge of $v_k=v_{\text{\tiny H}}$.} Next, we~turn~to the other case in which buyer $k$ has~a~high preference,~i.e., \mbox{$v_k=v_{\text{\tiny H}}$}. In this case, the unknown buyers $i$ and $j$ may~manipulate their social interactions to hurdle~the seller's learning process.\footnote{For the sake of completeness, we present a formal characterization of all pure-strategy PBE in Appendix \ref{High-known:Pure-PBE}.} One may wonder whether the seller still benefits from such prior knowledge to facilitate her subsequent learning and gain revenue. The following lemma~provides the answer to this question.

\begin{lemma}\label{High_known:no_benefit}
If the following condition holds:
\begin{equation}\label{High_known:no_benefit_condition}
    \frac{2}{3}v_{\text{\tiny H}}<v_{\text{\tiny L}}< v_{\text{\tiny H}}-(1-l),
\end{equation}
then there always exists a pure-strategy PBE such that the seller does not gain any expected revenue after learning from a known high-preference buyer $k$ with $v_k=v_{\text{\tiny H}}$.
\end{lemma}

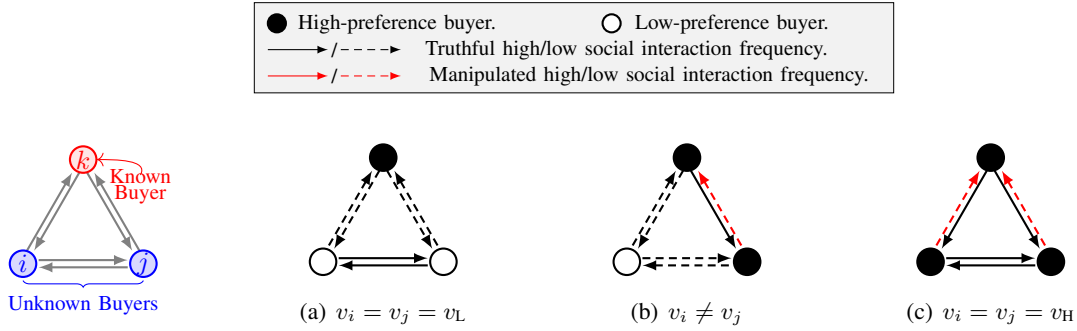
\begin{figure}[h]
\begin{center}
\begin{tikzpicture}
\filldraw[color=black,fill=black!5,ultra thin] (-4.25,0.1) rectangle (4.25,1.3);
\filldraw[color=black,fill=black,thick] (-3.95,1.025) circle (0.12);
\filldraw[color=black,fill=white,thick] (0.5,1.025) circle (0.12);
\node at (-2.35,1.02) {\scriptsize{High-preference buyer.}};
\node at (2.1,1.02) {\scriptsize{Low-preference buyer.}};
\draw[black,-latex,thin] (-4.07,0.675) -- (-3.27,0.675);
\draw[black,-latex,thin,densely dashed] (-3.1,0.675) -- (-2.3,0.675);
\draw[red,-latex,thin] (-4.07,0.325) -- (-3.27,0.325);
\draw[red,-latex,thin,densely dashed] (-3.1,0.325) -- (-2.3,0.325);
\node at (0.7,0.67) {\scriptsize{Truthful high/low social interaction frequency.}};
\node at (1,0.32) {\scriptsize{Manipulated high/low social interaction frequency.}};
\node at (-3.185,0.675){\scriptsize{/}};
\node at (-3.185,0.325){\scriptsize{/}};
\end{tikzpicture}
\end{center}
\vspace{-12pt}
\begin{minipage}[t]{0.22\linewidth}
\vspace{-1pt}
\begin{center}
\begin{tikzpicture}[scale=0.8]
\draw[thick,-latex,black!50] (-0.74,0.06) -- (0.74,0.06);
\draw[thick,latex-,black!50] (-0.74,-0.06) -- (0.74,-0.06);
\draw[thick,latex-,black!50] (-0.185,1.515) -- (-0.935,0.216);
\draw[thick,-latex,black!50] (-0.065,1.515) -- (-0.815,0.216);
\draw[thick,latex-,black!50] (0.185,1.515) -- (0.935,0.216);
\draw[thick,-latex,black!50] (0.065,1.515) -- (0.815,0.216);
\filldraw[color=red,fill=red!10,thick] (0,1.732) circle (0.22);
\filldraw[color=blue,fill=blue!15,thick] (-1,0) circle (0.22);
\filldraw[color=blue,fill=blue!15,thick] (1,0) circle (0.22);
\node [black] at (0,1.732) {\color{red}\small $k$};
\node [black] at (-1,0) {\color{blue}\small $i$};
\node [black] at (1,0) {\color{blue}\small $j$};
\node [black,align=center] at (0,-0.695) {\color{blue}\scriptsize Unknown Buyers};
\node [black,align=right] at (0.95,1.45) {\color{red}\scriptsize Known};
\node [black,align=right] at (0.95,1.15) {\color{red}\scriptsize Buyer};
\draw [<-,thin,red] (0.23,1.732) to [in = 60, out = 0] (0.95,1.475);
\draw[blue,decorate,decoration=brace,thin] (1,-0.28) -- (-1,-0.28);
\end{tikzpicture}
\end{center}
\label{Ring_H_Template}
\end{minipage}
\subfigure[$v_i=v_j=v_{\text{\tiny L}}$]{
\begin{minipage}[t]{0.22\linewidth}
\vspace{0pt}
\begin{center}
\begin{tikzpicture}[scale=0.8]
\draw[thick,-latex,black] (-0.74,0.06) -- (0.74,0.06);
\draw[thick,latex-,black] (-0.74,-0.06) -- (0.74,-0.06);
\draw[thick,latex-,black,densely dashed] (-0.185,1.515) -- (-0.935,0.216);
\draw[thick,-latex,black,densely dashed] (-0.065,1.515) -- (-0.815,0.216);
\draw[thick,latex-,black,densely dashed] (0.185,1.515) -- (0.935,0.216);
\draw[thick,-latex,black,densely dashed] (0.065,1.515) -- (0.815,0.216);
\filldraw[color=black,fill=black,thick] (0,1.732) circle (0.22);
\filldraw[color=black,fill=white,thick] (-1,0) circle (0.22);
\filldraw[color=black,fill=white,thick] (1,0) circle (0.22);
\end{tikzpicture}
\end{center}
\label{Ring_H_NoRevenueGain_a}
\end{minipage}
}
\subfigure[$v_i\neq v_j$]{
\begin{minipage}[t]{0.22\linewidth}
\vspace{0pt}
\begin{center}
\begin{tikzpicture}[scale=0.8]
\draw[thick,-latex,black,densely dashed] (-0.74,0.06) -- (0.74,0.06);
\draw[thick,latex-,black,densely dashed] (-0.74,-0.06) -- (0.74,-0.06);
\draw[thick,latex-,black,densely dashed] (-0.185,1.515) -- (-0.935,0.216);
\draw[thick,-latex,black,densely dashed] (-0.065,1.515) -- (-0.815,0.216);
\draw[thick,latex-,red,densely dashed] (0.185,1.515) -- (0.935,0.216);
\draw[thick,-latex,black] (0.065,1.515) -- (0.815,0.216);
\filldraw[color=black,fill=black,thick] (0,1.732) circle (0.22);
\filldraw[color=black,fill=white,thick] (-1,0) circle (0.22);
\filldraw[color=black,fill=black,thick] (1,0) circle (0.22);
\end{tikzpicture}
\end{center}
\label{Ring_H_NoRevenueGain_b}
\end{minipage}
}
\subfigure[$v_i=v_j=v_{\text{\tiny H}}$]{
\begin{minipage}[t]{0.22\linewidth}
\vspace{0pt}
\begin{center}
\begin{tikzpicture}[scale=0.8]
\draw[thick,-latex,black] (-0.74,0.06) -- (0.74,0.06);
\draw[thick,latex-,black] (-0.74,-0.06) -- (0.74,-0.06);
\draw[thick,latex-,red,densely dashed] (-0.185,1.515) -- (-0.935,0.216);
\draw[thick,-latex,black] (-0.065,1.515) -- (-0.815,0.216);
\draw[thick,latex-,red,densely dashed] (0.185,1.515) -- (0.935,0.216);
\draw[thick,-latex,black] (0.065,1.515) -- (0.815,0.216);
\filldraw[color=black,fill=black,thick] (0,1.732) circle (0.22);
\filldraw[color=black,fill=black,thick] (-1,0) circle (0.22);
\filldraw[color=black,fill=black,thick] (1,0) circle (0.22);
\end{tikzpicture}
\end{center}
\label{Ring_H_NoRevenueGain_c}
\end{minipage}
}
\vspace{-12pt}
\caption{Illustration of the three buyers' social interactions in the PBE in which the seller does not gain in her expected revenue given the upper buyer $k$'s preference is high.}
\label{Ring_H_NoRevenueGain}
\end{figure}

Fig. \ref{Ring_H_NoRevenueGain} visually demonstrates the equilibrium social interactions among the three buyers under the condition~of~\eqref{High_known:no_benefit_condition}. Different from \mbox{$v_k=v_{\text{\tiny L}}$}, here the prior knowledge of \mbox{$v_k=v_{\text{\tiny H}}$} does not directly reveal the true preference type of buyer $k$'s friends $i$ and $j$. Recall that the seller's access to buyers' pairwise social interaction data is limited to the common interaction frequency $\hat{x}$ in \eqref{hat} (with a minimum operation), which leaves room for buyers $i$ and $j$ to hide their private preferences through unilateral manipulation, as detailed in the following.

As illustrated in Fig. \ref{Ring_H_NoRevenueGain}, high-preference buyers $i$ and~$j$~are motivated to manipulate their social interactions and avoid the seller's personalized pricing. More specifically, in Fig. \ref{Ring_H_NoRevenueGain_b}, the high-preference buyer $j$ manipulates his social interaction with the upper high-preference buyer $k$ into a low interaction~frequency, aiming to mimic the low-preference buyer~$i$. In Fig. \ref{Ring_H_NoRevenueGain_c}, the two high-preference buyers $i$ and $j$ also manipulate to choose a low frequency towards the upper high-preference buyer $k$, while remaining honest when interacting with each other. By doing so, they resemble the social behaviors of the two low-preference buyers $i$ and $j$ in Fig. \ref{Ring_H_NoRevenueGain_a}. 

As a result, the common interactions between the (upper) known buyer and any (lower) unknown buyer turn~out to be the same (i.e., \mbox{$\hat{x}^*_{ki}=\hat{x}^*_{kj}=0$}) across different preference distributions from Fig. \ref{Ring_H_NoRevenueGain_a} to Fig. \ref{Ring_H_NoRevenueGain_c}. In other words, the prior knowledge of the high-preference buyer $k$ does not provide any information to the seller. This urges the seller to adopt uniform pricing for the unknown buyers, resulting in the same revenue as the two-buyer case in the absence of buyer $k$. 

\subsubsection{General Multi-buyer Network with Prior Knowledge}\label{subsec:general_buyer}

In the final part of this subsection, we return to the general social network of multiple interconnected unknown buyers. Remarkably, the conclusion that prior knowledge may not provide any value to the seller remains consistent in the~context of a more general network.

The following proposition provides the sufficient condition under which the seller does not benefit from introducing known buyers to an unknown buyer network, which generalizes Lemma \ref{High_known:no_benefit} beyond the three-buyer instance. Technically, we examine the scenario where an arbitrary number of known high-preference buyers are introduced to connect with $N$ unknown buyers in the online social network. 

\begin{proposition}\label{High_known:no_benefit:general}
If the following condition holds:
\begin{equation}\label{High_known:no_benefit:general_condition}
\frac{N}{N+1}v_{\text{\tiny H}}<v_{\text{\tiny L}}< v_{\text{\tiny H}}-K(1-l),
\end{equation}
where $K$ is the maximum number of known buyers that a single unknown buyer can be connected to, then there always exists a pure-strategy PBE such that the seller does not gain any expected revenue after learning from those known high-preference buyers.
\end{proposition}

The proof of Proposition \ref{High_known:no_benefit:general} is constructive (more details in Appendix \ref{Appendix:High_known:no_benefit:general}). Particularly, we construct a strategy profile wherein any unknown high-preference buyer manipulates his social interaction with any connected known buyer into a low frequency, even if they share the same high preference. This manipulation then renders the common social interactions between the known and unknown buyers appear identical, regardless of the underlying true preference distributions. Consequently, the prior knowledge of high-preference buyers does not provide any information to the seller, who thus adopts uniform pricing for all unknown buyers. The major idea of this construction echoes the illustration of Fig. \ref{Ring_H_NoRevenueGain} for the three-buyer instance in Section \ref{subsec:three-buyer}.

In conclusion, it is important to emphasize that learning previous buyers' preferences does not necessarily contribute to inferring the preferences of the other buyers in the seller's subsequent learning process.

\section{Numerical Study with Real-world Data}\label{Numerical}

After conducting extensive theoretical analyses in the preceding sections, we proceed to apply our proposed pricing mechanism to more practical scenarios, utilizing Facebook social network data sourced from \cite{snapnets}.

Specifically, in Section \ref{numerical1:data}, we introduce the real-world data to construct pairwise social interactions among buyers, building on our PBE analysis in Section \ref{PBE2}. In Section \ref{numerical2}, we evaluate the performance of different pricing mechanisms to confirm our previous theoretical analyses in Sections \ref{analysis} and~\ref{sec:multi}. Moreover, this empirical evaluation showcases the superiority of our strategic-learning-based pricing mechanism and its robustness against buyer manipulation.

\begin{figure*}
    \centering
	\subfigure[Subnetwork visualization]{
        \begin{minipage}{0.3\linewidth}
            \centering
            \includegraphics[width=1.1\linewidth]{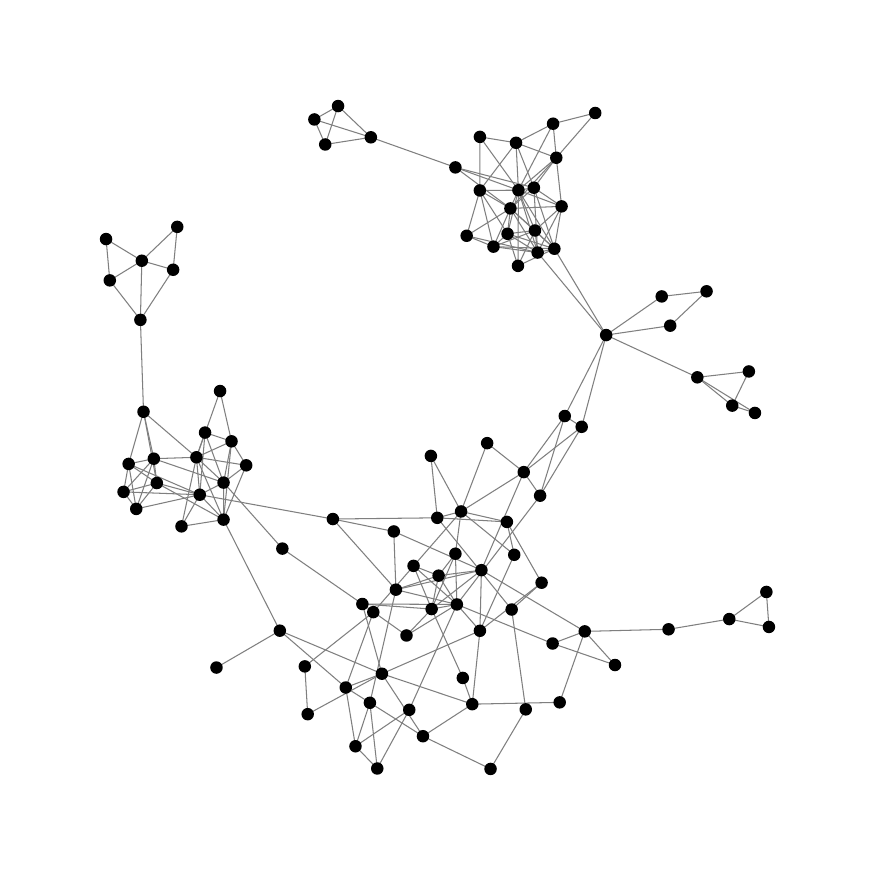}
            \vspace{-20pt}
            \label{fig_visualization}
		\end{minipage}
	}
     \subfigure[Truthful interaction data]{
        \begin{minipage}{0.3\linewidth}
            \centering
            \includegraphics[width=1.1\linewidth]{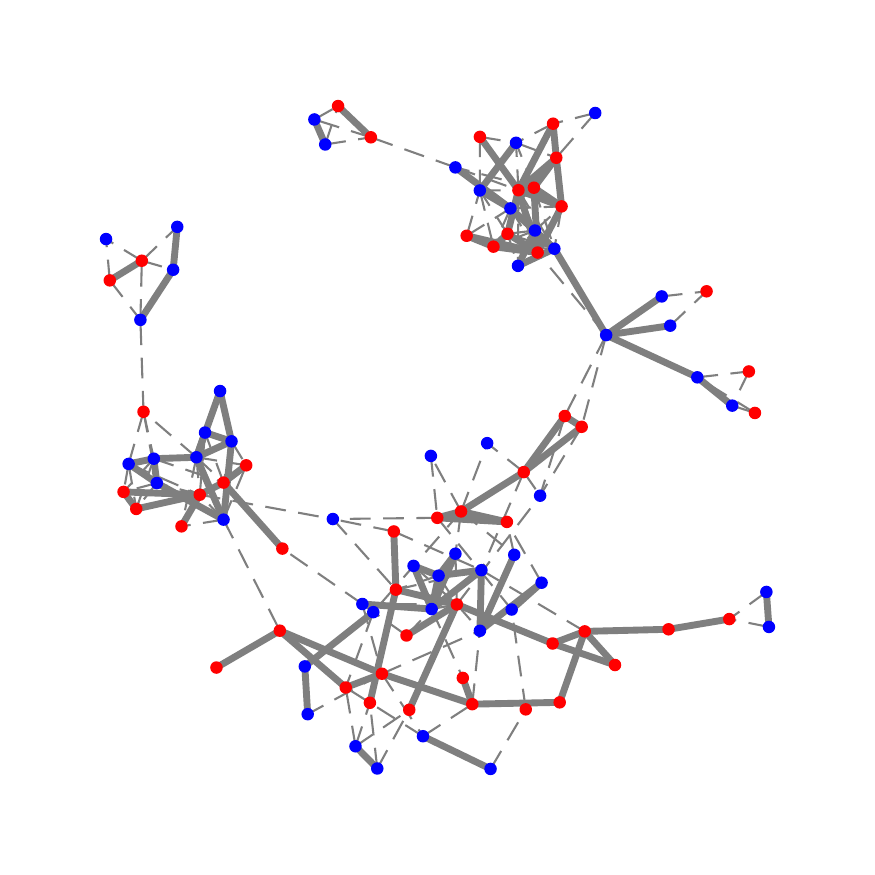}
            \vspace{-20pt}
            \label{fig_truthful_data}
		\end{minipage}
	}
	\subfigure[Manipulated interaction data]{
        \begin{minipage}{0.3\linewidth}
            \centering
            \includegraphics[width=1.1\linewidth]{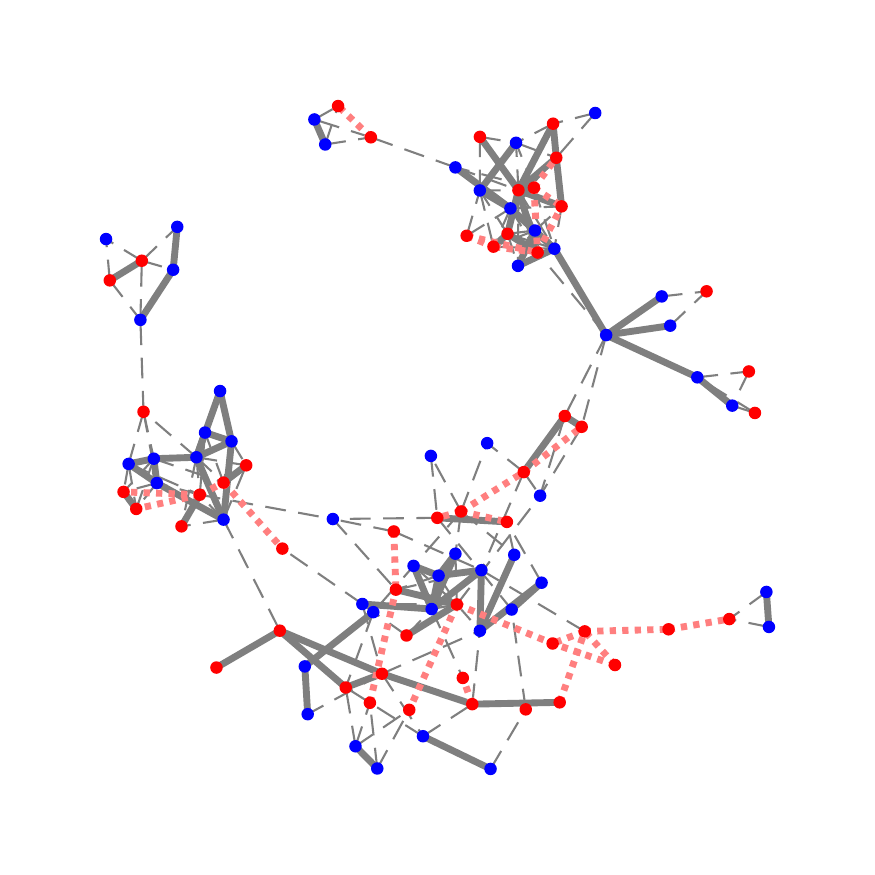}
            \vspace{-20pt}
        \label{fig_manipulated_data}
		\end{minipage}
	}
    \caption{(a) A subnetwork of the Facebook social network sourced from \cite{snapnets}. (b) The truthful pairwise common interaction frequency $\hat{x}$ among the subnetwork buyers, and (c) the manipulated pairwise common interaction frequency $\hat{x}^*$ among the subnetwork buyers in the PBE. In Fig. \ref{fig_truthful_data} and Fig. \ref{fig_manipulated_data}, the high-preference buyers are represented as red nodes and low-preference buyers as blue nodes, and the preference of each buyer is randomly generated according to a uniform prior distribution. In Fig. \ref{fig_manipulated_data}, we visualize the buyers' manipulated data when $v_{\text{\tiny H}}=3.8$ and $v_{\text{\tiny L}}=1.9$, and we set the social loss parameter to $l=0.5$.}
    \vspace{-12pt}
\end{figure*}

\subsection{Dataset Description}\label{numerical1:data}

To start, we generate synthetic data for buyers' social interactions based on the real-world Facebook social network sourced from \cite{snapnets}. Specifically, we focus on a typical induced subnetwork of the original Facebook network in \cite{snapnets}. The subnetwork consists of $100$ nodes and $230$ edges, with an average degree of $4.6$. Fig. \ref{fig_visualization} provides a visualization of this subnetwork, where each node represents an individual buyer and the edge captures the inter-buyer connections. 

Now, we explain the generation of buyers' truthful and manipulated social interaction data in Fig. \ref{fig_truthful_data} and Fig.~\ref{fig_manipulated_data}, respectively. Specifically, we randomly generate the preference of each subnetwork buyer, either high or low, following a uniform prior distribution. We represent the high-preference buyers as red nodes, while the low-preference buyers are depicted as blue nodes. Then, in Fig. \ref{fig_truthful_data}, we construct the truthful pairwise common interaction frequency $\hat{x}$ among the buyers based on inter-buyer preference correlations. In particular, the width of each edge in Fig. \ref{fig_truthful_data} represents the level of the common interaction frequency $\hat{x}_{ij}$ for any connected buyer pair $(i,j)$: thin dashed edges indicate a low level $0$, while thick solid edges represent a high level $1$. 

Moving forward to Fig. \ref{fig_manipulated_data}, we construct the manipulated interaction frequency $\hat{x}^*$ among buyers, following the PBE analysis of our two-buyer model in Propositions \ref{region1}-\ref{Region4}. In Fig. \ref{fig_manipulated_data}, we visualize the buyers' social interaction data when $v_{\text{\tiny H}}=3.8$ and $v_{\text{\tiny L}}=1.9$. Notably, we highlight those manipulated interactions as red thick dashed edges, indicating that they should truthfully exhibit a high level but have been manipulated to a lower frequency level.

\subsection{Performance Evaluation}\label{numerical2}

Next, we evaluate the performance of our proposed pricing mechanism based on the generated interaction data in Fig. \ref{fig_truthful_data} and Fig. \ref{fig_manipulated_data}. To achieve this, we employ the two benchmarks introduced in Section \ref{benchmark} and analyzed for the multi-buyer network setting in Section \ref{subsec:multi_buyer_benchmark}, and we also establish their connection to the existing literature (e.g., \cite{conitzer2012hide,2017Is,zhang2011perils,zhu_nearoptimal}). In addition, since obtaining the PBE for general multi-buyer networks is challenging with buyer manipulation, we will adapt our proposed pricing mechanism in the original two-buyer model to the experimental setting. Below, we outline the two benchmarks along with the pricing mechanism we proposed.

\begin{itemize}
    \item \textit{No-learning-based Pricing (NLP)} (e.g., \cite{conitzer2012hide,2017Is,zhang2011perils}): The NLP benchmark, introduced in Section \ref{bench_nolearning},~involves the seller's uniform pricing to all buyers. Its equilibrium analysis is provided in Lemma \ref{lemma:nolearning}. Variations of this benchmark are commonly adopted in the literature on behavior-based price discrimination.
    
    \item \textit{Undisclosed-learning-based Pricing (ULP)} (e.g., \cite{zhu_nearoptimal}): The ULP benchmark, introduced in Section \ref{bench_unaware}, has been analyzed in Proposition \ref{unaware_general} for the multi-buyer network setting. In this case, the seller exercises ULP based on the truthful data in Fig. \ref{fig_truthful_data} in the absence of buyer awareness. As such, the seller's revenue gain through ULP serves as an upper bound for the performance of our proposed pricing mechanism with buyer awareness. A pricing policy similar to ULP was recently analyzed for the first time in the literature on pricing in networks as in \cite{zhu_nearoptimal}.\footnote{In the manuscript \cite{zhu_nearoptimal}, Zhu et al. also incorporated the homophily effect within the social network and accounted for the dependency on consumer valuation to guide the seller's dynamic pricing. It is worth noting that Zhu et al. in \cite{zhu_nearoptimal} proposed approximation algorithms based on a different modeling framework, wherein the underlying network is formed through a branching process.}
    
    \item \textit{Our Strategic-learning-based Pricing (SLP)}: In our proposed SLP mechanism, the seller strategically learns from buyers' manipulated data in Fig. \ref{fig_manipulated_data} to determine the price offerings. While previous analyses in Section 8.3.3 have established key structural results for the seller's strategic learning in general multi-buyer networks, determining the pricing strategy in the PBE remains highly challenging in the presence of buyer manipulation (see Section \ref{subsec:general_buyer}). To address this complexity, we will adapt our SLP mechanism designed for the original two-buyer model (see Section \ref{PBE1}) to the experimental setting here.
\end{itemize}

\begin{figure}[h]
    \centering
    \vspace{-15pt}
    \includegraphics[width=0.4\linewidth]{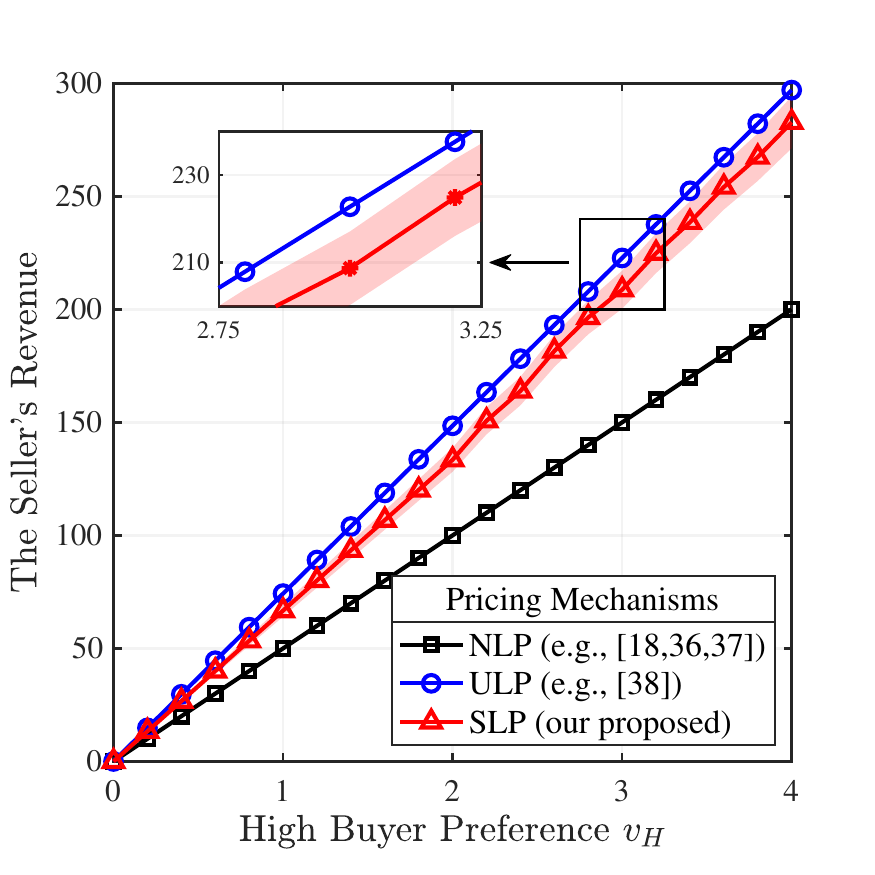}
    \vspace{-10pt}
    \caption{The seller's empirical revenue under the three pricing mechanisms (NLP, ULP, SLP) versus the high buyer preference $v_{\text{\tiny H}}$. In this experiment, we set the social loss parameter to $l=0.5$. We vary the high buyer preference $v_{\text{\tiny H}}$ within the range of $[0,4]$ while fixing the ratio between low and high buyer preferences $v_{\text{\tiny L}}/v_{\text{\tiny H}}$ as $1/2$.}
    \label{fig_seller_revenue_mani}
\end{figure}

Now, we compare the seller's empirical revenue obtained from the $100$ buyers within the Facebook subnetwork in Fig. \ref{fig_visualization} under the three pricing mechanisms. Specifically, we generate a random arrival sequence for these buyers in Stage II and perform the experiment $10,000$ times to obtain the seller's revenue results. In Fig. \ref{fig_seller_revenue_mani}, we depict the seller's average revenue under different pricing mechanisms: NLP (black curve), ULP (blue curve), and our proposed SLP (red curve); the transparent red patch represents the variance in the revenue distribution under our SLP.

Let us start with the scenario in the absence of buyers' awareness and potential manipulation. As shown in Fig. \ref{fig_seller_revenue_mani}, the ULP benchmark outperforms the NLP benchmark with a significant improvement of $48.5\%$ in this experiment, and this improvement closely aligns with the theoretical bound of $49.5\%$ in Proposition \ref{unaware_general}. Overall, the superiority of the ULP benchmark underscores the immense potential of incorporating buyers' social interaction data into the seller's learning and pricing process.

Finally, we shift our attention to the context of buyer awareness and specifically examine our proposed SLP. As shown in Fig. \ref{fig_seller_revenue_mani}, our SLP demonstrates a significant improvement over the NLP benchmark, ranging from $33.46\%$ to $41.37\%$. This superiority aligns with the insight from Proposition \ref{salerevenue1} derived under the two-buyer model. Furthermore, despite buyer manipulation, our SLP experiences only a minor reduction in revenue gain compared to the performance upper bound of ULP, ranging from $4.80\%$ to $10.13\%$ in this experiment. This finding highlights the robustness of our proposed SLP mechanism against buyer manipulation, consistent with the key insight from Proposition \ref{salerevenue2} derived under the two-buyer model.

\section{Extensions and Robustness}\label{extensions}

Thus far, we have focused on a uniform prior distribution of buyer preferences and also assumed a binary social interaction frequency setting. In this section, we extend our original two-buyer model in Section \ref{systemmodel} to showcase the robustness of our previous analysis methodology and major insights.

Specifically, we first introduce a non-uniform prior distribution on buyer preferences and investigate the implications of this non-uniformity in Section \ref{Extension1:non-uniform}. In Section~\ref{Extension2:continuous}, we further consider a continuous social interaction frequency setting. As we will show later, under a general functional form of the social interaction utility, this continuous-interaction-frequency model and its PBE analysis finally converge to our original binary frequency setting.

\subsection{Extension to Non-uniform Prior Distribution}\label{Extension1:non-uniform}

In this subsection, we expand our original two-buyer model in Section \ref{systemmodel} to introduce a non-uniform prior distribution on buyer preferences. Suppose that each buyer's preference $v_i$ follows a distribution with $\text{Pr}(v_{\text{\tiny H}})=\alpha$ and $\text{Pr}(v_{\text{\tiny L}})=1-\alpha$. In this extension, without much loss of generality, we focus on the case where $\alpha<0.5$, indicating a higher proportion of low-preference buyers in the product market. As we will illustrate later, this extension showcases the robustness of our previous analysis methodology and major insights. Moreover, we also explore the impact of the prior distribution on the equilibrium structure.

Similar to the analysis in Section \ref{PBE1} for our original model, we can easily confirm that the key structural property for the PBE under the uniform prior distribution remains valid here: \textit{social manipulations can only occur between two high-preference buyers}. Following this equilibrium space reduction, we proceed to alternate forward analysis with backward analysis to characterize the new PBE, and this analysis will follow a similar rationale in Section \ref{PBE2}.

Next, we delve into the new PBE outcome under the non-uniform prior distribution. To illustrate the key equilibrium structure, we classify the new PBE into five regions in Fig.~\ref{fig:newPBE_nonuniform}. 

\begin{figure}[h]
	\centering
    \vspace{-10pt}
	\centering
    \includegraphics[width=0.45\linewidth]{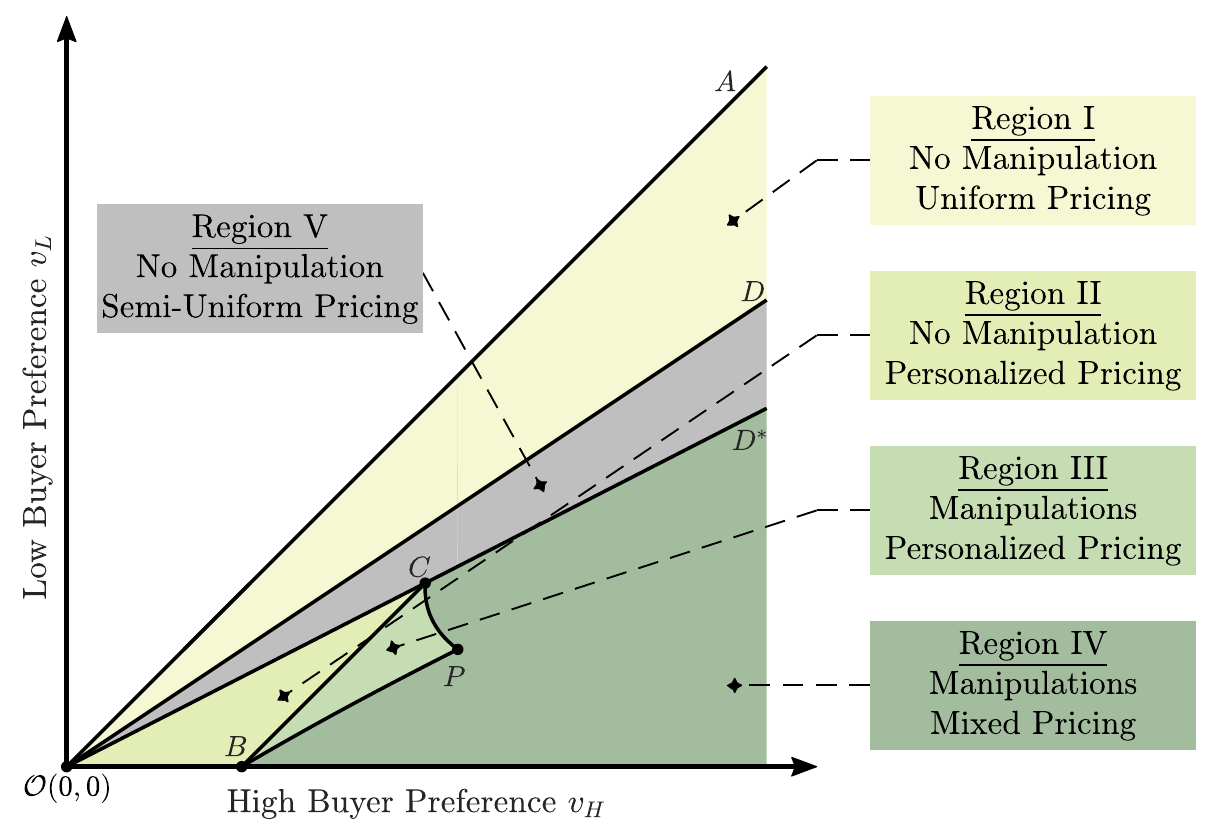}
    \vspace{-10pt}
	\caption{The new PBE under the non-uniform prior distribution with $\alpha<0.5$ comprises five structural regions. The PBE outcome is categorized into different preference regions considering high-preference buyers' manipulation strategy and the seller's pricing strategy in the equilibrium. The corresponding boundary functions can be found in Table A8 for any $\alpha\le 0.5$ in Appendix~\ref{Appendix:Extension1:non-uniform}.}
    \label{fig:newPBE_nonuniform}
    \vspace{-5pt}
\end{figure} 

As illustrated in Fig. \ref{fig:newPBE_nonuniform}, Regions I-IV of the PBE under a non-uniform prior distribution with $\alpha<0.5$ exhibit similar structures and properties with their counterparts in Fig. \ref{regions} under the uniform prior distribution with $\alpha=0.5$. However, it is important to notice that a distinct equilibrium case emerges in Region V of Fig. \ref{fig:newPBE_nonuniform} under $\alpha<0.5$, as outlined in the following proposition.

\begin{proposition}\label{non-uniform:newRegion}
In the strategic-learning model, a new equilibrium case emerges under the non-uniform prior distribution with $\alpha<0.5$, compared to the PBE under a uniform prior distribution with $\alpha=0.5$. Specifically, for $\alpha<0.5$, if the relative difference between the low and high buyer preferences is moderate within the range of $2\alpha^2/(3\alpha^2-2\alpha+1)< v_{\text{\tiny L}}/v_{\text{\tiny H}}<2/3$ (see Region V of Fig.  \ref{fig:newPBE_nonuniform}), there exists a unique PBE as follows.
\begin{itemize}
    \item \emph{In Stage I}, buyers $i$ and $j$ with the same high preference choose the maximum social interaction frequency in \eqref{region1_buyer} without engaging in any manipulation.
    \item \emph{In Stage II}, when observing buyers' low common interaction frequency $\hat{x}^*=0$, the seller charges a high price \mbox{$p_1^*=v_{\text{\tiny H}}$} in the first selling period and enables personalized pricing in the second selling period; whereas observing $\hat{x}^*=1$, the seller charges a low uniform price $p_1^*=p_2^*=v_{\text{\tiny L}}$ in each period, i.e.,
    \begin{equation*}
    \left\{
    \begin{aligned}
    &p_1^*=p_2^*=v_{\text{\tiny L}},&\textrm{if $\ \hat{x}^\ast=1$,}\label{new:non:lowuniform}\\
    &p_1^*=v_{\text{\tiny H}}, \eqref{condition_2},&\textrm{if $\ \hat{x}^\ast=0$.} 
    \end{aligned}
    \right.
    \end{equation*}
\end{itemize}
\end{proposition}

Proposition \ref{non-uniform:newRegion} shows that in the newly-emerged equilibrium case within Region V of Fig. \ref{fig:newPBE_nonuniform}, the seller strategically sets a low uniform price when observing $\hat{x}^*=1$ to cater to the demand of low-preference buyers. Specifically, this pricing choice is driven by a higher proportion of low-preference buyers in the product market when $\alpha<0.5$. As a result, when engaging in honest social interactions with $\hat{x}=1$, high-preference buyers can always obtain the maximum purchase surplus of $v_{\text{\tiny H}}-v_{\text{\tiny L}}$ in Stage II, eliminating their incentives to manipulate. Notice that, as the prior distribution approaches uniformity with $\alpha\to 0.5^-$, this new equilibrium case diminishes, and the structure in Fig. \ref{fig:newPBE_nonuniform} degenerates into Fig. \ref{regions} under the uniform prior distribution.

\begin{corollary}\label{non-uniform:comparison}
Under a non-uniform prior distribution with $\alpha<0.5$, the equilibrium manipulation probability $\rho^*$ for high-preference buyers in Stage I is lower than the original PBE under a uniform prior with $\alpha=0.5$.
\end{corollary}

Corollary \ref{non-uniform:comparison} suggests the intuition that high-preference buyers are less likely to manipulate when there are more low-preference buyers in the product market. This finding aligns with the intuition behind Proposition \ref{non-uniform:newRegion} above.

\subsection{Robustness to Continuous Interaction Frequency}\label{Extension2:continuous}

Next, we extend our original two-buyer model in Section~\ref{systemmodel} to incorporate continuous social interaction frequency. In this extension, we consider a general functional form for the buyer's social interaction utility, and our new analysis will demonstrate the robustness of our previous analysis methodology and major results.

To start, we revisit the original model formulation for buyers' social interactions in Section \ref{model_stageI} and modify it to accommodate a continuous interaction frequency setting. Specifically, we first redefine the strategy space for each buyer's social interaction frequency with the other buyer as the normalized set $[0,1]$, i.e., $x_{ij}\in[0,1]$ and $x_{ji}\in[0,1]$ for the two buyers $i$ and $j$, respectively. Then, we reformulate the social interaction utility $u_i(x_{ij},x_{ji})$ for each buyer $i$ in a general functional form, as detailed in the following definition.

\begin{definition}\label{functional_form_social_utility}
The social interaction utility $u_i(x_{ij},x_{ji})$ for each buyer $i$ when interacting with the other buyer $j$ is defined as follows.

(i) When the two buyers $i$ and $j$ have the same preference, i.e., $v_i=v_j$, the social interaction utility function $u_i(x_{ij},x_{ji})$ satisfies the following three conditions:
    \begin{itemize}
        \item \emph{Non-negative:} $u_i(x_{ij},x_{ji})\ge 0$;
        \item \emph{Monotonic:} $u_i(x_{ij},x_{ji})$ increases in both $x_{ij}$ and $x_{ji}$;
        \item \emph{Normalized:} $u_i(1,1)=1$, $u_i(0,0)=0$, and \begin{equation}\label{loss_utility}
        u_i(x_{ij},1)-u_i(x_{ij},0)=l,\ \forall x_{ij}\in\{0,1\},
        \end{equation}
        i.e., when both buyers have the same preference, the social interaction loss incurred by one buyer due to the other buyer's lack of social engagement is denoted by the exogenous parameter $l\in(0,1)$.
    \end{itemize}

(ii) When the two buyers $i$ and $j$ have different preferences, i.e., $v_i\neq v_j$, the social interaction utility function $u_i(x_{ij},x_{ji})$ satisfies the following two conditions:\footnote{Notice that we can further impose a normalized condition on the social interaction utility function $u_i(x_{ij},x_{ji})$ when the two buyers have different preferences. Specifically, this involves $u_i(1,1)=-c$, $u_i(0,0)=0$, and $u_i(x_{ij},0)-u_i(x_{ij},1)=r$ for any $x_{ij}\in\{0,1\}$ with $r\in(0,c)$. However, while this formulation will correspond to Table \ref{table_differ} under the binary frequency setting, it is important to notice that this normalization condition is not necessary to ensure the robustness of our previous major results derived under the binary setting. Therefore, we omit it in Definition \ref{functional_form_social_utility} to maintain the generality of this new extension.}
\begin{itemize}
    \item \emph{Non-positive:} $u_i(x_{ij},x_{ji})\le 0$;
    \item \emph{Monotonic:} $u_i(x_{ij},x_{ji})$ decreases in both $x_{ij}$ and $x_{ji}$.
\end{itemize}
\end{definition}

Definition \ref{functional_form_social_utility} allows the social interaction utility $u_i(x_{ij},x_{ji})$ for each buyer $i$, when buyers $i$ and $j$ have the same preference $v_i=v_j$, to be an arbitrary non-negative and increasing function satisfying certain normalized conditions. Specifically, more frequent interactions between buyers with a common preference contribute to enhanced social happiness by fostering empathy and strengthening connections \cite{seiter2015secret}. Similar to our original model, we normalize the maximum social interaction value to $1$ and the minimum value to $0$. Moreover, in order to account for the social loss incurred by one buyer due to the other buyer's lack of engagement when their preferences align, we introduce the exogenous parameter $l\in (0,1)$ in \eqref{loss_utility} as in Table \ref{table_same} under binary interaction frequencies.

On the other hand, when the two buyers hold different preferences $v_i \neq v_j$, Definition~\ref{functional_form_social_utility} allows the social interaction utility $u_i(x_{ij},x_{ji})$ for each buyer $i$ to be any non-positive and decreasing function. This disutility formulation helps capture buyer $i$'s experienced embarrassment in maintaining social interactions with individuals who have different opinions \cite{mcpherson2001birds}. Similar to the analysis in Section \ref{PBE1}, we will start with an efficient reduction of the large equilibrium space to guide the forward analysis for the seller's strategic learning in this new extension. In particular, we can easily show that Lemmas \ref{lemma1_buyer} and \ref{pricebinary} continue to hold here: \textit{social manipulations can only occur between two high-preference buyers}.

We now focus on the remaining case wherein both buyers have the same high preference, i.e., $v_i=v_j=v_{\text{\tiny H}}$. In this new extension, with the introduction of a continuous strategy space for each buyer's social interaction decision, manipulation towards a frequency between the maximum and minimum values becomes possible. However, as suggested in the following lemma, we can establish the existence of polarization in the choices of equilibrium interaction frequency for those high-preference buyers.

\begin{lemma}\label{continuous_frequency:binary_manipulation}
In the continuous-interaction-frequency model, the PBE dictates that a high-preference buyer $i$, when interacting with another high-preference buyer $j$ in Stage I, will choose a social interaction frequency of either the maximum value of $1$ or the minimum of $0$, rather than any intermediate frequencies, i.e., $x_{ij}^*\in\{1,0\}$.
\end{lemma}

Lemma \ref{continuous_frequency:binary_manipulation} narrows our attention to binary social interaction frequencies $\{1,0\}$ for each high-preference buyer in Stage I. Specifically, to conceal their preference correlation, the two high-preference buyers must manipulate their common interaction frequency to $\hat{x}=0$, to mimic buyers with different preferences. In fact, any $\hat{x}$ within the range of $(0,1)$ would unmistakably signal to the seller that manipulation has taken place between the two buyers, thereby revealing their underlying high preferences. Therefore, manipulating the social interaction frequency to any level in between $(0,1)$ does not help improve a buyer's purchase surplus in Stage II but reduces his social utility (see Definition \ref{functional_form_social_utility}). 

Until now, we have demonstrated that the structural properties of the PBE in our original model remain the same in this more general continuous frequency setting. This robustness guarantees the validity of all subsequent PBE analyses under our original binary frequency setting in Sections \ref{PBE2} and \ref{analysis}. As a final remark, under a general functional form of the social interaction utility, the continuous-interaction-frequency model and its PBE analysis converge to the binary frequency setting.

\section{Conclusion}\label{conclusion}

This paper presents the first analytical study regarding how buyers may strategically manipulate their social interaction signals considering their preference correlations, and how a seller can take buyers' strategic social behaviors into consideration when designing the pricing scheme. One distinct feature of our study is the consideration of a double-layered information asymmetry between the seller and buyers, integrating both individual buyer information and inter-buyer correlation information.

Our analysis reveals that only high-preference buyers tend to manipulate their social interactions to evade personalized pricing, but surprisingly, their payoffs may actually worsen as a result. In addition, we show that the seller can considerably profit from the learning practice regardless of buyers' awareness. Indeed, buyers' learning-aware manipulation only slightly reduces the seller's revenue. In light of the increasingly strict privacy regulations, it is advisable for sellers to keep transparent with buyers about their social data access and follow-up learning practices. This finding is in line with current informed-consent industry practices for data sharing. Finally, we investigate the seller's dynamic learning across multiple interconnected buyers, and reveal that learning previous buyers' preferences does not necessarily contribute to inferring the other buyers in the seller's subsequent learning process.

There are some intriguing directions to study in the future. For example, our analysis remains largely silent about the intertemporal effect of prior buyers' learning-aware data manipulations, which is an important direction for future research. Another interesting issue we have left out in the current study is the impact of network structure, which may require an alternative model with further approximations. Furthermore, delving into the dynamics of buyers' information disclosure through repeated interactions between the seller and buyers presents another promising avenue for exploration.

\section*{Acknowledgments}
This work was supported by the National Natural Science Foundation of China (Project $\text{62271434}$), Shenzhen Science and Technology Innovation Program (Project $\text{JCYJ20210324120011032}$), Guangdong Basic and Applied Basic Research Foundation (Project $\text{2021B1515120008}$), Shenzhen Key Lab of Crowd Intelligence Empowered Low-Carbon Energy Network (No. $\text{ZDSYS20220606100601002}$), Shenzhen Stability Science Program 2023, and the Shenzhen Institute of Artificial Intelligence and Robotics for Society. This work is also supported in part by the Ministry of Education, Singapore, under its Academic Research Fund Tier 2 Grant with Award no. MOE-T2EP20121-0001; in part by SUTD Kickstarter Initiative (SKI) Grant with no. SKI 2021\_04\_07; and in part by the Joint SMU-SUTD Grant with no. 22-LKCSB-SMU-053. 

\bibliographystyle{IEEEtran}
\bibliography{sample-base}

\onecolumn
\appendices

\setcounter{fact}{0}
\setcounter{lemma}{0}
\setcounter{claim}{0}
\setcounter{table}{0}
\setcounter{figure}{0}
\setcounter{theorem}{0}
\setcounter{equation}{0}
\setcounter{corollary}{0}
\setcounter{definition}{0}
\setcounter{proposition}{0}

\renewcommand{\thefact}{A\arabic{fact}}
\renewcommand{\thelemma}{A\arabic{lemma}}
\renewcommand{\theclaim}{A\arabic{claim}}
\renewcommand{\thetable}{A\arabic{table}}
\renewcommand{\thefigure}{A\arabic{figure}}
\renewcommand{\theequation}{A\arabic{equation}}
\renewcommand{\thecorollary}{A\arabic{corollary}}
\renewcommand{\thedefinition}{A\arabic{definition}}
\renewcommand{\theproposition}{A\arabic{proposition}}

\section{Personalized Pricing and Its Practical Integration with Social Network Data}\label{Appendix:Evidence}

To strengthen the practical importance of our work, we delve into Amazon's practice of social data access and data-driven pricing, which exemplifies an important case study.
Specifically, Amazon has heavily relied on dynamic pricing, leveraging the power of big data \cite{wakabayashi_2022_does}. Today, by integrating Amazon profiles with Facebook accounts, buyers grant Amazon access to a wealth of social network data to enhance their personalized shopping experiences \cite{-_by_-_boulton_2021}. 

This integration enables Amazon to acquire and employ social network data to train its personalization features \cite{bustos_2011_facebook,schulze_2018_facebook}. 
This process of learning-enabled personalization corresponds to Stage II of our system model, as outlined in Section 3.2. Furthermore, it is worth noting that price discrimination is commonly practiced in various consumer markets through personalized coupons, discounts, and fees \cite{dube2017scalable}. Another related strategy is product steering, which involves \mbox{changing} the order of search results when different buyers search for the same products online, with the aim of directing high-valuation buyers towards higher-priced products \cite{hannak2014measuring}. Notably, our model of the seller's personalized pricing in Section 3.2 captures these practices (in a potentially implicit way).

In a broader context, data-driven price discrimination is also prevalent in the consumer credit market worldwide, particularly in China. Within the credit domain, price discrimination involves offering varying payment conditions based on the bank's evaluation of the borrower's repayment likelihood \cite{garcia2023algorithmic}. Notably, Ant Financial, a division of Alibaba, generates the Zhima (Sesame) Credit score by considering multiple factors beyond each individual's financial history, including interpersonal relationships on social media \cite{bonatti2020consumer}. Moreover, fintech companies also perform credit evaluation by incorporating
various types of information, such as individuals' behavior data in social media \cite{garcia2023algorithmic}.

\section{Proof of Lemma \ref{lemma:nolearning}}

In the no-learning benchmark, where the seller cannot learn from buyers' social interactions, Stages I and II of the dynamic game are then decoupled. As such, no buyer has the incentive to manipulate, and each buyer $i$ decides his social interaction frequency $x_{ij}$ only to maximize the social interaction utility $u_i(x_{ij},x_{ji})$ in Stage I. Therefore, the buyers' equilibrium social decisions could be attained through the normal-form game-theoretical analysis according to Tables \ref{table_same} and \ref{table_differ}, depending on whether their preferences are the same or not. Since the seller cannot relate one buyer's preference to the other's in this benchmark, the two selling periods in Stage II are decoupled as well. Therefore, the seller charges the same uniform price in both selling periods as
\begin{equation}
\begin{aligned}
p_1^*=p_2^*=&\arg\max_{p_t}\ p_t\cdot\textrm{Pr}(v_t\ge p_t),\ \forall t\in\{1,2\}\\=&\arg\max_{p_t\in\{v_{\text{\tiny L}},v_{\text{\tiny H}}\}} \left[v_{\text{\tiny L}}\boldsymbol{1}(p_t=v_{\text{\tiny L}})+\frac{v_{\text{\tiny H}}}{2}\boldsymbol{1}(p_t=v_{\text{\tiny H}})\right],\ \forall t\in\{1,2\}\\=&\left\{
\begin{aligned}
v_{\text{\tiny H}}, &\quad\textrm{if $\ 2v_{\text{\tiny L}} \le v_{\text{\tiny H}}$,}\\
v_{\text{\tiny L}}, &\quad\textrm{otherwise;}
\end{aligned}
\right.
\end{aligned}
\end{equation}
where the third equality with binary optimization variables follows from Lemma \ref{pricebinary}. Notice that Lemma \ref{pricebinary}, which restricts our attention to the seller's binary pricing choices, remains valid in this no-learning benchmark.

\section{Proof o f Lemma \ref{unaware}}

This proof consists of two main steps. Specifically, we first analyze the users' social interactions in Stage I in Appendix \ref{Appendix:stepI_lemma2} and then the seller's pricing in Stage II in Appendix \ref{Appendix:stepII_lemma2}. Notice that Lemma \ref{pricebinary} of the seller's binary pricing choices remains valid in this benchmark.
\subsection{Step 1: Analysis of the buyers' social decisions in Stage I}\label{Appendix:stepI_lemma2}
In the undisclosed-learning benchmark, the buyers are not aware of the seller's follow-up learning; thus, they do not take into account the implications of their social interactions in Stage I for the seller's pricing in Stage II. Therefore, the buyers' equilibrium social decisions in Stage I remain the same as those in the no-learning benchmark.
\subsection{Step 2: Analysis of the seller's pricing in Stage II}\label{Appendix:stepII_lemma2}
The analysis of the seller's pricing in Stage II is conducted through backward induction. In what follows, we develop our discussions according to the buyers' common interaction frequency $\hat{x}^*$ in \eqref{hat} observed by the seller in the equilibrium.

\subsubsection{Upon observing $\hat{x}^*=1$}
\textbf{(i) Posterior belief update.} When observing the buyers' high common interaction frequency $\hat{x}^*=1$ in Stage I, the seller's posterior belief of the buyers' preferences is given by \mbox{${\rm{Pr}}(v_1=v_2=v_{\text{\tiny H}}|\hat{x}^*=1)={\rm{Pr}}(v_1=v_2=v_{\text{\tiny L}}|\hat{x}^*=1)=1/2$}, according to the Bayes' theorem. Moreover, the two buyers are with the same preference with probability one conditioned on $\hat{x}^*=1$, i.e., 
${\rm{Pr}}(v_1=v_2|\hat{x}^*=1)=1$. \textbf{(ii) Second-period pricing analysis.} Then, we start with the pricing analysis for the second selling period in Stage II. (a) When setting the first-period price as $p_1=v_{\text{\tiny L}}$, the seller cannot learn the first-arriving buyer's preference; thus, the optimal second-period pricing scheme is simply given by $p_2=\arg\max p_2\cdot\textrm{Pr}(v_2\ge p_2)$, i.e., $p_2=v_{\text{\tiny H}}$ if $v_{\text{\tiny L}}/v_{\text{\tiny H}}\le 1/2$ and $p_2=v_{\text{\tiny L}}$ otherwise. In this way, the seller's expected second-period revenue is $\max\left\{v_{\text{\tiny H}}/2,v_{\text{\tiny L}}\right\}$. (b) When setting the first-period price as $p_1=v_{\text{\tiny H}}$, the seller charges $p_2=v_{\text{\tiny L}}$ in the second period if the first-arriving buyer does not purchase $a_1=0$ (i.e., $v_1=v_{\text{\tiny L}}$); then the corresponding second-period revenue is $v_{\text{\tiny L}}$. But if the first-arriving buyer purchases $a_1=1$ at $p_1=v_{\text{\tiny H}}$ (i.e., $v_1=v_{\text{\tiny H}}$), the second-period price is $p_2=v_{\text{\tiny H}}$ and the corresponding second-period revenue is $v_{\text{\tiny H}}$. \textbf{(iii) First-period pricing optimization.} Now, we are ready for the optimization problem of the seller's first-period pricing:
\begin{equation}
p_1^*=\arg\max_{p_1\in\{v_{\text{\tiny H}},v_{\text{\tiny L}}\}} \left[\left(\frac{1}{2}\left(0+v_{\text{\tiny L}}\right)+\frac{1}{2}\left(v_{\text{\tiny H}}+v_{\text{\tiny H}}\right)\right)\boldsymbol{1}(p_1=v_{\text{\tiny H}})+\left(v_{\text{\tiny L}}+\max\left\{\frac{v_{\text{\tiny H}}}{2},v_{\text{\tiny L}}\right\}\right)\boldsymbol{1}(p_1=v_{\text{\tiny L}})\right],
\end{equation}
which yields the seller's optimal first-period price $p_1=v_{\text{\tiny L}}$ if ${v_{\text{\tiny L}}}/v_{\text{\tiny H}}\ge 2/3$ and $p_1=v_{\text{\tiny H}}$ otherwise.

\subsubsection{Upon observing $\hat{x}^*=0$}
\textbf{(i) Posterior belief update.} When observing the low common interaction frequency $\hat{x}^*=0$ in Stage I, the seller's posterior belief of buyers' preferences is given by \mbox{${\rm{Pr}}(v_1=v_{\text{\tiny H}},v_2=v_{\text{\tiny L}}|\hat{x}^*=0)={\rm{Pr}}(v_1=v_{\text{\tiny L}},v_2=v_{\text{\tiny H}}|\hat{x}^*=0)=1/2$}, according to the Bayes' theorem. Moreover, the two buyers have different preferences with probability one conditioned on $\hat{x}^*=0$, i.e., 
${\rm{Pr}}(v_1\neq v_2|\hat{x}^*=0)=1$. \textbf{(ii) Second-period pricing analysis.} Then, we start with the pricing analysis for the second selling period in Stage II. (a) When setting the first-period price as $p_1=v_{\text{\tiny L}}$, the seller cannot learn the first-arriving buyer's preference; thus, the optimal second-period pricing scheme is simply given by $p_2=\arg\max p_2\cdot\textrm{Pr}(v_2\ge p_2)$, i.e., $p_2=v_{\text{\tiny H}}$ if $v_{\text{\tiny L}}/v_{\text{\tiny H}}\le 1/2$ and $p_2=v_{\text{\tiny L}}$ otherwise. In this way, the seller's expected second-period revenue is $\max\left\{v_{\text{\tiny H}}/2,v_{\text{\tiny L}}\right\}$. (b) When setting the first-period price as $p_1=v_{\text{\tiny H}}$, the seller charges $p_2=v_{\text{\tiny H}}$ in the second period if the first-arriving buyer does not purchase $a_1=0$ (i.e., $v_1=v_{\text{\tiny L}}$); then the corresponding second-period revenue is $v_{\text{\tiny H}}$. But if the first-arriving buyer purchases $a_1=1$ at $p_1=v_{\text{\tiny H}}$ (i.e., $v_1=v_{\text{\tiny H}}$), the second-period price is $p_2=v_{\text{\tiny L}}$ and the corresponding second-period revenue is $v_{\text{\tiny L}}$. \textbf{(iii) First-period pricing optimization.} Now, we are ready for the optimization problem of the seller's first-period pricing:
\begin{equation}
p_1^*=\arg\max_{p_1\in\{v_{\text{\tiny H}},v_{\text{\tiny L}}\}} \left[\left(\frac{1}{2}\left(0+v_{\text{\tiny H}}\right)+\frac{1}{2}\left(v_{\text{\tiny H}}+v_{\text{\tiny L}}\right)\right)\boldsymbol{1}(p_1=v_{\text{\tiny H}})+\left(v_{\text{\tiny L}}+\max\left\{\frac{v_{\text{\tiny H}}}{2},v_{\text{\tiny L}}\right\}\right)\boldsymbol{1}(p_1=v_{\text{\tiny L}})\right],
\end{equation}
which yields the seller's optimal first-period price $p_1=v_{\text{\tiny L}}$ if ${v_{\text{\tiny L}}}/v_{\text{\tiny H}}\ge 2/3$ and $p_1=v_{\text{\tiny H}}$ otherwise.

Combining the above analysis together completes the proof for Lemma \ref{unaware}.

\section{Proof of Corollary \ref{salerevenue_un}}
To conduct comparative analyses between the two benchmarks in Section \ref{benchmark}, we first calculate the seller's expected revenue $\tilde{\Pi}^0$ in the no-learning benchmark and $\tilde{\Pi}^1$ in the undisclosed-learning benchmark:
\begin{equation}\label{PI_0_No}
\tilde{\Pi}^0=\left\{
\begin{aligned}
v_{\text{\tiny H}}, &\quad\textrm{if $\ v_{\text{\tiny L}}/v_{\text{\tiny H}}< 1/2$,}\\
2v_{\text{\tiny L}}, &\quad\textrm{otherwise;}
\end{aligned}
\right.
\quad\text{and}
\quad\tilde{\Pi}^1=\left\{
\begin{aligned}
&v_{\text{\tiny H}}+v_{\text{\tiny L}}/2, &\quad\textrm{if $\ {v_{\text{\tiny L}}}/v_{\text{\tiny H}}< 2/3$,}\\
&2{v_{\text{\tiny L}}}, &\quad\textrm{otherwise.}
\end{aligned}
\right.
\end{equation}
Now we compare $\tilde{\Pi}^0$ with $\tilde{\Pi}^1$, and the expected revenue gain from the seller's undisclosed learning practice is given by
\begin{equation}
\frac{\tilde{\Pi}^1-\tilde{\Pi}^0}{\tilde{\Pi}^0}=\left\{
\begin{aligned}
&0, &\quad\textrm{if $\ {v_{\text{\tiny L}}}/v_{\text{\tiny H}}\ge 2/3$,}\\
&v_{\text{\tiny H}}/2v_{\text{\tiny L}}-3/4, &\quad\textrm{if $\ 1/2\le{v_{\text{\tiny L}}}/v_{\text{\tiny H}}<2/3$,}\\
&v_{\text{\tiny L}}/2v_{\text{\tiny H}}, &\quad\textrm{otherwise;}
\end{aligned}
\right.
\end{equation}
which is always non-negative and strictly positive as long as ${v_{\text{\tiny L}}}/v_{\text{\tiny H}}<2/3$; its maximum $25\%$ is attained when ${v_{\text{\tiny L}}}/v_{\text{\tiny H}}=1/2$.

\section{Proof of Lemma \ref{lemma1_buyer}}\label{proof_lemma_buyer}
The proof applies Lemma \ref{pricebinary}, and the analysis rationale is demonstrated in the paragraph after Lemma \ref{lemma1_buyer} in Section \ref{PBE1}.

\section{Proof of Lemma \ref{pricebinary}}\label{proof_binary}

This proof is developed by contradiction. First suppose that the seller charges $p_t<v_{\text{\tiny L}}$ for each selling period $t\in\{1,2\}$ in the equilibrium. However, the seller always obtains a strictly higher revenue by increasing the price $p_t$ to $v_{\text{\tiny L}}$, which does not affect the arriving buyer's purchase decision $a_t$. This then results in a contradiction. Next, suppose that the seller charges $p_t>v_{\text{\tiny H}}$ in the equilibrium; in this case, no buyer purchases. But the seller becomes strictly better off by setting $p_t=v_{\text{\tiny H}}$, again a contradiction. Finally, we assume on the contrary that there exists an equilibrium where the seller charges strictly between low and high preferences, i.e., $v_{\text{\tiny L}}<p_t<v_{\text{\tiny H}}$. Increasing the price $p_t$ to $v_{\text{\tiny H}}$ does not affect the buyers' purchase decisions but strictly increases the seller's revenue. This contradicts the sequential rationality criterion again.

\section{Proof of Propositions \ref{region1}-\ref{Region4}}\label{Appendix:PBE}
This appendix completes the characterization of all PBE for the strategic-learning model, and the proof will follow the two steps introduced in Section \ref{PBE2}. In particular, Proposition \ref{region1} and Proposition \ref{Region2} are reached by Case 1 and Case 2 in Appendix \ref{part1}, respectively; in Appendix \ref{part3}, Proposition \ref{Region3} is reached by Case 1, while Cases 2 and 3 together lead to Proposition \ref{Region4}.

\subsection{Step 1: Analysis of the seller's pricing in Stage II}
We first apply backward induction to analyze the seller's optimal pricing schemes in Stage II. To start with, we now update the seller's posterior belief of the buyers' preferences in (\ref{belief1}) and (\ref{belief0}).

Let $s$ denote the conditional probability of the event that buyer $i$ and $j$ make the common interaction frequency $\hat{x}\triangleq\min\{x_{ij},x_{ji}\}=1$ given that they are with the same high preference $v_i=v_j=v_{\text{\tiny H}}$. i.e.,
\begin{align}\label{rho01_condbelief:app}
    s\triangleq\textrm{Pr}\left(\hat{x}=1|v_i=v_j=v_{\text{\tiny H}}\right)=(1-\rho_i)^2=(1-\rho)^2.
\end{align}
Substitute (\ref{rho01_condbelief:app}) into the seller's posterior beliefs in (\ref{belief1}) and (\ref{belief0}), we have
\begin{equation}\label{rho01_postbelief1}
{\rm{Pr}}(v_i=v_j=v_{\text{\tiny H}}|\hat{x}=1)=\frac{s}{s+1} {\ \rm and\ }{\rm{Pr}}(v_i=v_j=v_{\text{\tiny L}}|\hat{x}=1)=\frac{1}{s+1},
\end{equation}
and
\begin{align}\label{rho01_postbelief2}
{\rm{Pr}}(v_i=v_j=v_{\text{\tiny H}}|\hat{x}=0)=\frac{1-s}{3-s} {\ \rm and\ }{\rm{Pr}}(v_i\neq v_j|\hat{x}=0)=\frac{2}{3-s}.
\end{align}



Based on the posterior beliefs in (\ref{rho01_postbelief1}) and (\ref{rho01_postbelief2}), the seller's optimal pricing in Stage II is provided in the following lemmas. 

\begin{lemma}\label{rho01_lemma1}
Given the posterior beliefs in (\ref{rho01_postbelief1}), when observing buyers' high common interaction frequency $\hat{x}=1$, the seller's optimal pricing in Stage II depends on the preference diversity $v_{\text{\tiny L}}/v_{\text{\tiny H}}$ as follows.
\begin{itemize}
    \item \textbf{Case A}: If ${2s}/{(1+2s)}\le v_{\text{\tiny L}}/v_{\text{\tiny H}}<1$, the seller charges a uniform price $p_1=p_2=v_{\text{\tiny L}}$;
    \item \textbf{Case B}: If $0<v_{\text{\tiny L}}/v_{\text{\tiny H}}\le {2s}/{(1+2s)}$, the seller charges a high price $p_1=v_{\text{\tiny H}}$ in the first selling period. The second-period price personalizes, depending on the purchase record $a_1\in\{0,1\}$ of the first period in Stage II:
        \begin{equation}
            p_2(a_1)=v_{\text{\tiny H}}\boldsymbol{1}(a_1=1)+v_{\text{\tiny L}}\boldsymbol{1}(a_1=0).
        \end{equation}
\end{itemize}
\end{lemma}

\begin{proof}
    
When observing buyers' high common interaction frequency $\hat{x}=1$, based on the posterior belief in (\ref{rho01_postbelief1}), the seller updates that 
\begin{equation}
{\rm Pr}(v_1=v_{\text{\tiny H}}|\hat{x}=1)={\rm Pr}(a_1=1|\hat{x}=1;p_1=v_{\text{\tiny H}})={\rm Pr}(v_2=v_{\text{\tiny H}}|\hat{x}=1;p_1=v_{\text{\tiny L}})=\frac{s}{s+1},
\end{equation}
and
\begin{equation}
{\rm Pr}(v_1=v_{\text{\tiny L}}|\hat{x}=1)={\rm Pr}(a_1=0|\hat{x}=1;p_1=v_{\text{\tiny H}})={\rm Pr}(v_2=v_{\text{\tiny L}}|\hat{x}=1;p_1=v_{\text{\tiny L}})=\frac{1}{s+1}.
\end{equation}
Specifically, we first analyze the second selling period. (a) When setting the first-period price $p_1=v_{\text{\tiny L}}$, the seller's optimal second-period price is $p_2=v_{\text{\tiny H}}$ if $v_{\text{\tiny L}}/v_{\text{\tiny H}}\le \frac{s}{s+1}$ and $p_2=v_{\text{\tiny L}}$ otherwise. In this way, the seller's expected second-period revenue is $\max\left\{\frac{s}{s+1}v_{\text{\tiny H}},v_{\text{\tiny L}}\right\}$. (b) When setting the first-period price $p_1=v_{\text{\tiny H}}$, the seller charges $p_2=v_{\text{\tiny L}}$ in the second period if the first-arriving buyer does not purchase $a_1=0$. Here, the corresponding second-period revenue is $v_{\text{\tiny L}}$. If the first-arriving buyer purchases $a_1=1$ at $p_1=v_{\text{\tiny H}}$, the second-period price is $p_2=v_{\text{\tiny H}}$ and the second-period revenue is $v_{\text{\tiny H}}$.
Now we turn to the seller's first-period problem:
\begin{equation}
    \max_{p_1\in\{v_{\text{\tiny H}},v_{\text{\tiny L}}\}}\left\{\left(\frac{s}{s+1}\left(v_{\text{\tiny H}}+v_{\text{\tiny H}}\right)+\frac{1}{s+1}v_{\text{\tiny L}}\right)\boldsymbol{1}(p_1=v_{\text{\tiny H}})+\left(v_{\text{\tiny L}}+\max\left\{\frac{s}{s+1}v_{\text{\tiny H}},v_{\text{\tiny L}}\right\}\right)\boldsymbol{1}(p_1=v_{\text{\tiny L}})\right\}.
\end{equation}
which yields the seller's optimal first-period price $p_1=v_{\text{\tiny L}}$ if $v_{\text{\tiny L}}/v_{\text{\tiny H}}\ge {2s}/{(1+2s)}$ and $p_2=v_{\text{\tiny H}}$ otherwise. 

\end{proof}

\begin{lemma}\label{rho01_lemma2}
Given the posterior beliefs in (\ref{rho01_postbelief2}), when observing buyers' low common interaction frequency $\hat{x}=0$, the seller's optimal pricing in Stage II depends on the preference diversity $v_{\text{\tiny L}}/v_{\text{\tiny H}}$ as follows.
\begin{itemize}
    \item \textbf{Case I}: If ${(3-s)}/{(4-s)}\le v_{\text{\tiny L}}/v_{\text{\tiny H}}<1$, the seller charges a uniform price $p_1=p_2=v_{\text{\tiny L}}$.
    \item \textbf{Case II}: If $(1-s)/(2-s)\le v_{\text{\tiny L}}/v_{\text{\tiny H}}\le {(3-s)}/{(4-s)}$, the seller charges a high price $p_1=v_{\text{\tiny H}}$ in the first selling period. The second-period price personalizes, depending on the purchase record $a_1\in\{0,1\}$ of the first period in Stage II:
        \begin{equation}
            p_2(a_1)=v_{\text{\tiny L}}\boldsymbol{1}(a_1=1)+v_{\text{\tiny H}}\boldsymbol{1}(a_1=0).
        \end{equation}
    \item \textbf{Case III}: If $0<v_{\text{\tiny L}}/v_{\text{\tiny H}}\le (1-s)/(2-s)$, the seller charges a uniform price $p_1=p_2=v_{\text{\tiny H}}$.
\end{itemize}
\end{lemma}

\begin{proof}
    When observing buyers' low common interaction frequency $\hat{x}=0$, based on the posterior belief in (\ref{rho01_postbelief2}), the seller updates that 
\begin{equation}
{\rm Pr}(v_1=v_2=v_{\text{\tiny H}}|\hat{x}=0)=\frac{1-s}{3-s},
\end{equation}
and
\begin{equation}
{\rm Pr}(v_1=v_{\text{\tiny H}},v_2=v_{\text{\tiny L}}|\hat{x}=0)={\rm Pr}(v_1=v_{\text{\tiny L}},v_2=v_{\text{\tiny H}}|\hat{x}=0)=\frac{1}{3-s};
\end{equation}
that is, we have
\begin{equation}
{\rm Pr}(v_1=v_{\text{\tiny H}}|\hat{x}=0)={\rm Pr}(a_1=1|\hat{x}=0;p_1=v_{\text{\tiny H}})={\rm Pr}(v_2=v_{\text{\tiny H}}|\hat{x}=1;p_1=v_{\text{\tiny L}})=\frac{2-s}{3-s},
\end{equation}
and
\begin{equation}
{\rm Pr}(v_1=v_{\text{\tiny L}}|\hat{x}=0)={\rm Pr}(a_1=0|\hat{x}=1;p_1=v_{\text{\tiny H}})={\rm Pr}(v_2=v_{\text{\tiny L}}|\hat{x}=1;p_1=v_{\text{\tiny L}})=\frac{1}{3-s}.
\end{equation}
Specifically, we first analyze the second selling period. (a) When setting the first-period price $p_1=v_{\text{\tiny L}}$, the seller's optimal second-period price is $p_2=v_{\text{\tiny H}}$ if $v_{\text{\tiny L}}/v_{\text{\tiny H}}\le\frac{2-s}{3-s}$ and $p_2=v_{\text{\tiny L}}$ otherwise. In this way, the seller's expected second-period revenue is $\max\left\{\frac{2-s}{3-s}v_{\text{\tiny H}},v_{\text{\tiny L}}\right\}$. (b) When setting the first-period price $p_1=v_{\text{\tiny H}}$, the seller charges $p_2=v_{\text{\tiny H}}$ in the second period if the first-arriving buyer does not purchase $a_1=0$. Here, the corresponding second-period revenue is $v_{\text{\tiny H}}$. If the first-arriving buyer purchases $a_1=1$ at $p_1=v_{\text{\tiny H}}$, the second-period price is $p_2=v_{\text{\tiny H}}$ if $v_{\text{\tiny L}}/v_{\text{\tiny H}}\le\frac{1-s}{2-s}$ and $p_2=v_{\text{\tiny L}}$ otherwise. As a result, the seller's expected second-period revenue is $\max\left\{\frac{1-s}{2-s}v_{\text{\tiny H}},v_{\text{\tiny L}}\right\}$. Now we turn to the seller's first-period problem:
\begin{multline}
\max_{p_1\in\{v_{\text{\tiny H}},v_{\text{\tiny L}}\}}\bigg\{\bigg(\frac{2-s}{3-s}\left(v_{\text{\tiny H}}+\max\left\{\frac{1-s}{2-s}v_{\text{\tiny H}},v_{\text{\tiny L}}\right\}\right)+\frac{1}{3-s}v_{\text{\tiny H}}\bigg)\boldsymbol{1}(p_1=v_{\text{\tiny H}})\\+\left(v_{\text{\tiny L}}+\max\left\{\frac{2-s}{3-s}v_{\text{\tiny H}},v_{\text{\tiny L}}\right\}\right)\boldsymbol{1}(p_1=v_{\text{\tiny L}})\bigg\}.
\end{multline}
which yields the seller's optimal first-period price $p_1=v_{\text{\tiny L}}$ if $v_{\text{\tiny L}}/v_{\text{\tiny H}}\ge{(3-s)}/{(4-s)}$ and $p_2=v_{\text{\tiny H}}$ otherwise.

\end{proof} 

\subsection{Step 2: Analysis of the buyers' social interactions in Stage I}

First, we analyze buyers' coupled social behaviors in Stage I, based on the seller's optimal pricing derived in Step 1. Then, we identify the conditions for which any equilibrium could be sustained.
\begin{figure*}[h]
\centering
\subfigure[Cases in Lemma \ref{rho01_lemma1}]{
\begin{minipage}[t]{0.31\linewidth}
\centering
\footnotesize
\begin{tikzpicture}[scale=2, domain=0:4]
	\draw[-latex] (0,0) -- (1.2,0) node[right] {$s$};
	\draw[-latex] (0,0) -- (0,1.2) node[right] {${v_{\text{\tiny L}}}/{v_{\text{\tiny H}}}$};
	\draw[color=black, very thick] plot[domain=0:1,samples=200] (\x,{2*\x/(1+2*\x)});
 	\draw[color=black] (0,1) -- (1,1);
 	\draw[color=black] (1,0) -- (1,1);
 	\draw[color=black,dotted] (0,2/3) -- (1,2/3);
  	\node (O) at (0,0) [below left] {$O$};
  	\node (a) at (1,0) [below] {$1$};
  	\node (e) at (0,1) [left] {$1$};
  	\node (b) at (0,2/3) [left] {$2/3$};
  	\node (c) at (1/3,3/4) {$A$};
  	\node (d) at (2/3,1/3) {$B$};
  	\draw[color=black,thick] (0,1) -- (0.05,1);
 	\draw[color=black,thick] (1,0) -- (1,0.05);
	\draw[color=black,thick] (0,2/3) -- (0.05,2/3);
	\node at (1.1,4/7) [right] {\color{black}\tiny{$\frac{v_{\text{\tiny L}}}{v_{\text{\tiny H}}}=\frac{2s}{1+2s}$}};
 	\draw[color=black,ultra thin] (2/3,4/7) -- (1.1,4/7);
\end{tikzpicture}
\end{minipage}
}
\subfigure[Cases in Lemma \ref{rho01_lemma2}]{
\begin{minipage}[t]{0.31\linewidth}
\centering
\footnotesize
\begin{tikzpicture}[scale=2, domain=0:4]
	\draw[-latex] (0,0) -- (1.2,0) node[right] {$s$};
	\draw[-latex] (0,0) -- (0,1.2) node[right] {${v_{\text{\tiny L}}}/{v_{\text{\tiny H}}}$};
	\draw[color=black,very thick] plot[domain=0:1,samples=200] (\x,{(3-\x)/(4-\x)});
	\draw[color=black,very thick] plot[domain=0:1,samples=200] (\x,{(1-\x)/(2-\x)});
 	\draw[color=black] (0,1) -- (1,1);
 	\draw[color=black] (1,0) -- (1,1);
 	\draw[color=black,dotted] (0,2/3) -- (1,2/3);
  	\node (O) at (0,0) [below left] {$O$};
  	\node (a) at (1,0) [below] {$1$};
  	\node (e) at (0,1) [left] {$1$};
  	\node (f) at (0,3/4) [left] {\tiny{$3/4$}};
  	\node (g) at (0,1/2) [left] {\tiny{$1/2$}};
  	\node (b) at (0,2/3) [left] {\tiny{$2/3$}};
  	\node (c) at (1/2,6/7) {$I$};
  	\node (d) at (2/3,1/2) {$II$};
  	\node (d) at (1/3,1/5) {$III$};
  	\draw[color=black,thick] (0,1) -- (0.05,1);
 	\draw[color=black,thick] (1,0) -- (1,0.05);
 	\draw[color=black,thick] (0,1/2) -- (0.05,1/2);
	\draw[color=black,thick] (0,2/3) -- (0.05,2/3);
	\draw[color=black,thick] (0,3/4) -- (0.05,3/4);
	\draw[color=black,ultra thin] (0.5,1/3) -- (-0.1,1/3);
 	\node at (-0.1,1/3) [left] {\tiny{$\frac{v_{\text{\tiny L}}}{v_{\text{\tiny H}}}=\frac{1-s}{2-s}$}};
    \draw[color=black,ultra thin] (0.5,5/7) -- (1.05,5/7);
    \node at (1.05,5/7) [right] {\tiny{$\frac{v_{\text{\tiny L}}}{v_{\text{\tiny H}}}=\frac{3-s}{4-s}$}};
\end{tikzpicture}
\end{minipage}
}
\subfigure[Combined Cases for Step 2]{
\begin{minipage}[t]{0.31\linewidth}
\centering
\footnotesize
\begin{tikzpicture}[scale=2, domain=0:4]
	\draw[-latex] (0,0) -- (1.2,0) node[right] {$s$};
	\draw[-latex] (0,0) -- (0,1.2) node[right] {${v_{\text{\tiny L}}}/{v_{\text{\tiny H}}}$};
	\draw[color=blue,very thick] plot[domain=0:1/3,samples=200] (\x,{2*\x/(1+2*\x)});
	\draw[color=gray!90!black,very thick] plot[domain=1/3:1,samples=200] (\x,{2*\x/(1+2*\x)});
	\draw[color=red,very thick] plot[domain=0:1,samples=200] (\x,{(3-\x)/(4-\x)});
	\draw[color=green!90!black,very thick] plot[domain=0:1/3,samples=200] (\x,{(1-\x)/(2-\x)});
	\draw[color=yellow!90!black,very thick] plot[domain=1/3:1,samples=200] (\x,{(1-\x)/(2-\x)});
 	\draw[color=black] (0,1) -- (1,1);
 	\draw[color=black] (1,0) -- (1,1);
 	\draw[color=black,dotted] (0,2/3) -- (1,2/3);
 	\draw[color=black,dotted] (0,2/5) -- (1/3,2/5);
 	\draw[color=black,dotted] (1/3,0) -- (1/3,2/5);
  	\node (O) at (0,0) [below left] {$O$};
  	\node at (1,0) [below] {$1$};
  	\node at (1/3,0) [below] {$1/3$};
  	\node at (0,1) [left] {$1$};
  	\node at (0,3/4) [left] {\tiny{$3/4$}};
  	\node at (0,1/2) [left] {\tiny{$1/2$}};
  	\node at (0,2/3) [left] {\tiny{$2/3$}};
  	\node at (0,2/5) [left] {\tiny{$2/5$}};
  	\draw[color=black,thick] (0,1) -- (0.05,1);
  	\draw[color=black,thick] (1/3,0) -- (1/3,0.05);
 	\draw[color=black,thick] (1,0) -- (1,0.05);
 	\draw[color=black,thick] (0,1/2) -- (0.05,1/2);
	\draw[color=black,thick] (0,2/3) -- (0.05,2/3);
	\draw[color=black,thick] (0,3/4) -- (0.05,3/4);
	\draw[color=black,thick] (0,2/5) -- (0.05,2/5);
  	\node at (2.35/3,0.495) {\tiny{Appendix}};
  	\node at (2.5/3,0.41) {\tiny{\ref{part3}}};
  	\node at (2.5/3,0.325) {\tiny{Case 1}};
  	\draw[color=yellow!90!black,ultra thin] (2/3,1/4) -- (1.1,1/4);
  	\node at (1.1,1/4) [right] {\tiny{Appendix \ref{part3} Case 3}};
 	\draw[color=gray!90!black,ultra thin] (2/3,4/7) -- (1.1,4/7);
 	\node at (1.1,4/7) [right] {\tiny{Appendix \ref{part3} Case 2}};
\end{tikzpicture}
\end{minipage}
\label{combined_cases}
}
\caption{Combined Cases to Discuss in Step 2 based on Lemma \ref{rho01_lemma1} and Lemma \ref{rho01_lemma2}}
\label{fig_rho01_cases}
\end{figure*}
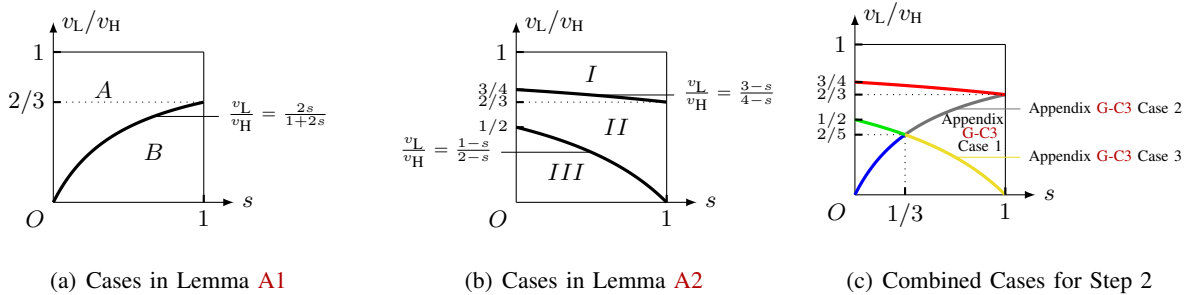

\subsection{Step 2: Analysis of Buyers' Interactions in Stage I and Characterization of PBE}\label{step2_app}

Next, we continue to analyze buyers' social behaviors in Stage I. In what follows, we combine Lemma \ref{rho01_lemma1} and Lemma \ref{rho01_lemma2} to examine all possible cases regarding the seller's pricing decisions in Stage II, as illustrated by boundaries, colored curves, and regions between curves in Fig. \ref{combined_cases}. Specifically, we first calculate high-preference buyers' purchase surplus and accordingly update the expected final payoff with their social interaction utility in Table \ref{table_same}. Next, through normal-form game-theoretic analysis, we derive high-preference buyers' equilibrium social interactions in Stage I based on the derived expected final payoff table. Finally, we examine the conditions for which the derived equilibrium could be sustained as a PBE, ensuring the consistency between buyers' social interaction decisions in Stage I and the seller's belief.

\subsubsection{First consider when $s=1$, high-preference buyers do not manipulate with $\rho=0$}\label{part1}

\textbf{Case 1:} If $2/3\le v_{\text{\tiny L}}/v_{\text{\tiny H}}<1$, then these buyers anticipate the seller's uniform pricing $p_1^*=p_2^*=v_{\text{\tiny L}}$ in Stage II. It follows that a high-preference buyer always obtains an expected purchase surplus of $\tilde{\pi}_i(v_i=v_{\text{\tiny H}})=v_{\text{\tiny H}}-v_{\text{\tiny L}}$, independent of their social interaction decisions. Now, we update high-preference buyer $i$'s expected final payoff as below, i.e., social interaction utility in Stage I plus this expected purchase surplus in Stage II.\footnote{Notice that as both high-preference buyers are symmetric, for convenience, we hereafter draw the final payoff matrix with only one buyer's (i.e., high-preference buyer $i$'s) final payoff.}

\begin{table}[h]
\caption{Expected Final Payoff Matrix for High-Preference Buyer $i$ in Case 1 of Appendix \ref{part1}}
\begin{center}
\footnotesize
\begin{tikzpicture}[thick]
\draw (-2,0) rectangle (4,1);
\draw (-2,0.5) -- (4,0.5);
\draw (1,0) -- (1,1);
\foreach \x/\xtext in {-0.5/{$1+v_{\text{\tiny H}}-v_{\text{\tiny L}}$},2.5/{$1-l+v_{\text{\tiny H}}-v_{\text{\tiny L}}$}}
\draw (\x,0.75) node {\xtext};
\foreach \x/\xtext in {-0.5/{$l+v_{\text{\tiny H}}-v_{\text{\tiny L}}$},2.5/{$v_{\text{\tiny H}}-v_{\text{\tiny L}}$}}
\draw (\x,0.25) node {\xtext};
\node at (-2,0.75) [left] {$x_{ij}=1$};
\node at (-2,0.25) [left] {$x_{ij}=0$};
\node at (-0.5,1) [above] {$x_{ji}=1$};
\node at (2.5,1) [above] {$x_{ji}=0$};
\end{tikzpicture}
\end{center}
\end{table}

Through normal-form game-theoretic analysis on the payoff matrix above, high-preference buyers' equilibrium social decisions are $x_{ij}^*(v_i=v_j=v_{\text{\tiny H}})=x_{ji}^*(v_j=v_i=v_{\text{\tiny H}})=1$. That is, such equilibrium social strategies are \textit{consistent} with the seller's belief with $s=1$. This implies that the strategy profile described above, including the seller's and buyers', can be sustained as a PBE. \textit{Thus far, we have completed the proof for Proposition \ref{region1}.}$\qed$

\textbf{Case 2:} If $0< v_{\text{\tiny L}}/v_{\text{\tiny H}}\le 2/3$, high-preference buyers anticipate the seller's personalized pricing in Stage II. i.e., the seller offers $p_1=v_{\text{\tiny H}}$ in the first selling period, and charges  (\ref{condition_1}) or (\ref{condition_2}) in the second selling period, depending on the observed buyers' common interaction frequency $\hat{x}$. Accordingly, high-preference buyers obtain an expected purchase surplus $\tilde{\pi}_i(v_i=v_{\text{\tiny H}})=0$ if $\hat{x}=1$ (by honestly choosing $x_{ij}(v_i=v_j=v_{\text{\tiny H}})=x_{ji}(v_i=v_j=v_{\text{\tiny H}})=1$), and $\tilde{\pi}_i(v_i=v_{\text{\tiny H}})=\frac{v_{\text{\tiny H}}-v_{\text{\tiny L}}}{2}$ otherwise if $\hat{x}=0$ through social manipulation. Now we update the expected final payoff matrix for high-preference buyer $i$ as below.

\begin{table}[h]
\caption{Expected Final Payoff Matrix for High-Preference Buyer $i$ in Case 2 of Appendix \ref{part1}}
\begin{center}
\footnotesize
\begin{tikzpicture}[thick]
\draw (-2,0) rectangle (4,1);
\draw (-2,0.5) -- (4,0.5);
\draw (1,0) -- (1,1);
\foreach \x/\xtext in {-0.5/{$1$},2.5/{$1-l+\frac{v_{\text{\tiny H}}-v_{\text{\tiny L}}}{2}$}}
\draw (\x,0.75) node {\xtext};
\foreach \x/\xtext in {-0.5/{$l+\frac{v_{\text{\tiny H}}-v_{\text{\tiny L}}}{2}$},2.5/{$\frac{v_{\text{\tiny H}}-v_{\text{\tiny L}}}{2}$}}
\draw (\x,0.25) node {\xtext};
\node at (-2,0.75) [left] {$x_{ij}=1$};
\node at (-2,0.25) [left] {$x_{ij}=0$};
\node at (-0.5,1) [above] {$x_{ji}=1$};
\node at (2.5,1) [above] {$x_{ji}=0$};
\end{tikzpicture}
\end{center}
\end{table}

Then, we discuss two cases as follows. (a) If $v_{\text{\tiny H}}-v_{\text{\tiny L}}>2(1-l)$, the high-preference buyers' equilibrium social decision is to manipulate with probability $\rho^*=1-2(1-l)/(v_{\text{\tiny H}}-v_{\text{\tiny L}})\in(0,1)$, implying a mixed strategy. This is \textit{inconsistent} with the seller's belief with $s=1$. Hence, the strategy profile described above \textit{cannot} be sustained as a PBE. (b) If $v_{\text{\tiny H}}-v_{\text{\tiny L}}\le 2(1-l)$, the high-preference buyers' equilibrium social decisions are $x_{ij}^*(v_i=v_j=v_{\text{\tiny H}})=x_{ji}^*(v_j=v_i=v_{\text{\tiny H}})=1$, implying the strategy profile described above can be sustained as a PBE. \textit{Thus far, we have completed the proof for Proposition \ref{Region2}.}$\qed$

\subsubsection{Then consider another extreme case when $s=0$, high-preference buyers always manipulate with $\rho=1$}\label{part2}

In this case, high-preference buyers then anticipate the seller's optimal pricing in Stage II based on Lemma \ref{rho01_lemma1} and Lemma \ref{rho01_lemma2}. Through similar analysis to Appendix \ref{part1}, it turns out that high-preference buyers' equilibrium social decisions in Stage I would be $x_{ij}^*(v_i=v_j=v_{\text{\tiny H}})=x_{ji}^*(v_j=v_i=v_{\text{\tiny H}})=1$ in this case. This implies that the strategy profile with $\rho=1$ \textit{cannot} be sustained as a PBE.

\subsubsection{Finally we come to the case when $0<s<1$, high-preference buyers manipulate with positive probability $\rho$}\label{part3}

\textbf{Case 1:} When (See Region between Gray and Yellow curves in Fig. \ref{combined_cases})
\begin{equation}\label{caseBII_condition}
    \frac{1-s}{2-s}<\frac{v_{\text{\tiny L}}}{v_{\text{\tiny H}}}<\frac{2s}{1+2s},
\end{equation}
high-preference buyers anticipate the seller's optimal pricing in Stage II as described in Case B of Lemma \ref{rho01_lemma1} and Case II of Lemma \ref{rho01_lemma2}. Now, we first update the expected final payoff matrix for high-preference buyer $i$ as below.

\begin{table}[h]
\caption{Expected Final Payoff Matrix for High-Preference Buyer $i$ in Case 1 of Appendix \ref{part3}}
\begin{center}
\footnotesize
\begin{tikzpicture}[thick]
\draw (-2,0) rectangle (4,1);
\draw (-2,0.5) -- (4,0.5);
\draw (1,0) -- (1,1);
\foreach \x/\xtext in {-0.5/{$1$},2.5/{$1-l+\frac{v_{\text{\tiny H}}-v_{\text{\tiny L}}}{2}$}}
\draw (\x,0.75) node {\xtext};
\foreach \x/\xtext in {-0.5/{$l+\frac{v_{\text{\tiny H}}-v_{\text{\tiny L}}}{2}$},2.5/{$\frac{v_{\text{\tiny H}}-v_{\text{\tiny L}}}{2}$}}
\draw (\x,0.25) node {\xtext};
\node at (-2,0.75) [left] {$x_{ij}=1$};
\node at (-2,0.25) [left] {$x_{ij}=0$};
\node at (-0.5,1) [above] {$x_{ji}=1$};
\node at (2.5,1) [above] {$x_{ji}=0$};
\end{tikzpicture}
\end{center}
\end{table}

Then, we derive that high-preference buyers' equilibrium social decision is to manipulate with probability
\begin{equation}\label{caseBII_manipurho}
    \rho^*=1-\frac{2(1-l)}{v_{\text{\tiny H}}-v_{\text{\tiny L}}}.
\end{equation}
To sustain the mixed PBE, we need to ensure $\rho\in(0,1)$, i.e.,
\begin{equation}\label{caseBII_1}
v_{\text{\tiny H}}-v_{\text{\tiny L}}>2(1-l).
\end{equation}
Plus, plugging (\ref{caseBII_manipurho}) into (\ref{caseBII_condition}) and simplifying yields
\begin{gather}
    v_{\text{\tiny L}}(v_{\text{\tiny H}}-v_{\text{\tiny L}})<8(1-l)^2,\label{caseBII_2}\\
    (v_{\text{\tiny H}}-v_{\text{\tiny L}})(v_{\text{\tiny H}}-2v_{\text{\tiny L}})<4(1-l)^2.\label{caseBII_3}   
\end{gather}
That is, under conditions (\ref{caseBII_1}), (\ref{caseBII_2}) and (\ref{caseBII_3}), the strategy profile described above can be sustained as a PBE. \textit{This completes our proof for Proposition \ref{Region3}.} $\qed$

\textbf{Case 2:} When (see Gray curve in Fig. \ref{combined_cases})
\begin{gather}
    \frac{v_{\text{\tiny L}}}{v_{\text{\tiny H}}}=\frac{2s}{1+2s},\label{caseABII_condition}\\
    \frac{2}{5}<\frac{v_{\text{\tiny L}}}{v_{\text{\tiny H}}}<\frac{2}{3}.\label{caseABII_1}
\end{gather}
In this case, high-preference buyers anticipate the seller's mixed pricing strategy in Stage II upon $\hat{x}=1$. Let $\beta$ denote the probability with which the seller charges a uniform price $p_1=p_2=v_{\text{\tiny L}}$, and $1-\beta$ denote the probability with which the seller charges a high first-period price and then enables personalized pricing in the second selling period as in Case B of Lemma \ref{rho01_lemma1}. Whereas upon $\hat{x}=0$, the seller charges as in Case II of Lemma \ref{rho01_lemma2}. Similarly, we now update the expected final payoff matrix for high-preference buyer $i$ as below.

\begin{table}[h]
\caption{Expected Final Payoff Matrix for High-Preference Buyer $i$ in Case 2 of Appendix \ref{part3}}
\begin{center}
\footnotesize
\begin{tikzpicture}[thick]
\draw (-2,0) rectangle (4,1);
\draw (-2,0.5) -- (4,0.5);
\draw (1,0) -- (1,1);
\foreach \x/\xtext in {-0.5/{$1+\beta(v_{\text{\tiny H}}-v_{\text{\tiny L}})$},2.5/{$1-l+\frac{v_{\text{\tiny H}}-v_{\text{\tiny L}}}{2}$}}
\draw (\x,0.75) node {\xtext};
\foreach \x/\xtext in {-0.5/{$l+\frac{v_{\text{\tiny H}}-v_{\text{\tiny L}}}{2}$},2.5/{$\frac{v_{\text{\tiny H}}-v_{\text{\tiny L}}}{2}$}}
\draw (\x,0.25) node {\xtext};
\node at (-2,0.75) [left] {$x_{ij}=1$};
\node at (-2,0.25) [left] {$x_{ij}=0$};
\node at (-0.5,1) [above] {$x_{ji}=1$};
\node at (2.5,1) [above] {$x_{ji}=0$};
\end{tikzpicture}
\end{center}
\end{table}

Accordingly, one can derive that high-preference buyers' equilibrium social decision is to manipulate with probability
\begin{equation}\label{caseABII_manipurho1}
    \rho^*=1-\frac{2(1-l)}{(1-2\beta)(v_{\text{\tiny H}}-v_{\text{\tiny L}})}; 
\end{equation}
plus, equivalently, we also have from (\ref{caseABII_condition})
\begin{equation}\label{caseABII_manipurho2}
    \rho^*=1-\sqrt{\frac{v_{\text{\tiny L}}}{2(v_{\text{\tiny H}}-v_{\text{\tiny L}})}}.
\end{equation}
Combining (\ref{caseABII_manipurho1}) and (\ref{caseABII_manipurho2}), we have for the seller's mixed strategy
\begin{equation}\label{caseABII_beta}
    \beta = \frac{1}{2}-\frac{\sqrt{2}(1-l)}{\sqrt{(v_{\text{\tiny H}}-v_{\text{\tiny L}})v_{\text{\tiny L}}}}.
\end{equation}
To sustain the PBE, we need to ensure $\rho\in(0,1)$ as well as $\beta\in(0,1)$:
\begin{equation}\label{caseABII_2}
v_{\text{\tiny L}}(v_{\text{\tiny H}}-v_{\text{\tiny L}})>8(1-l)^2.
\end{equation}
That is, under conditions (\ref{caseABII_1}) and (\ref{caseABII_2}), the strategy profile described above can be sustained as a PBE. \textit{This completes our proof for Case 1 in Proposition \ref{Region4}.} $\qed$

\textbf{Case 3:} When (see Yellow curve in Fig. \ref{combined_cases})
\begin{gather}
\frac{v_{\text{\tiny L}}}{v_{\text{\tiny H}}}=\frac{1-s}{2-s},\label{caseB23_condition}\\
0<\frac{v_{\text{\tiny L}}}{v_{\text{\tiny H}}}<\frac{2}{5}.\label{caseB23_1}
\end{gather}
In this case, high-preference buyers anticipate the seller's mixed pricing strategy in Stage II upon $\hat{x}=0$. Let $\beta$ denote the probability with which the seller charges a high first-period price and then enables personalized pricing in the second selling period as in Case II of Lemma \ref{rho01_lemma2}, and $1-\beta$ denote the probability with which the seller charges a uniform price $p_1=p_2=v_{\text{\tiny H}}$,. Whereas upon $\hat{x}=1$, the seller charges as in Case B of Lemma \ref{rho01_lemma1}. Similarly, we now update the expected final payoff matrix for high-preference buyer $i$ as below.

\begin{table}[h]
\caption{Expected Final Payoff Matrix for High-Preference Buyer $i$ in Case 3 of Appendix \ref{part3}}
\begin{center}
\footnotesize
\begin{tikzpicture}[thick]
\draw (-2,0) rectangle (4,1);
\draw (-2,0.5) -- (4,0.5);
\draw (1,0) -- (1,1);
\foreach \x/\xtext in {-0.5/{$1$},2.5/{$1-l+\beta\frac{v_{\text{\tiny H}}-v_{\text{\tiny L}}}{2}$}}
\draw (\x,0.75) node {\xtext};
\foreach \x/\xtext in {-0.5/{$l+\beta\frac{v_{\text{\tiny H}}-v_{\text{\tiny L}}}{2}$},2.5/{$\beta\frac{v_{\text{\tiny H}}-v_{\text{\tiny L}}}{2}$}}
\draw (\x,0.25) node {\xtext};
\node at (-2,0.75) [left] {$x_{ij}=1$};
\node at (-2,0.25) [left] {$x_{ij}=0$};
\node at (-0.5,1) [above] {$x_{ji}=1$};
\node at (2.5,1) [above] {$x_{ji}=0$};
\end{tikzpicture}
\end{center}
\end{table}

We can derive that high-preference buyers' equilibrium social decision is to manipulate with probability
\begin{equation}\label{caseB23_manipurho1}
    \rho^*=1-\frac{2(1-l)}{\beta(v_{\text{\tiny H}}-v_{\text{\tiny L}})};
\end{equation}
plus, equivalently, we also have from (\ref{caseB23_condition})
\begin{equation}\label{caseB23_manipurho2}
\rho^*=1-\sqrt{\frac{v_{\text{\tiny H}}-2v_{\text{\tiny L}}}{v_{\text{\tiny H}}-v_{\text{\tiny L}}}}.
\end{equation}
Combining (\ref{caseB23_manipurho1}) and (\ref{caseB23_manipurho2}), we have for the seller's mixed strategy
\begin{equation}\label{caseB23_beta}
    \beta=\frac{2(1-l)}{\sqrt{(v_{\text{\tiny H}}-v_{\text{\tiny L}})(v_{\text{\tiny H}}-2v_{\text{\tiny L}})}}.
\end{equation}
To sustain the PBE, we need to ensure $\rho\in(0,1)$ as well as $\beta\in(0,1)$:
\begin{equation}\label{caseB23_2}
(v_{\text{\tiny H}}-v_{\text{\tiny L}})(v_{\text{\tiny H}}-2v_{\text{\tiny L}})>4(1-l)^2.
\end{equation}
That is, under conditions (\ref{caseB23_1}) and (\ref{caseB23_2}), the strategy profile described above can be sustained as a PBE. \textit{This completes our proof for Case 2 in Proposition \ref{Region4}.} $\qed$

\textbf{Remaining Cases:} For the remaining cases in Fig. \ref{combined_cases}, one can verify that the PBE with $0<\rho<1$ \textit{cannot} be sustained following similar arguments. Specifically, it turns out that high-preference buyers' equilibrium social decisions in Stage I would be $x_{ij}^*(v_i=v_j=v_{\text{\tiny H}})=x_{ji}^*(v_j=v_i=v_{\text{\tiny H}})=1$ in all these cases. This is \textit{inconsistent} with the seller's belief in (\ref{rho01_postbelief1}) and (\ref{rho01_postbelief2}).

\section{Proof of Proposition \ref{dilemma_buyer}}

In this appendix, we compare the average buyer payoff in (\ref{average_buyer}) under the strategic-learning model with the undisclosed-learning benchmark. For ease of exposition, we hereafter develop our PBE proof in terms of the five $v_{\text{\tiny H}},v_{\text{\tiny L}}$ preference regions as summarized in Fig. \ref{regions_5}.

\begin{figure*}[h]
\begin{minipage}[t]{0.3\linewidth}
\centering
\vspace{0pt}
\centering
\includegraphics[width=0.8\linewidth]{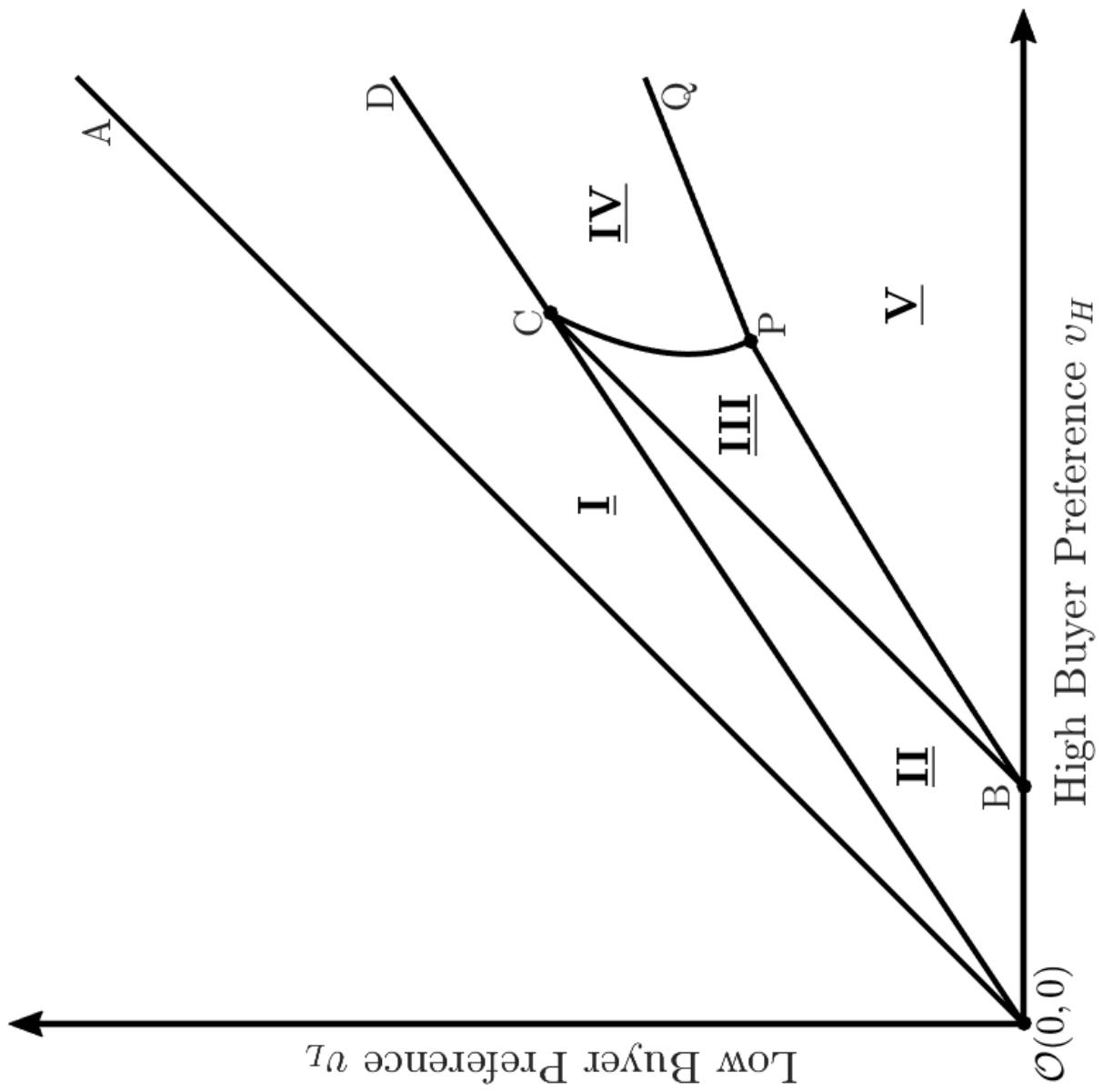}
\end{minipage}
\begin{minipage}[t]{0.3\linewidth}
\centering
\vspace{21pt}
\tiny
\begin{tabular}{cp{4cm}}
\specialrule{\heavyrulewidth}{1.2pt}{0pt}
{\makecell[c]{Curve}} &{\makecell[c]{Corresponding Equation in $(v_\text{\tiny H},v_\text{\tiny L})$-Plane}}\\
\specialrule{\lightrulewidth}{0pt}{1.2pt}
OA &{\makecell[c]{$v_{\text{\tiny L}}=v_{\text{\tiny H}}$}}\\
OCD &{\makecell[c]{$v_{\text{\tiny L}}=2v_{\text{\tiny H}}/3$}}\\
PQ &{\makecell[c]{$v_{\text{\tiny L}}=2v_{\text{\tiny H}}/5$}}\\
BC &{\makecell[c]{$v_{\text{\tiny H}}-v_{\text{\tiny L}}=2(1-l)$}}\\
CP &{\makecell[c]{$(v_{\text{\tiny H}}-v_{\text{\tiny L}})v_{\text{\tiny L}}=8(1-l)^2$}}\\
BP &{\makecell[c]{$(v_{\text{\tiny H}}-v_{\text{\tiny L}})(v_{\text{\tiny H}}-2v_{\text{\tiny L}})=4(1-l)^2$}}\\
\specialrule{\heavyrulewidth}{1.2pt}{0pt}
\end{tabular}
\end{minipage}
\hspace{20pt}
\begin{minipage}[t]{0.3\linewidth}
\centering
\vspace{35pt}
\tiny
\begin{tabular}{p{1.2cm}<{\centering}p{2.75cm}<{\centering}}
\specialrule{\heavyrulewidth}{1.2pt}{0pt}
Regions & Key PBE Results\\
\specialrule{\lightrulewidth}{0.55pt}{0.77pt}
I  &Proposition \ref{region1}\\
II &Proposition \ref{Region2}\\
III  &Proposition \ref{Region3}\\
IV  & Case 1 of Proposition \ref{Region4}\\
V  & Case 2 of Proposition \ref{Region4}\\
\specialrule{\heavyrulewidth}{1.2pt}{0pt}
\end{tabular}
\end{minipage}
\caption{PBE in Five Regions}
\label{regions_5}
\end{figure*}

\subsection{Undisclosed-Learning and Strategic-Learning Average Buyer Payoff} 

First, we provide the expressions for the average buyer payoffs under the undisclosed-learning benchmark $\tilde{\pi}^{\rm UN}$ and that under the strategic-learning model $\tilde{\pi}^{\rm ST}$. Here, we take the expectation over various buyer preference distributions while accounting for the preference correlation between buyers: $(v_i,v_j)\in\{(v_{\text{\tiny H}},v_{\text{\tiny H}}),(v_{\text{\tiny H}},v_{\text{\tiny L}}),(v_{\text{\tiny L}},v_{\text{\tiny H}}),(v_{\text{\tiny L}},v_{\text{\tiny L}})\}$. i.e.,

\begin{multline}\label{cal_payoff}
\mathbb{E}_{v_i,v_j}\{\tilde{\pi}_i(x_{ij},x_{ji})\}
=\frac{1}{4} \tilde{\pi}_i(v_i=v_j=v_{\text{\tiny H}})+\frac{1}{4} \tilde{\pi}_i(v_i=v_{\text{\tiny H}},v_j=v_{\text{\tiny L}})\\+\frac{1}{4}\tilde{\pi}_i(v_i=v_{\text{\tiny L}},v_j=v_{\text{\tiny H}})+\frac{1}{4} \tilde{\pi}_i(v_i=v_j=v_{\text{\tiny L}}). 
\end{multline}

Based on Proposition \ref{unaware}, we first compute the average buyer payoff under the \textit{undisclosed-learning} benchmark. If $v_{\text{\tiny L}}/v_{\text{\tiny H}}\le2/3$, the seller enables personalized pricing and thus any buyer receives zero purchase surplus. That is, we only need to account for the social interaction utility:
\begin{equation}
\tilde{\pi}^{\textrm{UN}}=\frac{1}{4}\cdot 1+ \frac{1}{4}\cdot 0+\frac{1}{4}\cdot 0+\frac{1}{4}\cdot 1=\frac{1}{2}.
\end{equation}
If $v_{\text{\tiny L}}/v_{\text{\tiny H}}\le 2/3$, the seller charges a uniform price $p_1^*=p_2^*=v_{\text{\tiny L}}$ and thus only high-preference buyers receive zero purchase surplus. We then have:
\begin{equation}
\tilde{\pi}^{\textrm{UN}}=\frac{1}{4}(1+v_{\text{\tiny H}}-v_{\text{\tiny L}}) + \frac{1}{4}(0+v_{\text{\tiny H}}-v_{\text{\tiny L}})+\frac{1}{4}\cdot 0+\frac{1}{4}\cdot 1=\frac{1}{2}=\frac{v_{\text{\tiny H}}-v_{\text{\tiny L}}+1}{2}.    
\end{equation}
Therefore, under the undisclosed-learning benchmark, the equilirbium average buyer payoff is summarized as below.
\begin{subnumcases}{\tilde{\pi}^{\textrm{UN}}\triangleq\mathbb{E}_{v_i,v_j}\{\tilde{\pi}_i(x_{ij},x_{ji})\}=}
\frac{1}{2}, &\textrm{if $\frac{v_{\text{\tiny L}}}{v_{\text{\tiny H}}}\le \frac{2}{3}$,}\label{payoff_un_1}\\
\frac{v_{\text{\tiny H}}-v_{\text{\tiny L}}+1}{2}, &\textrm{if $\frac{v_{\text{\tiny L}}}{v_{\text{\tiny H}}}\ge \frac{2}{3}$},\label{payoff_un_2}
\end{subnumcases}

We now compute the average buyer payoff under the \textit{strategic-learning} model, based on Propositions \ref{region1}-\ref{Region4}. Based on (\ref{cal_payoff}), we have the following analysis. When buyer preferences $v_{\text{\tiny H}},v_{\text{\tiny L}}$ belong to Region I of Fig. \ref{regions_5}, based on Proposition \ref{region1}, we have:
\begin{equation}
\tilde{\pi}^{\textrm{ST}}=\frac{1}{4}(1+v_{\text{\tiny H}}-v_{\text{\tiny L}}) + \frac{1}{4}(0+v_{\text{\tiny H}}-v_{\text{\tiny L}})+\frac{1}{4}\cdot 0+\frac{1}{4}\cdot 1=\frac{1}{2}=\frac{v_{\text{\tiny H}}-v_{\text{\tiny L}}+1}{2}.    
\end{equation} 

When buyer preferences $v_{\text{\tiny H}},v_{\text{\tiny L}}$ belong to Region II of Fig. \ref{regions_5}, based on Proposition \ref{Region2}, the seller enables personalized pricing and thus any buyer receives zero purchase surplus. Then, we have:
\begin{equation}
\tilde{\pi}^{\textrm{ST}}=\frac{1}{4}\cdot 1+ \frac{1}{4}\cdot 0+\frac{1}{4}\cdot 0+\frac{1}{4}\cdot 1=\frac{1}{2}.
\end{equation}

When buyer preferences $v_{\text{\tiny H}},v_{\text{\tiny L}}$ belong to Region III of Fig. \ref{regions_5}, based on Proposition \ref{Region3}, we have:
\begin{equation}\label{payoff_region3}
\tilde{\pi}^{\textrm{ST}}=\frac{1}{4}\tilde{\pi}^{\textrm{ST}}_i(v_i=v_j=v_{\text{\tiny H}})+\frac{1}{4}\cdot 0+\frac{1}{4}\cdot 0+\frac{1}{4}\cdot 1.
\end{equation}
Here, the payoff of high-preference buyer $i$ when interacting with buyer $j$ of the same high preference is given by
\begin{multline}\label{payoff_region3_vh}
\tilde{\pi}^{\textrm{ST}}_i(v_i=v_j=v_{\text{\tiny H}})=(1-\rho^*)^2\cdot 1+\rho^*(1-\rho^*)\cdot(l+\frac{v_{\text{\tiny H}}-v_{\text{\tiny L}}}{2})\\+\rho^*(1-\rho^*)\cdot(1-l+\frac{v_{\text{\tiny H}}-v_{\text{\tiny L}}}{2})+(\rho^*)^2\cdot\frac{v_{\text{\tiny H}}-v_{\text{\tiny L}}}{2}.
\end{multline}
with manipulation probability $\rho^*$ in (\ref{caseBII_manipurho}). Specifically, the first term in (\ref{payoff_region3_vh}) (and (\ref{payoff_region4_vh}), (\ref{payoff_region5_vh}) below) is the buyer payoff without manipulation from either buyer, whereas the second and the third term is that with either buyer's manipulation and the fourth term is that with both buyer's manipulation. Plugging (\ref{payoff_region3_vh}) into (\ref{payoff_region3}) and simplifying yields the average buyer payoff in (\ref{payoff_3}). 

When buyer preferences $v_{\text{\tiny H}},v_{\text{\tiny L}}$ belong to Region IV of Fig. \ref{regions_5}, based on Case 1 of Proposition \ref{Region4}, we have (\ref{payoff_region3}), where the payoff of high-preference buyer $i$ when interacting with buyer $j$ of the same high preference is given by
\begin{multline}\label{payoff_region4_vh}
\tilde{\pi}^{\textrm{ST}}_i(v_i=v_j=v_{\text{\tiny H}})=(1-\rho^*)^2\cdot(1+\beta(v_{\text{\tiny H}}-v_{\text{\tiny L}}))+\rho^*(1-\rho^*)\cdot(l+\frac{v_{\text{\tiny H}}-v_{\text{\tiny L}}}{2})\\+\rho^*(1-\rho^*)\cdot(1-l+\frac{v_{\text{\tiny H}}-v_{\text{\tiny L}}}{2})+(\rho^*)^2\cdot\frac{v_{\text{\tiny H}}-v_{\text{\tiny L}}}{2}.
\end{multline}
where $\rho^*$ is the manipulation probability in (\ref{caseABII_manipurho2}) and $\beta$ is the mixed pricing parameter in (\ref{caseABII_beta}). Here, $\beta(v_{\text{\tiny H}}-v_{\text{\tiny L}})$ is the expected purchase surplus when high-preference buyers do not manipulate with $\hat{x}^*=1$. Plugging (\ref{payoff_region4_vh}) into (\ref{payoff_region3}) and simplifying yields the average buyer payoff in (\ref{payoff_4}). 

When buyer preferences $v_{\text{\tiny H}},v_{\text{\tiny L}}$ belong to Region V of Fig. \ref{regions_5}, based on Case 2 of Proposition \ref{Region4}, we have (\ref{payoff_region3}), where the payoff of high-preference buyer $i$ when interacting with buyer $j$ of the same high preference is given by
\begin{multline}\label{payoff_region5_vh}
\tilde{\pi}^{\textrm{ST}}_i(v_i=v_j=v_{\text{\tiny H}})=(1-\rho^*)^2\cdot1+\rho^*(1-\rho^*)\cdot(l+\beta\frac{v_{\text{\tiny H}}-v_{\text{\tiny L}}}{2})\\+\rho^*(1-\rho^*)\cdot(1-l+\beta\frac{v_{\text{\tiny H}}-v_{\text{\tiny L}}}{2})+(\rho^*)^2\cdot\beta\frac{v_{\text{\tiny H}}-v_{\text{\tiny L}}}{2}.
\end{multline}
where $\rho^*$ is the manipulation probability in (\ref{caseB23_manipurho2}) and $\beta$ is the mixed pricing parameter in (\ref{caseB23_beta}). Here, $\beta\frac{v_{\text{\tiny H}}-v_{\text{\tiny L}}}{2}$ is the expected purchase surplus when high-preference buyers manipulate to $\hat{x}^*=0$. Plugging (\ref{payoff_region5_vh}) into (\ref{payoff_region3}) and simplifying yields the average buyer payoff in (\ref{payoff_5}).

To summarize, under the strategic-learning model, the equilirbium average buyer payoff is summarized as below.
\begin{subnumcases}{\tilde{\pi}^{\textrm{ST}}=}
\frac{v_{\text{\tiny H}}-v_{\text{\tiny L}}+1}{2}, &\textrm{In Region I of Fig. \ref{regions_5},}\label{payoff_1}\\
\frac{1}{2}, &\textrm{In Region II of Fig. \ref{regions_5},}\label{payoff_2}\\
\frac{l(1-l)}{2(v_{\text{\tiny H}}-v_{\text{\tiny L}})}+\frac{v_{\text{\tiny H}}-v_{\text{\tiny L}}}{8}+\frac{1}{4},&\textrm{In Region III of Fig. \ref{regions_5},}\label{payoff_3}\\
\frac{l}{4}\sqrt{\frac{v_{\text{\tiny L}}}{2(v_{\text{\tiny H}}-v_{\text{\tiny L}})}}+\frac{v_{\text{\tiny H}}-v_{\text{\tiny L}}}{8}+\frac{1}{4},&\textrm{In Region IV of Fig. \ref{regions_5},}\label{payoff_4}\\
\frac{1-l}{4}\sqrt{\frac{v_{\text{\tiny H}}-v_{\text{\tiny L}}}{v_{\text{\tiny H}}-2v_{\text{\tiny L}}}}+\frac{l}{4}\sqrt{\frac{v_{\text{\tiny H}}-2v_{\text{\tiny L}}}{v_{\text{\tiny H}}-v_{\text{\tiny L}}}}+\frac{1}{4},&\textrm{In Region V of Fig. \ref{regions_5}.}\label{payoff_5} 
\end{subnumcases}

\subsection{Comparison between Average Buyer Payoff $\tilde{\pi}^{\textrm{UN}}$ and $\tilde{\pi}^{\textrm{ST}}$}

Next, we compare $\tilde{\pi}^{\textrm{ST}}$ under the strategic-learning model with $\tilde{\pi}^{\textrm{UN}}$ under the undisclosed-learning benchmark. In Regions I and II of Fig. \ref{regions_5}, we have $\tilde{\pi}^{\textrm{ST}}=\tilde{\pi}^{\textrm{UN}}$, i.e., buyers' awareness does affect their final payoff. 

In Region III of Fig. \ref{regions_5} where $2(1-l)\le v_{\text{\tiny H}}-v_{\text{\tiny L}}\le6(1-l)/\sqrt{3}$, the average buyer payoff in (\ref{payoff_3}) increases as $v_{\text{\tiny H}}-v_{\text{\tiny L}}$ increases if $l\le1/2$; otherwise, $\tilde{\pi}^{\textrm{ST}}$ first decreases and then increases with $v_{\text{\tiny H}}-v_{\text{\tiny L}}$. Also, the boundary conditions indicate that 
\begin{equation}
    \tilde{\pi}^{\textrm{ST}}\left(v_{\text{\tiny H}}-v_{\text{\tiny L}}=2(1-l)\right)=\tilde{\pi}^{\textrm{UN}}=\frac{1}{2};
\end{equation}
plus, at another boundary $v_{\text{\tiny H}}-v_{\text{\tiny L}}=6(1-l)/\sqrt{3}$, we have
\begin{equation}
    \left\{
    \begin{aligned}
    &\tilde{\pi}^{\textrm{ST}}\left(v_{\text{\tiny H}}-v_{\text{\tiny L}}=\frac{6(1-l)}{\sqrt{3}}\right)>\tilde{\pi}^{\textrm{UN}}=\frac{1}{2},&\quad\textrm{if $\frac{1}{2}<l<\frac{3-\sqrt{3}}{2}$},\\
    &\tilde{\pi}^{\textrm{ST}}\left(v_{\text{\tiny H}}-v_{\text{\tiny L}}=\frac{6(1-l)}{\sqrt{3}}\right)<\tilde{\pi}^{\textrm{UN}}=\frac{1}{2},&\quad\textrm{if $\frac{3-\sqrt{3}}{2}<l<1$}.
    \end{aligned}
    \right.
\end{equation}

Further in Region IV, the analysis trick here is to change variable by letting
\begin{equation}
\left\{
\begin{aligned}
&x=v_{\text{\tiny L}}/v_{\text{\tiny H}},\\
&y=v_{\text{\tiny H}}-v_{\text{\tiny L}};
\end{aligned}
\right.
\end{equation}
and the analysis proceeds with the main idea as follows. Now, the comparison between $\tilde{\pi}^{\textrm{ST}}$ in (\ref{payoff_4}) and $\tilde{\pi}^{\textrm{UN}}$ can be interpreted geometrically: to compare the region indicated by $\tilde{\pi}^{\textrm{ST}}-\tilde{\pi}^{\textrm{UN}}<0$ with the Region IV in this new coordinate system. What follows is to analyze these two regions, and it turns out only one boundary of each region is our focus in these new coordinates. i.e.,
\begin{equation}
\left\{
\begin{aligned}
&y=2(1-l)\sqrt{\frac{2(1-x)}{x}},&\quad\textrm{(Boundary of Region IV)}\\
&y=2-2l\sqrt{\frac{x}{2(1-x)}}.&\quad\textrm{(Boundary satisfying $\tilde{\pi}^{\textrm{ST}}=\tilde{\pi}^{\textrm{UN}}$)}
\end{aligned}
\right.    
\end{equation}

Finally in Region V, (\ref{payoff_5}) can be analyzed by changing variable with
\begin{equation}
    z=\sqrt{\frac{v_{\text{\tiny H}}-2v_{\text{\tiny L}}}{v_{\text{\tiny H}}-v_{\text{\tiny L}}}}.
\end{equation}
In this way, the first derivative of (\ref{payoff_5}) in terms of $z$ is $0$ at $z^*=\sqrt{1-l/l}$. Similar to the analysis in Region III of Fig. \ref{regions_5}, proofs can then be developed with the help of boundary conditions yet become more involved. In conclusion, the analyses above lead to Proposition \ref{dilemma_buyer} and Fig.\ref{userpayoff}.

\section{Proof of Proposition \ref{salerevenue1}}

In this appendix, we compare the seller's expected sale revenue under the strategic-learning model with the no-learning benchmark. In what follows, we first provide the expression for the seller's expected revenue $\tilde{\Pi}^{\textrm{ST}}$ under the strategic-learning model. Next, comparisons are made with the expected sale revenue $\tilde{\Pi}^{\textrm{NO}}$ in the no-learning benchmark. 

\subsection{Expected Revenue under Strategic-Learning Model}\label{strategic_revenue}

Based on Propositions \ref{region1}-\ref{Region4}, we now compute the seller's expected revenue $\tilde{\Pi}^{\textrm{ST}}$ under the \textit{strategic-learning} model. Here, we take the expectation over all possible arriving buyers in the market: $(v_1,v_2)\in\{(v_{\text{\tiny H}},v_{\text{\tiny H}}),(v_{\text{\tiny H}},v_{\text{\tiny L}}),(v_{\text{\tiny L}},v_{\text{\tiny H}}),(v_{\text{\tiny L}},v_{\text{\tiny L}})\}$. i.e.,
\begin{equation}
\tilde{\Pi}^{\textrm{ST}}=\frac{1}{4}\left(\Pi^{\textrm{ST}}(v_1=v_2=v_{\text{\tiny H}})+\Pi^{\textrm{ST}}(v_1=v_{\text{\tiny H}},v_2=v_{\text{\tiny L}})+\Pi^{\textrm{ST}}(v_1=v_{\text{\tiny L}},v_2=v_{\text{\tiny H}})+\Pi^{\textrm{ST}}(v_1=v_2=v_{\text{\tiny L}})\right).
\end{equation}

When buyer preferences $v_{\text{\tiny H}},v_{\text{\tiny L}}$ belong to Region I of Fig. \ref{regions_5}, based on Proposition \ref{region1}, the seller charges a uniform price $p_1^*=p_2^*=v_{\text{\tiny L}}$. Then, we have:
\begin{equation}
\tilde{\Pi}^{\textrm{ST}}=\tilde{\Pi}^{\textrm{ST}}(p_1^*=p_2^*=v_{\text{\tiny L}})=v_{\text{\tiny L}}+v_{\text{\tiny L}}=2v_{\text{\tiny L}}.
\end{equation}

When buyer preferences $v_{\text{\tiny H}},v_{\text{\tiny L}}$ belong to Region II of Fig. \ref{regions_5}, based on Proposition \ref{Region2}, we have:
\begin{equation}
\tilde{\Pi}^{\textrm{ST}}=\frac{1}{4}(v_{\text{\tiny H}}+v_{\text{\tiny H}})+\frac{1}{4}(v_{\text{\tiny H}}+v_{\text{\tiny L}})+\frac{1}{4}(0+v_{\text{\tiny H}})+\frac{1}{4}(0+v_{\text{\tiny L}})=v_{\text{\tiny H}}+\frac{1}{2}v_{\text{\tiny L}}.  
\end{equation}

When buyer preferences $v_{\text{\tiny H}},v_{\text{\tiny L}}$ belong to Region III of Fig. \ref{regions_5}, based on Proposition \ref{Region3}, we have:
\begin{equation}\label{revenue_region3}
\tilde{\Pi}^{\textrm{ST}}=\frac{1}{4}\Pi^{\textrm{ST}}(v_1=v_2=v_{\text{\tiny H}})+\frac{1}{4}(v_{\text{\tiny H}}+v_{\text{\tiny L}})+\frac{1}{4}(0+v_{\text{\tiny H}})+\frac{1}{4}(0+v_{\text{\tiny L}}).
\end{equation}
Here, the seller's revenue extracted from two buyers with the same high preference is given by
\begin{equation}\label{revenue_region3_vh}
\Pi^{\textrm{ST}}(v_1=v_2=v_{\text{\tiny H}})=(1-\rho^*)^2(v_{\text{\tiny H}}+v_{\text{\tiny H}})+(2\rho^*-(\rho^*)^2)(v_{\text{\tiny H}}+v_{\text{\tiny L}})
\end{equation}
with manipulation probability $\rho^*$ in (\ref{caseBII_manipurho}). Specifically, the first term in (\ref{revenue_region3_vh}) (and (\ref{revenue_region4_vh}), (\ref{revenue_region5_vh}) below) is the revenue in the absence of manipulation, whereas the second term is in the presence of manipulation. Plugging (\ref{revenue_region3_vh}) into (\ref{revenue_region3}) and simplifying yields the seller's expected revenue. 

When buyer preferences $v_{\text{\tiny H}},v_{\text{\tiny L}}$ belong to Region IV of Fig. \ref{regions_5}, based on Case 1 of Proposition \ref{Region4}, we have:
\begin{equation}\label{revenue_region4}
\tilde{\Pi}^{\textrm{ST}}=\frac{1}{4}\Pi^{\textrm{ST}}(v_1=v_2=v_{\text{\tiny H}})+\frac{1}{4}(v_{\text{\tiny H}}+v_{\text{\tiny L}})+\frac{1}{4}(0+v_{\text{\tiny H}})+\frac{1}{4}\Pi^{\textrm{ST}}(v_1=v_2=v_{\text{\tiny L}}).
\end{equation}
Here, the seller's revenue extracted from two buyers with the same high preference is given by
\begin{equation}\label{revenue_region4_vh}
\Pi^{\textrm{ST}}(v_1=v_2=v_{\text{\tiny H}})=(1-\rho^*)^2\left(\beta(v_{\text{\tiny L}}+v_{\text{\tiny L}})+(1-\beta)(v_{\text{\tiny H}}+v_{\text{\tiny H}})\right)+(2\rho^*-(\rho^*)^2)(v_{\text{\tiny H}}+v_{\text{\tiny L}}),    
\end{equation}
and the revenue extracted from two with the same low preference is given by
\begin{equation}\label{revenue_region4_vl}
\Pi^{\textrm{ST}}(v_1=v_2=v_{\text{\tiny L}})=\beta(v_{\text{\tiny L}}+v_{\text{\tiny L}})+(1-\beta)(0+v_{\text{\tiny L}}),    
\end{equation}
where $\rho^*$ is the manipulation probability in (\ref{caseABII_manipurho2}) and $\beta$ is the mixed pricing parameter in (\ref{caseABII_beta}). Specifically, the first term in (\ref{revenue_region4_vl}) (and (\ref{revenue_region5_vhl1}), (\ref{revenue_region5_vhl2}) below) is the revenue extracted with uniform pricing, whereas the second term is with personalized pricing. Plugging (\ref{revenue_region4_vh}) and (\ref{revenue_region4_vl}) into (\ref{revenue_region4}) and simplifying yields the seller's expected revenue. 

When buyer preferences $v_{\text{\tiny H}},v_{\text{\tiny L}}$ belong to Region V of Fig. \ref{regions_5}, based on Case 2 of Proposition \ref{Region4}, we have:
\begin{equation}\label{revenue_region5}
\tilde{\Pi}^{\textrm{ST}}=\frac{1}{4}\Pi^{\textrm{ST}}(v_1=v_2=v_{\text{\tiny H}})+\frac{1}{4}\Pi^{\textrm{ST}}(v_1=v_{\text{\tiny H}},v_2=v_{\text{\tiny L}})+\frac{1}{4}\Pi^{\textrm{ST}}(v_1=v_{\text{\tiny L}},v_2=v_{\text{\tiny H}})+\frac{1}{4}(0+v_{\text{\tiny L}}).
\end{equation}
Here, the seller's revenue extracted from two buyers with the same high preference is given by
\begin{equation}\label{revenue_region5_vh}
\Pi^{\textrm{ST}}(v_1=v_2=v_{\text{\tiny H}})=(1-\rho^*)^2(v_{\text{\tiny H}}+v_{\text{\tiny H}})+(2\rho^*-(\rho^*)^2)(\beta(v_{\text{\tiny H}}+v_{\text{\tiny L}})+(1-\beta)(v_{\text{\tiny H}}+v_{\text{\tiny H}})),    
\end{equation}
and the revenue extracted from two buyers with different preferences is given by
\begin{equation}\label{revenue_region5_vhl1}
\Pi^{\textrm{ST}}(v_1=v_{\text{\tiny H}},v_2=v_{\text{\tiny L}})=(1-\beta)(v_{\text{\tiny H}}+0)+\beta(v_{\text{\tiny H}}+v_{\text{\tiny L}}),    
\end{equation}
and
\begin{equation}\label{revenue_region5_vhl2}
\Pi^{\textrm{ST}}(v_1=v_{\text{\tiny L}},v_2=v_{\text{\tiny H}})=(1-\beta)(v_{\text{\tiny H}}+0)+\beta(0+v_{\text{\tiny H}}),    
\end{equation}
where $\rho^*$ is the manipulation probability in (\ref{caseB23_manipurho2}) and $\beta$ is the mixed pricing parameter in (\ref{caseB23_beta}). Plugging (\ref{revenue_region5_vh}), (\ref{revenue_region5_vhl1}) and (\ref{revenue_region5_vhl2}) into (\ref{revenue_region5}) and simplifying yields the seller's expected revenue.

In conclusion, in the strategic-learning model, the seller's expected revenue is summarized as below.
\begin{equation}\label{Revenu_Strategic_Two}
\tilde{\Pi}=\left\{
\begin{aligned}
&2v_{\text{\tiny L}}, &\textrm{In Region I of Fig. \ref{regions_5},}\\
&v_{\text{\tiny H}}+\frac{1}{2}v_{\text{\tiny L}}, &\textrm{In Region II of Fig. \ref{regions_5},}\\
&\frac{3}{4}v_{\text{\tiny H}}+\frac{3}{4}v_{\text{\tiny L}}+\frac{(1-l)^2}{v_{\text{\tiny H}}-v_{\text{\tiny L}}},&\textrm{In Region III of Fig. \ref{regions_5},}\\
&\frac{3}{4}v_{\text{\tiny H}}+\frac{7}{8}v_{\text{\tiny L}},&\textrm{In Region IV of Fig. \ref{regions_5},}\\
&v_{\text{\tiny H}}+\frac{1}{4}v_{\text{\tiny L}},&\textrm{In Region V of Fig. \ref{regions_5}.}
\end{aligned}
\right.
\end{equation}

\subsection{Comparison of the seller's expected revenue}
This appendix compares the seller's expected revenue $\tilde{\Pi}$ with $\tilde{\Pi}^{0}$ in \eqref{PI_0_No} . Particularly, the revenue gap $\tilde{\Pi}^{\rm ST}-\tilde{\Pi}^{\rm NO}$ is positive as long as $0<v_{\text{\tiny L}}<\frac{2}{3}v_{\text{\tiny H}}$, i.e., buyers have diverse product preferences. Furthermore, given any fixed $v_{\text{\tiny H}}$, the \textit{revenue gain} $\frac{\tilde{\Pi}^{\textrm{ST}}-\tilde{\Pi}^{\textrm{NO}}}{\tilde{\Pi}^{\textrm{NO}}}$ first increases with $v_{\text{\tiny L}}\in\left(0,\frac{1}{2}v_{\text{\tiny H}}\right)$, then decreases with $v_{\text{\tiny L}}\in\left[\frac{1}{2}v_{\text{\tiny H}},\frac{2}{3}v_{\text{\tiny H}}\right)$ and finally remains as $1$ with $v_{\text{\tiny L}}\in\left[\frac{2}{3}v_{\text{\tiny H}},v_{\text{\tiny H}}\right)$. Particularly, we also prove that the maximum gain is attained when $v_{\text{\tiny H}}=2v_{\text{\tiny L}}\le4(1-l)$, which can reach $25\%$.

\section{Proof of Proposition \ref{salerevenue2}}

This appendix would compare the seller's expected revenue $\tilde{\Pi}$ in \eqref{Revenu_Strategic_Two} in the strategic-learning model with $\tilde{\Pi}^1$ in \eqref{PI_0_No} in the undisclosed-learning benchmark. More specifically, the expected revenue loss from the buyer awareness is summarized in Table \ref{Table_A6_J}, which is always non-negative and its maximum $1/12\approx8.3\%$ is attained when $v_{\text{\tiny L}}/v_{\text{\tiny H}}=2/5$ and $v_{\text{\tiny H}}-v_{\text{\tiny L}}\ge6(1-l)/\sqrt{3}$ (see Lemma \ref{lemmaA3_J}). This then completes the proof for Proposition \ref{salerevenue2}.

\begin{table}[h]
\centering
\caption{Expected Revenue Loss from Buyer Awareness}
\begin{tabular}{cp{2cm}p{2cm}p{2.5cm}p{2cm}p{2cm}}
\specialrule{\heavyrulewidth}{1.2pt}{0pt}
\cellcolor[gray]{0.92}{\makecell[c]{Region \# of Fig. \ref{regions_5}}} &\cellcolor[gray]{0.92}{\makecell[c]{I}} &\cellcolor[gray]{0.92}{\makecell[c]{II}} &\cellcolor[gray]{0.92}{\makecell[c]{III}} &\cellcolor[gray]{0.92}{\makecell[c]{IV}}&\cellcolor[gray]{0.92}{\makecell[c]{V}}\\
\specialrule{\lightrulewidth}{0pt}{1.2pt}
\makecell[c]{$\frac{\tilde{\Pi}^1-\tilde{\Pi}}{\tilde{\Pi}^1}$} &\makecell[c]{$0$} &\makecell[c]{$0$} &{\makecell[c]{$\frac{(v_{\text{\tiny H}}-v_{\text{\tiny L}})^2-4(1-l)^2}{(v_{\text{\tiny H}}-v_{\text{\tiny L}})\left(4v_{\text{\tiny H}}+2v_{\text{\tiny L}}\right)}$}} &\makecell[c]{$\frac{2v_{\text{\tiny H}}-3v_{\text{\tiny L}}}{8v_{\text{\tiny H}}+4v_{\text{\tiny L}}}$}&\makecell[c]{$\frac{v_{\text{\tiny L}}}{4v_{\text{\tiny H}}+2v_{\text{\tiny L}}}$}\\
\specialrule{\heavyrulewidth}{1.2pt}{1.2pt}
\end{tabular}
\label{Table_A6_J}
\end{table}

\begin{lemma}\label{lemmaA3_J}
In Region III of Fig. \ref{regions_5}, $(\tilde{\Pi}^1-\tilde{\Pi})/\tilde{\Pi}^1$ reaches its maximum $8.3\%$ when $v_{\text{\tiny H}}=10(1-l)/\sqrt{3}$ and $v_{\text{\tiny L}}=4(1-l)/\sqrt{3}$.
\end{lemma}
\begin{proof}
Let $x=v_{\text{\tiny H}}-v_{\text{\tiny L}}$ and $y=v_{\text{\tiny L}}$. Then, we substitute the changed variables to write
\begin{equation}\label{region3_ratioloss}
    \frac{\tilde{\Pi}^1-\tilde{\Pi}}{\tilde{\Pi}^1}=\frac{x^2-4(1-l)^2}{4x^2+6xy}, \text{ where } x\in\left[2(1-l),\frac{6(1-l)}{\sqrt{3}}\right] \text{ and } y\in\left[x-\frac{4(1-l)^2}{x},\frac{8(1-l)^2}{x}\right],
\end{equation}
which is always positive and decreases with $y$ for any fixed x. Then \eqref{region3_ratioloss} reaches its maximum at $y=x-4(1-l)^2/x$ for any fixed $x$. Substituting back leads to the fact that \eqref{region3_ratioloss} increases with $x$. This then completes the proof.
\end{proof}

\section{Proof of Proposition \ref{unaware_general}}

\subsection{Analysis of SPE}
This appendix will prove the SPE characterization and its uniqueness for the undisclosed-learning benchmark. Notice that Lemma \ref{pricebinary}, which restricts our attention to the seller's binary pricing choices in Stage II, remains valid here. Moreover, the buyers will never manipulate in Stage I and behave the same as in the no-learning benchmark (See Appendix \ref{Appendix:stepI_lemma2}).

First, we establish key properties of the seller's equilibrium pricing with the following lemma. Notice that we will focus on the information set where the seller still does not know any remaining buyer's individual preference; otherwise, only one trivial equilibrium exists in the corresponding subgame, where the seller directly enables personalized pricing for each arriving buyer in the subsequent periods of the corresponding subgame.

\begin{lemma}
In the SPE of the undisclosed-learning benchmark, suppose the game has reached a given information set in the selling period $t\neq N$ of Stage II, where the seller still does not know any remaining buyer's individual preference. We have the following properties for the seller's equilibrium pricing in the corresponding subgame: (1) if $p_t=v_{\text{\tiny H}}$, then $p_s=v_s$ for each subsequent period $s> t$; (2) if $p_t=v_{\text{\tiny L}}$, then $p_s=v_{\text{\tiny L}}$ for each subsequent period $s> t$.
\end{lemma}

\begin{proof}
(i) For Property (1), through $p_t=v_{\text{\tiny H}}$, the seller learns the $t$-th-arriving buyer's preference from his purchase decision $a_t$. Since all buyers never manipulate in Stage I, the seller then perfectly infers the remaining unknown buyers' preferences. As such, the seller will offer a personalized price exactly equal to each subsequent arriving buyer's preference. This then completes the proof for Property (1). 
(ii) For Property (2), since $p_t=v_{\text{\tiny L}}$, the seller still does not know any remaining buyer's individual preference at the beginning of the selling period $t+1$. Then, one candidate pricing scheme in the corresponding subgame is $p_{t+1}=v_{\text{\tiny H}}$ and $p_{s}=v_s$ for each $s>t+1$, which is developed based on Property (1) and generates an expected revenue of $(N-t)v_{\text{\tiny H}}/2+(N-t-1)v_{\text{\tiny L}}/2$ in the selling periods from $t+1$ to $N$. It follows that an expected revenue of $(N-t)v_{\text{\tiny H}}/2+(N-t+1)v_{\text{\tiny L}}/2$ is attained during the selling periods from $t$ to $N$. However, this candidate pricing scheme \mbox{$\{p_t=v_{\text{\tiny L}};p_{t+1}=v_{\text{\tiny H}}; p_{s}=v_s, \forall s>t+1\}$} is strictly dominated by another pricing scheme \mbox{$\{p_t=v_{\text{\tiny H}}; p_{s}=v_s,\forall s>t\}$}, which generates a strictly higher expected revenue of \mbox{$(N-t+1)v_{\text{\tiny H}}/2+(N-t)v_{\text{\tiny L}}/2$} in the selling periods from $t$ to $N$. Therefore, for the selling period $t+1$, the seller's equilibrium pricing strategy must be $p_{t+1}=v_{\text{\tiny L}}$. By induction, we then complete the proof for Property (2).
\end{proof}

Now, we can focus on the very first selling period to characterize the seller's equilibrium pricing. More specifically, the seller is indifferent between two candidate pricing schemes, $\{p_1=v_{\text{\tiny H}};p_t=v_t,\forall t>2\}$ and $\{p_t=v_{\text{\tiny L}},\forall t\}$. This first-period indifference condition then suggests ${v_{\text{\tiny L}}}/v_{\text{\tiny H}}=N/(N+1)$.

\subsection{Comparison of the seller's expected revenue}
This appendix first calculates the seller's expected revenue $\tilde{\Pi}^0$ in the no-learning benchmark (see the SPE characterization in Lemma \ref{lemma:nolearning}) and $\tilde{\Pi}^1$ in the undisclosed-learning benchmark:
\begin{equation}\label{general_unaware_revenue1}
\tilde{\Pi}^0=\left\{
\begin{aligned}
&Nv_{\text{\tiny H}}/2, &\quad\textrm{if $\ v_{\text{\tiny L}}/v_{\text{\tiny H}}< 1/2$,}\\
&Nv_{\text{\tiny L}}, &\quad\textrm{otherwise;}
\end{aligned}
\right.
\end{equation}
and
\begin{equation}\label{general_unaware_revenue2}
\tilde{\Pi}^1=\left\{
\begin{aligned}
&Nv_{\text{\tiny H}}/2+(N-1)v_{\text{\tiny L}}/2, &\quad\textrm{if $\ {v_{\text{\tiny L}}}/v_{\text{\tiny H}}< N/(N+1)$,}\\
&N{v_{\text{\tiny L}}}, &\quad\textrm{otherwise.}
\end{aligned}
\right.
\end{equation}
Now we compare $\tilde{\Pi}^0$ with $\tilde{\Pi}^1$, and the expected revenue gain from the seller's undisclosed learning practice is given by
\begin{equation}
\frac{\tilde{\Pi}^1-\tilde{\Pi}^0}{\tilde{\Pi}^0}=\left\{
\begin{aligned}
&0, &\quad\textrm{if $\ {v_{\text{\tiny L}}}/v_{\text{\tiny H}}\ge N/(N+1)$,}\\
&v_{\text{\tiny H}}/2v_{\text{\tiny L}}-(N+1)/2N, &\quad\textrm{if $\ 1/2\le{v_{\text{\tiny L}}}/v_{\text{\tiny H}}<N/(N+1)$,}\\
&(N-1)v_{\text{\tiny L}}/Nv_{\text{\tiny H}}, &\quad\textrm{otherwise;}
\end{aligned}
\right.
\end{equation}
which is always non-negative; its maximum $(N-1)/2N$ is attained when ${v_{\text{\tiny L}}}/v_{\text{\tiny H}}=1/2$. This then completes the proof.

\section{Proof of Proposition \ref{unaware_general_known}}
The analysis rationale is demonstrated in the paragraph after Proposition \ref{unaware_general_known} in Section \ref{subsec:multi_buyer_benchmark}. It immediately follows that the seller's expected revenue extracted from the unknown buyers is $N(v_{\text{\tiny L}}+v_{\text{\tiny H}})/2$ in the presence of known buyers. Recall that the expected revenue counterpart in the undisclosed-learning benchmark without known buyers is given in \eqref{general_unaware_revenue1} and \ref{general_unaware_revenue2}; we have that the revenue improvement ratio from prior knowledge is always positive and reaches its maximum $1/2N$ when $v_{\text{\tiny L}}/v_{\text{\tiny H}}=N/(N+1)$.

\section{Proof of Proposition \ref{multiple:buyer_manipulation}}
The analyses in Appendix \ref{proof_lemma_buyer} and Appendix \ref{proof_binary} remain independent of the number of interconnected buyers and the amount of correlation information involved. Therefore, Lemmas \ref{lemma1_buyer} and \ref{pricebinary} continue to hold in a general social network of multiple interconnected unknown buyers, and Proposition \ref{multiple:buyer_manipulation} then follows.

\section{Proof of Lemma \ref{Lemma:price_prior_knowledge_VL_three}}
We first analyze the users' equilibrium manipulation decisions in Stage I. Specifically, similar to the analysis for Lemma~\ref{lemma1_buyer} (see Appendix \ref{proof_lemma_buyer}), low-preference buyers $i$ and $j$ in Fig. \ref{Ring_L_a} and Fig. \ref{Ring_L_b} never manipulate. Furthermore, since buyer $k$ has no incentive to manipulate, the common interaction frequency signal in \eqref{hat} is always $\hat{x}_{ik}=\hat{x}_{jk}=0$ for high-preference buyers $i$ and $j$ in Fig. \ref{Ring_L_b} and Fig. \ref{Ring_L_c}. As such, these high-preference buyers $i$ and $j$ cannot manipulate their preference correlation with buyer $k$ unilaterally, and thus, they also do not manipulate their social interactions with buyer $k$. The seller then perfectly infers the preferences of buyers $i$ and $j$ through their low-preference friend $k$. Therefore, buyers $i$ and $j$ also have no incentive to manipulate their social interactions with each other. Following the above discussions, we immediately arrive at the seller's equilibrium pricing in Stage II. More specifically, after perfectly learning all buyers' preferences from Stage I, the seller sets a personalized price equal to the arriving buyer's preference in each selling period of Stage II, i.e., $p_t^*=v_t$ for each $t\in\{1,2\}$. Then, the seller always extracts the maximum surplus from buyers $i$ and $j$ in the equilibrium, with the corresponding expected revenue given by
\begin{equation}\label{Equ:Revenue:prior_knowledge_VL}
    \tilde{\Pi}_{\{i,j\}}=\frac{1}{4}(v_{\text{\tiny L}}+v_{\text{\tiny L}})+\frac{1}{2}(v_{\text{\tiny L}}+v_{\text{\tiny H}})+\frac{1}{4}(v_{\text{\tiny H}}+v_{\text{\tiny H}})=v_{\text{\tiny L}}+v_{\text{\tiny H}}.
\end{equation}
By now, we have completely outlined the PBE and established its uniqueness. Finally, we will compare the seller's expected revenue in \eqref{Equ:Revenue:prior_knowledge_VL} ($\tilde{\Pi}_{\{i,j\}}$, with prior knowledge $v_k=v_{\text{\tiny L}}$) to \eqref{Revenu_Strategic_Two} ($\tilde{\Pi}$, without any prior knowledge). The prior knowledge of $v_k=v_{\text{\tiny L}}$ generates a revenue improvement ratio as summarized in Table \ref{tab:A7}, which is always positive and its maximum $3/11\approx27.3\%$ is attained when $v_{\text{\tiny L}}/v_{\text{\tiny H}}=2/5$ and $v_{\text{\tiny H}}\ge 10(1-l)/\sqrt{3}$. This then completes the proof for Lemma \ref{Lemma:price_prior_knowledge_VL_three}.

\begin{table}[h]
\centering
\caption{Revenue Improvement Ratio Generated from The prior knowledge of $v_k=v_{\text{\tiny L}}$}
\begin{tabular}{cp{2cm}p{2cm}p{2.5cm}p{2cm}p{2cm}}
\specialrule{\heavyrulewidth}{1.2pt}{0pt}
\cellcolor[gray]{0.92}{\makecell[c]{Region \# of Fig. \ref{regions_5}}} &\cellcolor[gray]{0.92}{\makecell[c]{I}} &\cellcolor[gray]{0.92}{\makecell[c]{II}} &\cellcolor[gray]{0.92}{\makecell[c]{III}} &\cellcolor[gray]{0.92}{\makecell[c]{IV}}&\cellcolor[gray]{0.92}{\makecell[c]{V}}\\
\specialrule{\lightrulewidth}{0pt}{1.2pt}
\makecell[c]{$\frac{\tilde{\Pi}_{\{i,j\}}-\tilde{\Pi}}{\tilde{\Pi}}$} &\makecell[c]{$\frac{v_{\text{\tiny H}}-v_{\text{\tiny L}}}{2v_{\text{\tiny L}}}$} &\makecell[c]{$\frac{v_{\text{\tiny L}}}{2v_{\text{\tiny H}}+v_{\text{\tiny L}}}$} &{\makecell[c]{$\frac{v_{\text{\tiny H}}^2-v_{\text{\tiny L}}^2-4(1-l)^2}{3(v_{\text{\tiny H}}^2-v_{\text{\tiny L}}^2)+4(1-l)^2}$}} &\makecell[c]{$\frac{2v_{\text{\tiny H}}+v_{\text{\tiny L}}}{6v_{\text{\tiny H}}+7v_{\text{\tiny L}}}$}&\makecell[c]{$\frac{3v_{\text{\tiny L}}}{4v_{\text{\tiny H}}+v_{\text{\tiny L}}}$}\\
\specialrule{\heavyrulewidth}{1.2pt}{1.2pt}
\end{tabular}
\label{tab:A7}
\end{table}

\section{Proof of Lemma \ref{High_known:no_benefit}}

\subsection{Complete Characterization of All Pure-strategy PBE}\label{High-known:Pure-PBE}

This appendix presents a formal characterization of all pure-strategy PBE for the three-buyer network instance with prior knowledge of $v_k=v_{\text{\tiny H}}$ in Section \ref{subsec:three-buyer}. Similar to the analysis for Lemma~\ref{lemma1_buyer} (see Appendix~\ref{proof_lemma_buyer}), the unknown low-preference buyers never manipulate in the equilibrium.

\begin{itemize}
    \item \textbf{Case I} with $v_{\text{\tiny H}}-v_{\text{\tiny L}}\ge 2(1-l)$ and $v_{\text{\tiny L}}/v_{\text{\tiny H}}\ge 1/3$. In Stage I, if $v_i\neq v_j=v_{\text{\tiny H}}$, $x^*_{jk}=1$; if $v_i=v_j=v_{\text{\tiny H}}$, $x^*_{ij}=x^*_{ik}=1$ and $x^*_{ji}=x^*_{jk}=0$. In Stage II, $p_1^*=p_2^*=v_{\text{\tiny L}}$ if $\{\hat{x}^*_{ij},\hat{x}^*_{ik},\hat{x}^*_{jk}\}=\{1,0,0\}$; $p_i^*=v_{\text{\tiny L}}$ if $\{\hat{x}^*_{ij},\hat{x}^*_{ik},\hat{x}^*_{jk}\}=\{0,0,1\}$.
    \item \textbf{Case II} with $v_{\text{\tiny H}}-v_{\text{\tiny L}}< 2(1-l)$. In Stage I, if $v_i\neq v_j=v_{\text{\tiny H}}$, $x^*_{jk}=1$; if $v_i=v_j=v_{\text{\tiny H}}$, $x^*_{ij}=x^*_{ik}=1$ and $x^*_{ji}=x^*_{jk}=1$. In Stage II, $p_1^*=p_2^*=v_{\text{\tiny L}}$ if $\{\hat{x}^*_{ij},\hat{x}^*_{ik},\hat{x}^*_{jk}\}=\{1,0,0\}$; $p_i^*=v_{\text{\tiny L}}$ if $\{\hat{x}^*_{ij},\hat{x}^*_{ik},\hat{x}^*_{jk}\}=\{0,0,1\}$; $p_1^*=p_2^*=v_{\text{\tiny H}}$ if $\{\hat{x}^*_{ij},\hat{x}^*_{ik},\hat{x}^*_{jk}\}=\{1,1,1\}$.
    \item \textbf{Case III} with $v_{\text{\tiny H}}-v_{\text{\tiny L}} >1-l$ and $v_{\text{\tiny L}}/v_{\text{\tiny H}}> 2/3$. In Stage I, if $v_i\neq v_j=v_{\text{\tiny H}}$, $x^*_{jk}=1$; if $v_i=v_j=v_{\text{\tiny H}}$, $x^*_{ij}=x^*_{ji}=1$ and $x^*_{ik}=x^*_{jk}=0$. In Stage II, $p_1^*=p_2^*=v_{\text{\tiny L}}$ if $\{\hat{x}^*_{ij},\hat{x}^*_{ik},\hat{x}^*_{jk}\}=\{1,0,0\}$; $p_i^*=v_{\text{\tiny L}}$ if $\{\hat{x}^*_{ij},\hat{x}^*_{ik},\hat{x}^*_{jk}\}=\{0,0,1\}$.
    \item \textbf{Case IV} with $v_{\text{\tiny H}}-v_{\text{\tiny L}} >1-l$ and $v_{\text{\tiny L}}/v_{\text{\tiny H}}> 2/3$. In Stage I, if $v_i\neq v_j=v_{\text{\tiny H}}$, $x^*_{jk}=0$; if $v_i=v_j=v_{\text{\tiny H}}$, $x^*_{ij}=x^*_{ik}=1$ and $x^*_{ji}=x^*_{jk}=1$. In Stage II, $p_1^*=p_2^*=v_{\text{\tiny L}}$ if $\{\hat{x}^*_{ij},\hat{x}^*_{ik},\hat{x}^*_{jk}\}=\{1,0,0\}$; $p_1^*=p_2^*=v_{\text{\tiny L}}$ if $\{\hat{x}^*_{ij},\hat{x}^*_{ik},\hat{x}^*_{jk}\}=\{0,0,0\}$.
    \item \textbf{Case V} with $v_{\text{\tiny H}}-v_{\text{\tiny L}} >1-l$ and $v_{\text{\tiny L}}/v_{\text{\tiny H}}> 2/3$. In Stage I, if $v_i\neq v_j=v_{\text{\tiny H}}$, $x^*_{jk}=0$; if $v_i=v_j=v_{\text{\tiny H}}$, $x^*_{ij}=x^*_{ji}=1$ and $x^*_{ik}=x^*_{jk}=0$. In Stage II, $p_1^*=p_2^*=v_{\text{\tiny L}}$ if $\{\hat{x}^*_{ij},\hat{x}^*_{ik},\hat{x}^*_{jk}\}=\{1,0,0\}$; $p_1^*=p_2^*=v_{\text{\tiny L}}$ if $\{\hat{x}^*_{ij},\hat{x}^*_{ik},\hat{x}^*_{jk}\}=\{0,0,0\}$.
\end{itemize}

\subsection{Proof for Lemma \ref{High_known:no_benefit}}

This proof immediately follows Case V in Appendix \ref{High-known:Pure-PBE}, where the seller's expected revenue from unknown buyers $i$ and $j$ is $2v_{\text{\tiny L}}$. This is the same as that obtained in the absence of known buyers (see \eqref{Revenu_Strategic_Two} in Region I of Fig. \ref{regions_5}).

\section{Proof of Proposition \ref{High_known:no_benefit:general}}\label{Appendix:High_known:no_benefit:general}

The proof of Proposition \ref{High_known:no_benefit:general} is constructive, which consists of two main steps. More specifically, Appendix \ref{Appendix:High_known:no_benefit:general_1} starts from constructing the equilibrium strategy profile and identifying the sufficient conditions to ensure its existence; Appendix \ref{Appendix:High_known:no_benefit:general_2} then compares the seller's expected revenue in this equilibrium to that in the absence of known high-preference buyers.

\subsection{Step I: Construction of the equilibrium strategy profile}\label{Appendix:High_known:no_benefit:general_1}

This appendix first constructs a strategy profile and then establishes the sufficient conditions for it to constitute the PBE. Particularly, our construction idea is inspired by the three-buyer network instance analysis in Lemma \ref{High_known:no_benefit} (see Fig. \ref{Ring_H_NoRevenueGain}). 

Recall that we consider an arbitrary number of known high-preference buyers connected with $N$ unknown buyers in the online social network. Now, we construct the strategy profile as follows. \textbf{(i) Buyer Strategies:} In Stage I, any unknown low-preference buyer never manipulates his social interactions with any other buyer; for each unknown high-preference buyer, he manipulates his social interaction with each connected known high-preference buyer into a low interaction frequency, while remaining honest when interacting with any unknown buyer. \textbf{(ii) Belief:} In Stage II, the seller can perfectly infer the preference correlation between any unknown buyer pair from the corresponding social interactions. However, the seller cannot learn any individual unknown buyer's preference at the beginning of Stage II. \textbf{(ii) Seller Strategies:} Moreover, in Stage II, the seller charges a low uniform price in each selling period, i.e., $p_t^*=v_{\text{\tiny L}}$ for each $t\in\{1,2,\dots,N\}$. 

Next, we will prove that this constructed strategy profile constitutes a pure-strategy PBE under the condition of \eqref{High_known:no_benefit:general_condition}. More specifically, given the constructed belief (ii) in Stage II, the seller's pricing optimization problem for the $N$ unknown buyers degenerates into that in the undisclosed-learning benchmark without any known buyer (see Proposition \ref{unaware_general}). With $v_{\text{\tiny L}}/v_{\text{\tiny H}}>N/(N+1)$, as indicated by the left-hand-side inequality in \eqref{High_known:no_benefit:general_condition}, we thus prove that the constructed strategy profile (iii) constitutes the seller's equilibrium pricing strategies (i.e., sequential rationality, see Large $v_{\text{\tiny L}}$ regime of Proposition~\ref{unaware_general}). On the flip side, since $p_t^*=v_{\text{\tiny L}}$ for each $t\in\{1,2,\dots,N\}$, an unknown high-preference buyer always receives the maximal purchase surplus $v_{\text{\tiny H}}-v_{\text{\tiny L}}$. As such, if an unknown high-preference buyer does not connect with any known buyer, unilateral deviation from the strategy profile (i) to manipulate social interactions cannot help improve the purchase surplus; thus, this unknown high-preference buyer has no incentive to manipulate. Otherwise, consider an unknown high-preference buyer $i$ connects with $K_i$ known buyers. Recall that an unknown high-preference buyer cannot manipulate his preference correlation with another unknown buyer unilaterally; hence, we will focus on his manipulation strategies with connected known high-preference buyers in what follows. To guarantee this buyer $i$ has no incentive to deviate from the strategy profile (i) where he chooses to manipulate his social interactions with each connected known buyer into a low frequency, this buyer $i$'s payoff should satisfy the following inequality:
\begin{multline*}
\tilde{\pi}_i(\text{Manipulate with each connected known buyer},\boldsymbol{x}^*_{-i})\\> \tilde{\pi}_i(\text{Honest with at least one connected known buyer},\boldsymbol{x}^*_{-i}),
\end{multline*}
which immediately yields $v_{\text{\tiny H}}-v_{\text{\tiny L}}>K_i(1-l)$, i.e., his potential purchase gain outweighs the social interaction loss from manipulation. Moreover, this holds for any such buyer if the following condition holds:
\begin{equation}
v_{\text{\tiny H}}-v_{\text{\tiny L}}>\max_{0\le K_i\le K}{K_i(1-l)}=K(1-l),
\end{equation}
which is exactly the right-hand-side inequality in \eqref{High_known:no_benefit:general_condition}.
Also, similar to the analysis for Lemma~\ref{lemma1_buyer} (see Appendix~\ref{proof_lemma_buyer}), those unknown low-preference buyers never manipulate in the equilibrium. In conclusion, the strategy profile (i) constitutes the buyers' equilibrium social interaction decisions. Lastly, given the buyers' strategy profile (i) in Stage I, the common social interaction frequency between any connected known and unknown buyer pair appears the same, i.e., $\hat{x}=0$. It then follows that the seller's updated belief in Stage II is consistent with the constructed belief (ii) (i.e., belief consistency). By now, we have shown that the constructed strategy profile (i)-(iii) is a pure-strategy PBE under \eqref{High_known:no_benefit:general_condition}.

\subsection{Step II: Comparison of the seller's expected revenue}\label{Appendix:High_known:no_benefit:general_2}

In this appendix, we first calculate the seller's expected revenue extracted from the $N$ unknown buyers in the constructed PBE in the presence of known high-preference buyers, which is $Nv_{\text{\tiny L}}$ due to the equilibrium uniform pricing. In the absence of known buyers, under the condition of \eqref{High_known:no_benefit:general_condition}, Lemma \ref{general_uniformPBE} suggests that the seller's expected revenue from those unknown buyers is again $Nv_{\text{\tiny L}}$ due to the equilibrium uniform pricing. We thus complete the proof for Proposition \ref{High_known:no_benefit:general}.

\begin{lemma}\label{general_uniformPBE}
Given $N$ interconnected unknown buyers in the online social network, if $v_{\text{\tiny L}}/v_{\text{\tiny H}}>N/(N+1)$, then there always exists a pure-strategy PBE that is the same as that in the Large $v_{\text{\tiny L}}$ regime of Proposition~\ref{unaware_general}.
\end{lemma}
\begin{proof}
Recall the Large $v_{\text{\tiny L}}$ regime of Proposition~\ref{unaware_general}; we first outline the equilibrium strategy profile in Lemma \ref{general_uniformPBE} as follows. \textbf{(i) Strategies:} Buyers never manipulate while the seller charges a low uniform price in each selling period. \textbf{(ii) Beliefs:} In Stage II, the seller perfectly infers the preference correlation between any unknown buyer pair from the corresponding social interactions. However, the seller cannot learn any individual buyer's preference at the beginning of Stage II. 

Next, we will prove that this strategy profile constitutes a pure-strategy PBE here. Specifically, given the belief (ii) in Stage II, the seller's pricing problem for the unknown buyers degenerates into that in the undisclosed-learning benchmark without any known buyer (see Proposition \ref{unaware_general}). Given $v_{\text{\tiny L}}/v_{\text{\tiny H}}>N/(N+1)$, we therefore prove that the strategy profile (i) constitutes the seller's equilibrium pricing strategies (i.e., sequential rationality). On the other hand, in Stage I, an unknown high-preference buyer cannot manipulate his preference correlation with another buyer unilaterally; thus, the unknown high-preference buyer has no incentive to manipulate. Also, similar to the analysis for Lemma~\ref{lemma1_buyer} (see Appendix~\ref{proof_lemma_buyer}), the unknown low-preference buyers never manipulate in the equilibrium. In conclusion, the strategy profile (i) constitutes the buyers' equilibrium social interaction decisions. Lastly, given the buyers' strategy profile (i) in Stage I, it follows that the seller's updated belief in Stage II is consistent with the belief (ii) (i.e., belief consistency). This then completes the proof.
\end{proof}

\section{Extension to Buyers' Continuous Preferences}\label{Appendix:continuous}

This appendix extends to explore the buyers' continuous preferences to demonstrate the robustness of our major results.

\subsection{Extension to Buyers' Continuous Preferences}\label{Appendix:continuous:contents}

In the main text, we have assumed buyers' discrete binary preferences $v_{\text{\tiny L}}$ and $v_{\text{\tiny H}}$ (see Section \ref{model_stageI}). In this section, we extend to a more general case of two-class \textit{continuous preferences}, to show the robustness of our key insights. Consider that each buyer $i$'s preference $v_i$ is uniformly distributed in the range $[0,\bar{v}]$. To model any two buyers' social interaction utilities according to their preference correlation in Table \ref{table}, we categorize buyers into two preference classes: \textit{low-end class} with their preferences in the range $[0, \bar{v}/2]$, and \textit{high-end class} in $[\bar{v}/2, \bar{v}]$. According to the \textit{homophily} theory, if the two buyers belong to the same class, their frequent interactions lead to positive social utilities in Table \ref{table_same}; otherwise, the interactions lead to negative utilities in Table \ref{table_differ}. 

Similar to our analysis in Sections \ref{PBE1} and \ref{PBE2}, we can also characterize the PBE here. However, one new difficulty is that the seller's equilibrium pricing decision is no longer binary in Lemma \ref{pricebinary}. Furthermore, the seller cannot directly infer the first-arriving buyer's preference class from his purchase decision. For example, even if the first-arriving buyer purchases under a price $p_1<\bar{v}/2$, he may not belong to the high-end preference class. In fact, he could still be in the low-end class with probability $\frac{\bar{v}-2p_1}{2(\bar{v}-p_1)}$. As such, the belief update of buyers' preferences after the first selling period contains some uncertainty, and we need to optimize the seller's expected revenue with such a \textit{noisy belief update}. It turns out that Lemma \ref{lemma1_buyer}'s result still holds here: \textit{no manipulation happens between two buyers in the same low-end class or in different preference classes.} As discussed in the following proposition, we only observe possible social manipulation between two high-end buyers.

\begin{proposition}\label{continuous_proposition}
Under the continuous-preference model, there exists a PBE as follows, depending on the maximum preference $\bar{v}$.
\begin{itemize}
    \item \emph{\textbf{Case 1} with low $\bar{v}\in\left(0,8(1-l)\right]$:} \emph{In Stage I}, two buyers ($i$ and $j$) in the high-end preference class do not manipulate their social interactions. \emph{In Stage II}, the seller charges $p_1^*=\frac{\bar{v}}{2}$ in the first selling period, and personalizes the second-period price based on the observed buyers' common interaction frequency $\hat{x}^*=\min\{x_{ij}^*,x_{ji}^*\}$ and the first-period purchase record $a_1^*$:
    \begin{equation}
    p_2^*(\hat{x}^*,a_1^*)=\left\{
    \begin{array}{ll}
    \frac{\bar{v}}{2}\boldsymbol{1}(a_1^*=1)+\frac{\bar{v}}{4}\boldsymbol{1}(a_1^*=0),&\quad\textrm{if $\ \hat{x}^\ast=1$.}\\
    \frac{\bar{v}}{4}\boldsymbol{1}(a_1^*=1)+\frac{\bar{v}}{2}\boldsymbol{1}(a_1^*=0),&\quad\textrm{if $\ \hat{x}^\ast=0$.}\label{continuous_condition_1}
    \end{array}
    \right.    
    \end{equation}
    
    \item \emph{\textbf{Case 2} with medium $\bar{v}\in\left(8(1-l),\frac{64}{3\sqrt{3}}(1-l)\right]$:} \emph{In Stage I}, two buyers ($i$ and $j$) in the high-end preference class manipulate with the same probability $\rho^*=1-2\sqrt[3]{(1-l)/\bar{v}}$. \emph{In Stage II}, the seller charges $p_1^*=\frac{\bar{v}}{2}$ in the first selling period, and personalizes the second period price as follows,
    \begin{equation}
    p_2^*(\hat{x}^*,a_1^*)=\left\{
    \begin{array}{ll}
    \frac{\bar{v}}{2}\boldsymbol{1}(a_1^*=1)+\frac{\bar{v}}{4}\boldsymbol{1}(a_1^*=0),&\quad\textrm{if $\ \hat{x}^\ast=1$.}\\
    \frac{\bar{v}}{4}\left(2-4\left((1-l)/\bar{v}\right)^\frac{2}{3}\right)\boldsymbol{1}(a_1^*=1)+\frac{\bar{v}}{2}\boldsymbol{1}(a_1^*=0),&\quad\textrm{if $\ \hat{x}^\ast=0$.}\label{continuous_condition_2}
    \end{array}
    \right.    
    \end{equation}
    
   \item \emph{\textbf{Case 3} with large $\bar{v}\in\left(\frac{64}{3\sqrt{3}}(1-l),\infty\right)$:} \emph{In Stage I}, two buyers ($i$ and $j$) in the high-end preference class manipulate with the same probability $\rho^*$, which is the unique solution to $8\bar{v}{(1-\rho^*)}^3-3\bar{v}(1-\rho^*)+32l-32=0$. \emph{In Stage II}, the seller charges $p_1^*=\frac{9-8\rho^*+4(\rho^*)^2}{16}\bar{v}<\frac{\bar{v}}{2}$ when observing $\hat{x}^*=1$ and $p_1^*=\frac{\bar{v}}{2}$ when observing $\hat{x}^*=0$ in the first selling period, and personalizes the second period price as follows,
    \begin{equation}
    p_2^*(\hat{x}^*,a_1^*)=\left\{
    \begin{array}{ll}
    \frac{\bar{v}}{2}\boldsymbol{1}(a_1^*=1)+\frac{\bar{v}}{4}\boldsymbol{1}(a_1^*=0),&\quad\textrm{if $\ \hat{x}^\ast=1$.}\\
    \frac{\bar{v}}{4}\left(1+2\rho^*-(\rho^*)^2\right)\boldsymbol{1}(a_1^*=1)+\frac{\bar{v}}{2}\boldsymbol{1}(a_1^*=0),&\quad\textrm{if $\ \hat{x}^\ast=0$.}\label{continuous_condition_3}
    \end{array}
    \right.    
    \end{equation}
\end{itemize}
\label{continuous_PBE}
\end{proposition}

The PBE here follows a similar structure to the PBE in Propositions \ref{region1}-\ref{Region4}. As $\bar{v}$ increases from Case 1 to Case 3, the difference between the two preference classes becomes more obvious. Then, once viewed as low-end class buyers for pricing, high-end buyers' potential purchase gain becomes larger to motivate their manipulation in Stage I. 

Besides non-binary pricing, we also highlight some new insights of this continuous model as follows. Specifically, in Cases 1 and 2 of Proposition \ref{continuous_PBE}, the seller first charges $p_1^*=\frac{\bar{v}}{2}$ to perfectly determine the first-arriving buyer's preference class through his purchase decision. Then, the seller learns the latter-arriving buyer's preference class by observing $\hat{x}^*$, and tailors the second-period price $p_2^*$ accordingly. However, in Case 3, the seller finds it costly to just learn the first-arriving buyer's preference class at $p_1^*=\frac{\bar{v}}{2}$, and may thus deviate to $p_1^*<\frac{\bar{v}}{2}$ to meet more demands in low-end class. Note that even if the seller knows a buyer's preference class, she does not know the \textit{precise} preference point to exercise perfect personalized pricing. 


\begin{proposition}\label{continuous_buyer}
Under the continuous-preference model, when aware of the seller's learning, the average buyer payoff may reduce with social manipulation. Particularly, the expected payoff of a buyer in the low-end preference class always reduces with such manipulation.
\end{proposition}

Regarding the impact of manipulation on buyers' payoff, Proposition \ref{continuous_buyer} is consistent with Proposition \ref{dilemma_buyer} under the simplified discrete-binary-preference model. Furthermore, we show that low-end buyers bear a surplus loss from high-end buyers' manipulation, as the seller fails to meet some demands from the low-end buyers when mitigating the potential loss from such manipulation.

Finally, we turn to examine the seller's revenue to confirm another main insight behind Propositions \ref{salerevenue1} and \ref{salerevenue2}. Under this continuous-preference model, we can show that the seller gains $12.5\%$ more expected revenue in the undisclosed-learning benchmark than the no-learning benchmark. Yet the seller's expected revenue only slightly decreases (by $3.9\%$ at most) with buyers' learning-aware manipulation, but still outperforms the no-learning benchmark. Therefore, the suggestion for the seller to inform buyers of the social data access and the follow-up pricing continues to hold here.

\subsection{Proof for Appendix \ref{Appendix:continuous:contents}}
In this appendix, we conduct extended analysis under the continuous-preference model as introduced above. In what follows, we start with \textit{forward analysis} to propose the seller's posterior belief of buyers' manipulation structure in Stage I. Then, we proceed with \textit{backward induction} to characterize the equilibrium of the continuous-preference model.

\subsubsection{Analysis for Seller's Learning and Pricing in Stage II}

We restrict our attention to PBEs where only two buyers in the same high-end preference class manipulate, and all these high-end buyers manipulate with the same probability. Particularly, this aligns well with what we discuss in Sections \ref{PBE1} and \ref{PBE2}. This is also motivated by the fact that these high-end buyers face the same tradeoff when deciding to manipulate: paying as if viewed as low-end buyers, or paying with personalized prices but enjoying the full social interaction benefit. In equilibrium, as we will show, all these high-end buyers will purchase the product in Stage II. Hence, the benefits of social manipulation are independent of high-end buyers' preferences. Instead, such buyers' preference level only determines their purchase surplus and the final payoff. 

Similarly, we generally allow these buyers to choose a mixed strategy, i.e., choosing the low interaction frequency $0$ with a manipulation probability and choosing the high frequency $1$ with the complementary probability. As both high-end buyers are symmetric, they would choose the same \textit{manipulation probability}. As such, we formally define the manipulation probability for high-end buyers as below.

\begin{definition}
When interacting with the other high-end buyer $j$ ($i$, respectively), we define high-end buyer $i$'s ($j$'s, respectively) manipulation probability for the event that he chooses low social interaction frequency $x_{ij}=0$ ($x_{ji}=1$, respectively) as $\rho$. i.e.,
\begin{equation}\label{continuous_rho}
    \rho\triangleq {\rm{Pr}}\left(x_{ij}(v_i,v_j)=0 \right)={\rm{Pr}}\left(x_{ji}(v_j,v_i)=0 \right),\quad\forall v_i,v_j\in[\bar{v}/2,\bar{v}].
\end{equation}
\end{definition}

Similar to (\ref{belief1}) and (\ref{belief0}) in Section \ref{PBE1}, to enable the seller's personalized pricing in Stage II, we now analyze the seller's posterior belief of buyers' preference classes by learning from their common interaction frequency $\hat{x}$. According to Bayes' theorem, we have for the seller:
\begin{equation}\label{continuous_belief1}
\left\{
\begin{aligned}
&{\rm{Pr}}(v_i\in\left[\bar{v}/2,\bar{v}\right],v_j\in\left[\bar{v}/2,\bar{v}\right]|\hat{x}=1)=\frac{{\rm{Pr}}(\hat{x}=1|v_i\in\left[\bar{v}/2,\bar{v}\right],v_j\in\left[\bar{v}/2,\bar{v}\right])}{{\rm{Pr}}(\hat{x}=1|v_i\in\left[\bar{v}/2,\bar{v}\right],v_j\in\left[\bar{v}/2,\bar{v}\right])+1},\\
&{\rm{Pr}}(v_i\in\left[0,\bar{v}/2\right],v_j\in\left[0,\bar{v}/2\right]|\hat{x}=1)=\frac{1}{{\rm{Pr}}(\hat{x}=1|v_i\in\left[\bar{v}/2,\bar{v}\right],v_j\in\left[\bar{v}/2,\bar{v}\right])+1}.
\end{aligned}
\right.
\end{equation}
and
\begin{equation}\label{continuous_belief0}
\left\{
\begin{aligned}
&{\rm{Pr}}(v_i\in\left[\bar{v}/2,\bar{v}\right],v_j\in\left[\bar{v}/2,\bar{v}\right]|\hat{x}=0)=\frac{1-{\rm{Pr}}(\hat{x}=1|v_i\in\left[\bar{v}/2,\bar{v}\right],v_j\in\left[\bar{v}/2,\bar{v}\right])}{3-{\rm{Pr}}(\hat{x}=1|v_i\in\left[\bar{v}/2,\bar{v}\right],v_j\in\left[\bar{v}/2,\bar{v}\right])},\\
&{\rm{Pr}}(v_i\in\left[0,\bar{v}/2\right],v_j\in\left[\bar{v}/2,\bar{v}\right]|\hat{x}=0)=\frac{2}{3-{\rm{Pr}}(\hat{x}=1|v_i\in\left[\bar{v}/2,\bar{v}\right],v_j\in\left[\bar{v}/2,\bar{v}\right])}.
\end{aligned}
\right.
\end{equation}

Based on the posterior beliefs in (\ref{continuous_belief1}) and (\ref{continuous_belief0}), we now proceed with \textit{backward induction} to analyze the seller's optimal pricing in Stage II in the following lemmas. Lemma \ref{continuous_price_lemma1} below first expands when the seller observes buyers' high common interaction frequency $\hat{x}=1$.

\begin{lemma}\label{continuous_price_lemma1}
Under the continuous-preference model, given the posterior belief in (\ref{continuous_belief1}), when observing buyers' high common interaction frequency $\hat{x}=1$, the seller's optimal pricing in Stage II depends on high-end buyers' manipulation probability $\rho$ in (\ref{continuous_rho}) as follows.
\begin{itemize}
    \item \emph{Case 1 with $(1-\rho)^2\ge3/4$:} The seller charges $p_1=\frac{\bar{v}}{2}$ in the first selling period, and personalizes the second-period price based on the first-period purchase record $a_1$:
    \begin{equation}
    p_2(a_1)=
    \frac{\bar{v}}{2}\boldsymbol{1}(a_1=1)+\frac{\bar{v}}{4}\boldsymbol{1}(a_1=0),
    \end{equation}
    \item \emph{Case 2 with $1/4\le(1-\rho)^2<3/4$:} The seller charges $p_1=\frac{9-8\rho+4\rho^2}{16}\bar{v}$ in the first selling period, and personalizes the second-period price based on the first-period purchase record $a_1$:
    \begin{equation}
    p_2(a_1)=
    \frac{\bar{v}}{2}\boldsymbol{1}(a_1=1)+\frac{\bar{v}}{4}\boldsymbol{1}(a_1=0),
    \end{equation}
    \item \emph{Case 3 with $(1-\rho)^2<1/4$:} The seller charges $p_1$ in the first selling period, where $p_1$ is the unique solution to $-4p_1/(2-2\rho+\rho^2)\bar{v}+1+(1-\rho)^4\bar{v}^2/4(2-2\rho+\rho^2)(\bar{v}-2p_1)^2=0$. The second-period price personalizes based on the first-period purchase record $a_1$:
    \begin{equation}
    p_2(a_1)=
    \frac{\bar{v}((1-\rho)^2\bar{v}+\bar{v}-2p_1)}{4(\bar{v}-2p_1)}\boldsymbol{1}(a_1=1)+\frac{\bar{v}}{4}\boldsymbol{1}(a_1=0),
    \end{equation}
\end{itemize}
\end{lemma}

\begin{proof}
When observing buyers' high common interaction frequency $\hat{x}=1$, based on the posterior belief in (\ref{continuous_belief1}), the seller updates about buyers' product preferences. Here, for convenience, let $s$ denote the conditional probability of the event that buyer $i$ and $j$ make the common interaction frequency $\hat{x}\triangleq\min\{x_{ij},x_{ji}\}=1$ given that they are in the same high-end class $v_j\in[\bar{v}/2,\bar{v}]$ and $v_i\in[\bar{v}/2,\bar{v}]$. i.e.,
\begin{align}
    s\triangleq\textrm{Pr}\left(\hat{x}=1|v_i\in[\bar{v}/2,\bar{v}],v_j\in[\bar{v}/2,\bar{v}]\right)=(1-\rho)^2.
\end{align}
Then, we have for the seller:
\begin{equation}
f(v_1|\hat{x}=1)=\left\{
\begin{aligned}
&\frac{2s}{(s+1)\bar{v}},&\quad\textrm{if $v_1\in[\bar{v}/2,\bar{v}]$},\\
&\frac{2}{(s+1)\bar{v}},&\quad\textrm{if $v_1\in[0,\bar{v}/2]$}.
\end{aligned}
\right.
\end{equation}
It follows that:
\begin{equation}
{\rm Pr}(a_1=1|\hat{x}=1;p_1\ge\frac{\bar{v}}{2})=\frac{2s(\bar{v}-p_1)}{(s+1)\bar{v}},
\end{equation}
and
\begin{equation}
{\rm Pr}(a_1=1|\hat{x}=1;p_1\le\frac{\bar{v}}{2})=\frac{(s+1)\bar{v}-2p_1}{(s+1)\bar{v}},    
\end{equation}
Specifically, we first analyze the second selling period. (a) When setting the first-period price $p_1\ge\frac{\bar{v}}{2}$, the first-arriving buyer is in the high-end class if he purchases the product $a_1=1$. Then, the seller learns that the latter-arriving buyer is in the same high-end class, and optimally charges $p_2=\frac{\bar{v}}{2}$ to extract the maximal expected second-period revenue $\frac{\bar{v}}{2}$ from the high-end buyers. Yet if the first-arriving buyer does not purchase $a_1=0$, the seller updates for the latter-arriving buyer's product preference as below:
\begin{equation}
f(v_2|\hat{x}=1;p_1\ge\frac{\bar{v}}{2},a_1=0)=\left\{
\begin{aligned}
&\frac{2(2p_1s-s\bar{v})}{(\bar{v}+2p_1s-s\bar{v})\bar{v}},&\quad\textrm{if $v_2\in[\bar{v}/2,\bar{v}]$},\\
&\frac{2}{(\bar{v}+2p_1s-s\bar{v})\bar{v}},&\quad\textrm{if $v_2\in[0,\bar{v}/2]$}.
\end{aligned}
\right.
\end{equation}
Hence, we can show that the seller's optimal second-period price is $p_2=\frac{\bar{v}+2p_1s-s\bar{v}}{4}$. In this way, the seller's expected second-period revenue is $\frac{\bar{v}+2p_1s-s\bar{v}}{8}$. 

(b) When setting the first-period price $p_1\le\frac{\bar{v}}{2}$, the first-arriving buyer is in the low-end class if he does not purchase the product $a_1=0$. Then, the seller learns that the latter-arriving buyer is in the same low-end class, and optimally charges $p_2=\frac{\bar{v}}{4}$ to extract the maximal expected second-period revenue $\frac{\bar{v}}{8}$ from the high-end buyers. Yet if the first-arriving buyer purchases $a_1=1$, the seller updates for the latter-arriving buyer's product preference as below:
\begin{equation}
f(v_2|\hat{x}=1;p_1\le\frac{\bar{v}}{2},a_1=1)=\left\{
\begin{aligned}
&\frac{2s}{s\bar{v}+\bar{v}-2p_1},&\quad\textrm{if $v_2\in[\bar{v}/2,\bar{v}]$},\\
&\frac{2(\bar{v}-2p_1)}{(s\bar{v}+\bar{v}-2p_1)\bar{v}},&\quad\textrm{if $v_2\in[0,\bar{v}/2]$}.
\end{aligned}
\right.
\end{equation}
We then derive that the seller's optimal second-period price is $p_2=\frac{\bar{v}}{2}$ with corresponding expected second-period revenue $\frac{s\bar{v}^2}{2(s\bar{v}+\bar{v}-2p_1)}$ when $\frac{(1-s)\bar{v}}{2}<p_1<\frac{\bar{v}}{2}$. Otherwise when $p_1<\frac{(1-s)\bar{v}}{2}$, the seller charges $p_2=\frac{\bar{v}(s\bar{v}+\bar{v}-2p_1)}{4(\bar{v}-2p_1)}$ in the second period, to extract the maximal expected second-period revenue $\frac{\bar{v}(s\bar{v}+\bar{v}-2p_1)}{8(\bar{v}-2p_1)}$. 

Now we turn to the seller's first-period problem, which is to maximize the expected sale revenue over the two selling periods in Stage II. As shown above, the seller's learning and thus second-period pricing is significantly affected by the first-period pricing. Technically, the objective first-period function is piece-wise and not necessarily convex/concave. Involved discussions and calculations (where multi-order derivatives are majorly used) are then inevitable, and we ignore those technical details here.
\end{proof}

Then, Lemma \ref{continuous_price_lemma2} below turns to the case when the seller observes buyers' low common interaction frequency $\hat{x}=0$.

\begin{lemma}\label{continuous_price_lemma2}
Under the continuous-preference model, given the posterior belief in (\ref{continuous_belief0}), when observing buyers' low common interaction frequency $\hat{x}=0$, the seller's optimal pricing in Stage II is as follows. The seller charges $p_1=\frac{\bar{v}}{2}$ in the first selling period, and personalizes the second-period price based on the first-period purchase record $a_1$:
    \begin{equation}
    p_2(a_1)=
    \frac{(1+2\rho-\rho^2)\bar{v}}{4}\boldsymbol{1}(a_1=1)+\frac{\bar{v}}{2}\boldsymbol{1}(a_1=0),
    \end{equation}
\end{lemma}

The proof here is similar to the proof of Lemma \ref{continuous_price_lemma1}. In brief, we repeat those key steps here, proceeded by \textit{backward induction}: In Step 1, we first analyze the posterior distributions of buyers' product preferences when observing buyers' low common interaction frequency $\hat{x}=0$. Then in Step 2, we start from analyzing the second selling period, and we develop two major cases regarding the first-period price $p_1\ge\frac{\bar{v}}{2}$ or $p_1\le\frac{\bar{v}}{2}$. Specifically, in the second period, the seller again updates her belief on the latter-arriving buyer's product preference and charges accordingly. Finally, in Step 3, we move to the seller's first-period problem, to maximize the expected revenue over the two selling periods in Stage II. Particularly, this step tends to be involved in calculations as the objective function is piece-wise and may be non-convex within certain pieces.

\subsubsection{Analysis for Buyers' Social Interactions in Stage I}

Now, we continue to analyze the buyers' social behaviors in Stage I based on the seller's optimal pricing in Lemma \ref{continuous_price_lemma1} and Lemma  \ref{continuous_price_lemma2}. Specifically, we first calculate buyers' purchase surplus and accordingly update the expected final payoff with their social interaction utility in Table \ref{table_same} and Table \ref{table_differ}. Next, through normal-form game-theoretic analysis, we derive buyers' equilibrium social interactions in Stage I based on the derived expected final payoff table. Finally, we examine the conditions for which the derived equilibrium could be sustained as a PBE, ensuring the consistency between buyers' social interaction decisions in Stage I and the seller's belief.

First consider when $(1-\rho)^2\ge3/4$, high-end buyers anticipate the seller's pricing as in Case 1 of Lemma \ref{continuous_price_lemma1} upon $\hat{x}=1$ and pricing as Lemma \ref{continuous_price_lemma2} upon $\hat{x}=0$. We now update the expected final payoff matrix for high-end buyer $i$ as below.

\begin{table}[h]
\caption{Expected Final Payoff Matrix for High-End Buyer $i$ when $(1-\rho)^2\ge3/4$}
\begin{center}
\footnotesize
\begin{tikzpicture}[thick]
\draw (-3,0) rectangle (5,1);
\draw (-3,0.5) -- (5,0.5);
\draw (1,0) -- (1,1);
\foreach \x/\xtext in {-1/{$1+v_i-\frac{\bar{v}}{2}$},3/{$1-l+v_i-\frac{\bar{v}}{2}+\frac{s\bar{v}}{8}$}}
\draw (\x,0.75) node {\xtext};
\foreach \x/\xtext in {-1/{$l+v_i-\frac{\bar{v}}{2}+\frac{s\bar{v}}{8}$},3/{$v_i-\frac{\bar{v}}{2}+\frac{s\bar{v}}{8}$}}
\draw (\x,0.25) node {\xtext};
\node at (-3,0.75) [left] {$x_{ij}=1$};
\node at (-3,0.25) [left] {$x_{ij}=0$};
\node at (-1,1) [above] {$x_{ji}=1$};
\node at (3,1) [above] {$x_{ji}=0$};
\end{tikzpicture}
\end{center}
\end{table}

Then, if $\bar{v}\le8(1-l)$, one can show that high-end buyers' equilibrium decisions are $x_{ij}^*=x_{ji}^*=1$ without manipulation, i.e. $(1-\rho)^2=1>3/4$. Yet if $\bar{v}>8(1-l)$ one can derive that high-end buyers' equilibrium social decision is to manipulate with probability 
\begin{equation}
    \rho=1-2\sqrt[3]{\frac{1-l}{\bar{v}}}.
\end{equation}
To sustain $3/4\le (1-\rho)^2<1$ in this case, we further need $\bar{v}\le 64(1-l)/3$ to hold. At the end, examining the expected final payoff of buyers both in the low-end class or in different classes shows the strategy profile above can be sustained.

Then consider when $1/4\le(1-\rho)^2<3/4$, high-end buyers anticipate the seller's pricing as in Case 2 of Lemma \ref{continuous_price_lemma1} upon $\hat{x}=1$ and pricing as Lemma \ref{continuous_price_lemma2} upon $\hat{x}=0$. We now update the expected final payoff matrix for high-end buyer $i$ as below.
\begin{table}[h]
\caption{Expected Final Payoff Matrix for High-End Buyer $i$ when $1/4\le(1-\rho)^2<3/4$}
\begin{center}
\footnotesize
\begin{tikzpicture}[thick]
\draw (-3,0) rectangle (5,1);
\draw (-3,0.5) -- (5,0.5);
\draw (1,0) -- (1,1);
\foreach \x/\xtext in {-1/{$1+v_i-\frac{1}{2}\left(\frac{4s+5}{16}\bar{v}+\frac{1}{2}\bar{v}\right)$},3/{$1-l+v_i-\frac{\bar{v}}{2}+\frac{s\bar{v}}{8}$}}
\draw (\x,0.75) node {\xtext};
\foreach \x/\xtext in {-1/{$l+v_i-\frac{\bar{v}}{2}+\frac{s\bar{v}}{8}$},3/{$v_i-\frac{\bar{v}}{2}+\frac{s\bar{v}}{8}$}}
\draw (\x,0.25) node {\xtext};
\node at (-3,0.75) [left] {$x_{ij}=1$};
\node at (-3,0.25) [left] {$x_{ij}=0$};
\node at (-1,1) [above] {$x_{ji}=1$};
\node at (3,1) [above] {$x_{ji}=0$};
\end{tikzpicture}
\end{center}
\end{table}

It turns out that high-end buyers' equilibrium social decision is to manipulate with probability $\rho$ satisfying
\begin{equation}\label{rho_equation}
8\bar{v}{(1-\rho)}^3-3\bar{v}(1-\rho)+32l-32=0.
\end{equation}
Specifically, there exists a solution to (\ref{rho_equation}) if and only if $1-l<3\sqrt{3}\bar{v}/64$. Particularly, such a solution is unique and satisfies $1-\sqrt{3}/2<\rho<1-\sqrt{3/8}$. At the end, examining the expected final payoff of buyers both in the low-end class or in different classes shows the strategy profile above can be sustained.

Finally when $(1-\rho)^2<1/4$, high-preference buyers anticipate the seller's pricing as in Case 3 of Lemma \ref{continuous_price_lemma1} at $\hat{x}=1$ with $p_1<(2\rho-\rho^2)/2\bar{v}$, and the seller's pricing as Lemma \ref{continuous_price_lemma2} at $\hat{x}=0$. Similarly, one can show that no equilibrium can be sustained with $(1-\rho)^2<1/4$.

\section{Extension to Non-uniform Prior Distributions: Supplementary Analysis}\label{Appendix:Extension1:non-uniform}

In this appendix, we present Table~\ref{Tab:boundary_function}, which details the boundary functions that delineate the different equilibrium regions depicted in Fig. \ref{fig:newPBE_nonuniform}. This figure provides a graphical illustration of the new PBE under a non-uniform prior distribution for an arbitrary $\alpha\le 0.5$, as analyzed in Section
\ref{Extension1:non-uniform}.

\begin{table}[h]
\caption{Boundary functions in Fig. \ref{fig:newPBE_nonuniform} with $\alpha\le 0.5$}
\scriptsize
\center
\begin{tabular}{cp{8cm}<{\centering}}
\specialrule{\heavyrulewidth}{1.2pt}{0pt}
\cellcolor[gray]{0.92}{\makecell[c]{Curve}} &{\makecell[c]{Corresponding boundary functions in $(v_H,v_L)$-plane}}\\
\specialrule{\lightrulewidth}{0pt}{1.2pt}
OA &\cellcolor[gray]{0.92}{$v_L=v_H$}\\
\cellcolor[gray]{0.92}{OD} & $v_L=2v_H/3$\\
OCD* & \cellcolor[gray]{0.92}{$v_L=2\alpha^2v_H/(3\alpha^2-2\alpha+1)$}\\
\cellcolor[gray]{0.92}{BC} &$v_H-v_L=2(1-l)$\\
CP & \cellcolor[gray]{0.92}{$(v_H-v_L)v_L=8\alpha^2(1-l)^2/(1-\alpha)^2$}\\
\cellcolor[gray]{0.92}{BP} & $(v_H-v_L)(v_H-v_L/\alpha)=4(1-l)^2$\\
\specialrule{\heavyrulewidth}{1.2pt}{0pt}
\end{tabular}
\label{Tab:boundary_function}
\end{table}

\end{document}